\documentclass{lmcs} %%% last changed 2014-08-20
\pdfoutput=1

\usepackage{lastpage}
\lmcsdoi{20}{4}{9}
\lmcsheading{}{\pageref{LastPage}}{}{}%
{Jan.~03,~2023}{Oct.~29,~2024}{}

\usepackage[utf8]{inputenc}


\usepackage{hyperref}
\theoremstyle{plain} %\crefname{satz}{Satz}{S\"atze}
\def\eg{{\em e.g.}}
\def\cf{{\em cf.}}
\def\ie{{\em i.e.}}

\usepackage{amsthm,amssymb,latexsym,stmaryrd,eufrak}
\usepackage{xifthen}% provides \isempty test
\usepackage{comment}
\usepackage{tikz}
\usetikzlibrary{arrows,shapes,automata,cd}
\usetikzlibrary{calc} %% coordinate calculation
  \usetikzlibrary{matrix,arrows} %% commutative diagrams with tikz
  \tikzset{
    commutative diagram/.style 2 args={
    	matrix of math nodes, row sep=#1,column sep=#2,
	text height=1.5ex, text depth=0.25ex},
    commutative diagram/.default={1cm}{1cm}
    }
  \tikzset{    
    skip loop/.style n args={3}{to path={-- ++(0,#1) -| node[pos=0.25,#2] {#3} (\tikztotarget)}},
    cross line/.style={preaction={draw=white, -, line width=6pt}}
  }

\providecommand*{\ifempty}[3]{\ifthenelse{\isempty{#1}}{#2}{#3}}
\newcommand{\parensmathoper}[2]{\ensuremath{{{#1}}\ifempty{#2}{}{(#2)}}}

\newcommand{\naturals}{\mathbb{N}} % 
\newcommand{\prationals}{\mathbb{Q}_{\geq 0}}
\newcommand{\preals}{\mathbb{R}_{\geq 0}}
\newcommand{\reals}{\mathbb{R}}

\newcommand{\tupl}[1]{\langle #1 \rangle} % tuples
\newcommand{\supp}{\parensmathoper{supp}} % support
\newcommand{\e}{\varepsilon}  % epsilon
\newcommand{\iso}{\cong} % categorical isomorphism
\newcommand{\nat}{\Rightarrow} % natural transformation
\newcommand{\inj}{\mathit{in}} % coproduct injections

\newcommand{\set}[2]{\left\{ #1 \ifempty{#2}{}{\mid #2} \right\}} % set
\newcommand{\K}[1][d]{\parensmathoper{\parensmathoper{\mathcal{K}}{#1}}} % Kantorovich distance
\newcommand{\C}{\mathcal{C}}

\newcommand{\FMP}[1][]{\textit{LMP}_{#1}} % Finite Markov processes
\newcommand{\MP}[1][X]{\textit{MP}_{#1}} % Markov processes
\newcommand{\MM}[1][X]{\textit{MM}_{#1}} % Mealy machines
\newcommand{\cat}{\mathbf}
\newcommand{\Set}{\mathbf{Set}} %% category of sets
\newcommand{\Met}{\mathbf{Met}} %% category of metric spaces
\newcommand{\CMet}{\mathbf{CMet}} %% category of complete metric spaces
\newcommand{\PMet}{\mathbf{PMet}} %% category of pseudometric spaces

\newcommand{\CC}{\mathbb{C}}
\renewcommand{\O}[1]{\mathcal{O}_{#1}}

\newcommand{\QA}[1][\Sigma]{\mathbf{QA}(#1)}
\newcommand{\CQA}[1][\Sigma]{\mathbf{CQA}(#1)}

\newcommand{\lebint}[3][{}]{\int_{#1} \, #2 \; {\mathrm{d}}#3}
\newcommand{\mprod}{\mathbin{\square}}
\newcommand{\ev}{\mathit{ev}}
\newcommand{\coev}{ \overline{\mathit{ev}} }

\newcommand{\Op}[1][T]{\mathcal{O}_{#1}}

\newcommand{\TT}[2][]{\mathbb{T}\ifempty{#2}{\ifempty{#1}{}{(#1)}}{\ifempty{#1}{(#2)}{(#1,#2)}}}  %% Set of terms
\newcommand{\Sub}[1][]{\ifempty{#1}{\mathcal{S}}{\mathcal{S}(#1)}} %% Set of substitutions
\newcommand{\V}[1][]{\ifempty{#1}{\mathcal{V}}{\mathcal{V}(#1)}} %% Set of quantitative equations
\newcommand{\E}[1][]{\ifempty{#1}{\mathcal{E}}{\mathcal{E}(#1)}} %% Set of conditional quantitative equations
\newcommand{\U}[1][]{{\mathcal{U}^{#1}}} %% Theory induced by a quantitative deduction system
\newcommand{\KK}[2][]{\mathbb{K}\ifempty{#1}{(#2)}{(#1,#2)}} %% Class of models for theory
\newcommand{\Refl}{\textsf{Refl}} %% Reflexivity Axiom
\newcommand{\Symm}{\textsf{Symm}} %% Symmetry Axiom
\newcommand{\Triang}{\textsf{Triang}} %% Triangular inequality Axiom
\newcommand{\Weak}{\textsf{Weak}} %% Weak Axiom
\newcommand{\Cont}{\textsf{Cont}} %% Continuity Axiom
\newcommand{\Nexp}[1]{#1\text{-}\textsf{NE}} %% Non-expansivity Axiom
\newcommand{\Subst}{\textsf{Subst}} %% Substitution Axiom
\newcommand{\Cut}{\textsf{Cut}} %% Cut rule
\newcommand{\Assum}{\textsf{Ass}} %% Assumption rule

\newcommand{\Lip}[1]{#1\text{-}\textsf{Lip}} %% Lipschitz Axiom
\newcommand{\IB}{\textsf{IB}}
\newcommand{\Bone}{\textsf{B1}}
\newcommand{\Btwo}{\textsf{B2}}
\newcommand{\SC}{\textsf{SC}}
\newcommand{\SA}{\textsf{SA}}

\newcommand{\Szero}{\textsf{S0}}
\newcommand{\Sone}{\textsf{S1}}
\newcommand{\Stwo}{\textsf{S2}}
\newcommand{\Sthree}{\textsf{S3}}
\newcommand{\Sfour}{\textsf{S4}}

\newcommand{\Idem}{\textsf{Idem}}
\newcommand{\Diag}{\textsf{Diag}}

\newcommand{\Zero}{\textsf{Zero}}
\newcommand{\Mult}{\textsf{Mult}}
\newcommand{\Diff}{\textsf{Diff}}

\newcommand{\Com}[1][]{\textsf{Com#1}}

\newcommand{\A}{\mathcal{A}} %% Generic quantitative algebra
\newcommand{\Alg}[1]{#1\text{-}\mathbf{Alg}}
\newcommand{\EMAlg}[1]{\mathbf{EM}(#1)}
\newcommand{\biAlg}[2]{\tupl{#1, #2}\text{-}\mathbf{biAlg}}
\newcommand{\biEMAlg}[2]{\mathbf{EM}\tupl{#1, #2}}
\newcommand{\distAlg}[1]{#1\text{-}\mathbf{biAlg}}
\newcommand{\tensorAlg}[2]{\mathbf{EM}_\mathbf{t}\tupl{#1, #2}}

\newcommand{\B}{\mathcal{B}}  %% Barycentric algebras 
\newcommand{\Semi}{\mathcal{S}}  %% Semilattices with 0  
\newcommand{\R}[1][E]{\mathcal{R}_{#1}}  %% Reader algebras
\newcommand{\rd}{\textsf{r}} %% reader operation
\newcommand{\Wr}[1][W]{\mathcal{W}_{#1}}  %% Writer algebras
\newcommand{\wrt}[1]{\textsf{w}_{#1}} %% reader operation

\newcommand{\D}{\mathcal{D}} %% probability over open states

\newcommand{\zero}{\mathbf{0}}
\newcommand{\dist}[1][c]{\mathbf{d}^{#1}}

\newcommand{\ens}[1]{\{ #1\}}
\begin{document}

\title{Sum and Tensor of Quantitative Effects}
\titlecomment{{\lsuper*}Extended and combined version of~\cite{BacciMPP18} (LICS'18) and \cite{BacciMPP21} (CALCO'21).}

\author[G.~Bacci]{Giorgio Bacci\lmcsorcid{0000-0003-4004-6049}}[a]	%required
\address{Department of Computer Science, Aalborg University, Aalborg, Denmark}	%required
\email{grbacci@cs.aau.dk}  %optional
%\thanks{thanks 1, optional.}	%optional

\author[R.~Mardare]{Radu Mardare\lmcsorcid{0000-0001-8660-1832}}[b]	%optional
\address{Department of Computer \& Information Sciences, University of Strathclyde, Glasgow, Scotland}	%optional
\email{r.mardare@strath.ac.uk}  %optional
%\thanks{thanks 2, optional.}	%optional

\author[P.~Panangaden]{Prakash Panangaden}[c]	%optional
\address{School of Computer Science, McGill University, Montreal, Canada}	%optional
\email{prakash.panangaden@mcgill.ca}
%\urladdr{name3@url3\quad\rm{(optionally, a web-page can be specified)}}  %optional
%\thanks{thanks 3, optional.}	%optional

\author[G.~Plotkin]{Gordon Plotkin}[d]	%optional
\address{LFCS, School of Informatics, University of Edinburgh, Edinburgh, Scotland}	%optional
\email{gdp@inf.ed.ac.uk}  %optional
%\thanks{thanks 2, optional.}	%optional

%% etc.

%% required for running head on odd and even pages, use suitable
%% abbreviations in case of long titles and many authors:

%%%%%%%%%%%%%%%%%%%%%%%%%%%%%%%%%%%%%%%%%%%%%%%%%%%%%%%%%%%%%%%%%%%%%%%%%%%

%% the abstract has to PRECEDE the command \maketitle:
%% be sure not to issue the \maketitle command twice!

\begin{abstract}
  \noindent Inspired by the seminal work of Hyland, Plotkin, and Power on the
  combination of algebraic computational effects via \emph{sum} and \emph{tensor},
  we develop an analogous theory for the combination of quantitative algebraic effects. 
  
  \emph{Quantitative algebraic effects} are monadic computational effects on categories of
  \emph{metric spaces}, which, moreover, have an algebraic presentation in the form of
  quantitative equational theories, a logical framework introduced by Mardare, Panangaden,
  and Plotkin that generalises equational logic to account for a concept of approximate
  equality.  As our main result, we show that the sum and tensor of two quantitative
  equational theories correspond to the categorical sum (\ie, coproduct) and tensor,
  respectively, of their effects qua monads. We further give a theory of
  \emph{quantitative effect transformers} based on these two operations, essentially
  providing quantitative analogues to the following monad transformers due to Moggi:
  exception, resumption, reader, and writer transformers.  Finally, as an application, we
  provide the first quantitative algebraic axiomatizations to the following coalgebraic
  structures: Markov processes, labelled Markov processes, Mealy machines, and Markov
  decision processes, each endowed with their respective bisimilarity metrics.  Apart from
  the intrinsic interest in these axiomatizations, it is pleasing they have been obtained
  as the composition, via sum and tensor, of simpler quantitative equational theories.
\end{abstract}

\maketitle

%% start the paper here:
\section{Introduction}\label{sec:intro}

The theory of computational effects began with the work of
Moggi~\cite{Moggi88,Moggi91} seeking a unified category-theoretic account of the
semantics of higher-order programming languages.  He modelled computational
effects (which he called notions of computation) employing strong monads on a
base category with a Cartesian closed structure.  With
Cenciarelli~\cite{Cenciarelli93}, he later extended the theory by allowing a
compositional treatment of various semantic phenomena such as state, IO,
exceptions, resumptions, etc, via the use of monad transformers.  This work was
followed up by the program of Plotkin and Power~\cite{Plotkin01,Plotkin02} on an
axiomatic understanding of computational effects as arising from operations and
equations via the use of Lawvere theories (see also~\cite{Hyland07}).  In a
fundamental contribution~\cite{HylandPP06}, jointly with Hyland, they developed a
unified modular theory for algebraic effects that supports their combination by
taking the \emph{sum} and \emph{tensor} of their Lawvere theories.  This allowed
them to recover in a pleasing structural algebraic way many of the monad
transformers considered by Moggi.

Quantitative equational logic, introduced by Mardare, Panangaden, and 
Plotkin~\cite{MardarePP:LICS16}, is a logical framework generalising standard equational logic to account for a concept of approximate equality.
The key idea is to introduce equations indexed by rational numbers
\begin{equation*}
t =_\e s
\end{equation*}
where $t, s$ are terms over a signature of operations.  One reads this as ``$s$ is within
$\e$ of $t$''.  The model theory of quantitative equational logic is developed into
quantitative universal algebras, that is, universal algebras with operations interpreted
as non-expansive maps on a metric space.  Quantitative equational logic is a logical
framework providing quantitative analogues of the core results of equational logic, such
as completeness theorems, constructions of free algebras, Cauchy completions of models,
and Birkoff-like (quasi-)variety theorems~\cite{MardarePP:LICS16, MardarePP17, FordMS21}.
Moreover and relevantly for the present paper, they are used to provide an algebraic
presentation of quantitative effects as freely generated monads on categories of metric
spaces.  As we will show in Section~\ref{sec:examples}, quantitative theories are
expressive enough to recover many quantitative effects of interest in computer science,
such as exceptions, interactive input/output, read, write, non-determinism, and
probabilistic choice.

Following Hyland et al.~\cite{HylandPP06}, in this paper
we develop the theory for the \emph{sum} and \emph{tensor of quantitative
equational theories}.

The sum combines two theories by taking their disjoint union.  In this sense, it is the
simplest combination supporting both given effects.  In contrast, the tensor additionally
imposes mutual commutation of the operations from each theory.  As such it refines the sum
of theories, which is just their unrestricted combination.  The sum and tensor of theories
arise in several contexts.  For example, in the semantics of programming languages, the
monad transformers for exception and resumption are given by a sum; and the transformers for
global state, reader, and writer are given by a tensor~\cite{HylandPP06}.

The main contributions of the present paper are:
\begin{enumerate}
  \item \label{i} 
  we prove that the sum and tensor of quantitative equational theories correspond 
  to the categorical sum (\ie, coproduct) and tensor, respectively, of their 
  induced quantitative effects as strong monads;
  \item \label{ii}
  we provide a quantitative presentation to the \emph{quantitative exception}
  and \emph{interactive input monads}, and obtain quantitative analogues to
  their corresponding Moggi's monad transformers at the level of theories using sum;
  \item \label{iii} 
  we give quantitative axiomatizations to the \emph{quantitative reader} 
  and \emph{writer monads}, from which we obtain analogues of their monad 
  transformers at the level of theories using tensor;
  \item \label{iv} 
  we provide the first axiomatizations of \emph{Markov processes}, 
  \emph{labelled Markov processes}, \emph{Mealy machines}, and 
  \emph{Markov decision processes with rewards}, each endowed with their respective 
  (discounted) bisimilarity metrics. 
\end{enumerate}

For the results in \eqref{i}, we require the quantitative theories to be axiomatized by a
set of quantitative inferences involving only quantitative equations between variables in
the premises.  As in~\cite{MardarePP17}, we call this type of theories \emph{basic}.  The equational
monad transformers mentioned in \eqref{ii} and \eqref{iii} are compelling evidence
for the usefulness of our compositional framework. Ideally, these
transformers could be implemented in a future quantitative extension of effectful programming
languages, such as \textsf{Eff}, \textsf{Koka}, or \textsf{Haskell}.

The axiomatizations listed in \eqref{iv} are major examples of our compositional theory of
quantitative effects.  On the one hand, we obtain the bisimilarity metrics for
Markov processes by starting from the theory of interpolative barycentric algebras (used
to axiomatize probability distributions with the Kantorovich metric) and by applying to it, in
turn, the exception and interactive-input theory transformers, which are two examples of
sum of theories.  On the other hand, labelled Markov processes and Markov decision
processes with rewards are obtained by complementing the axiomatization for Markov
processes with the missing computational effects.  We add the computational effect of reacting
to an action label by tensoring the basic theory of Markov
processes with that of quantitative reading computations
(corresponding to the reader transformer); while the computational effect of accumulating rewards is obtained by tensoring with the theory of quantitative writing computations (corresponding
to the writer transformer).  We illustrate our approach by decomposing the proposed
axiomatizations into their basic components and showing how to combine them step-by-step
to get the desired result.  The axiomatization of Mealy machines is obtained similarly and provides 
further evidence for the generality and simplicity of our compositional approach to
quantitative effects.

This article is an extended and combined version of~\cite{BacciMPP18} and
\cite{BacciMPP21}. Beyond providing all proofs which could not be published
in~\cite{BacciMPP18,BacciMPP21} because of space limitations, we refactored and simplified
several technical results.  The main examples of this refactorization are
Sections~\ref{sec:examples}, \ref{sec:sumWithException}, and \ref{sec:sumInput}.  In the
latter, we improve upon some of the results originally presented in~\cite{BacciMPP18} (\cf\ Corollaries~\ref{cor:ContractiveTransfomerMet}--\ref{cor:ResumptionTransfomerCMet}) by observing that quantitative theories induce only
monads with (at most) countable rank, a result due to Ford et al.~\cite{FordMS21} that we
did not know when writing~\cite{BacciMPP18}.  Moreover, the axiomatization of Mealy
machines (Section~\ref{sec:mealyMachines}) is new material not present in the conference
version of~\cite{BacciMPP21}.

\subsection*{Further Related Work}

In~\cite{HylandPP06,HylandLPP07} the sum and tensor of (enriched) Lawvere theories are
characterized as the colimit of certain cocones, and the correspondence with the sum and
tensor of monads is obtained via the equivalence between Lawvere theories and monads.
Since it is not hard to show that (basic) quantitative equational theories can be
characterized as metric-enriched Lawvere theories, one may think to recover the
correspondence with the operations on their monads via the equivalence with Lawvere
theories. Alas, quantitative equational theories and Lawvere theories are \emph{not} equivalent,
as the latter is more expressive than the framework
of Mardare et al.~\cite{MardarePP:LICS16} (metric-enriched Lawere theories allow generic operations with metric spaces as arities, while quantitative equational logic admits only operations with discrete arities).
An equivalence with \emph{discrete} Lawvere theories~\cite{HylandP06} (where arities are
just countable ordinals) does not hold either, because quantitative equations implicitly
generate morphisms (hence, operations in a Lawvere sense) with non-discrete arities which cannot be expressed in the framework of discrete Lawvere theories.

The above limitations required us to follow a different path which required us to prove the
two correspondences directly. For the correspondence with the sum of monads, we could
follow Kelly~\cite{Kelly1980}, which characterizes the Eilenberg-Moore algebras of the
coproduct of monads as bialgebras. 
However, characterizing the tensor bialgebras for the monads, which correspond to the Eilenberg-Moore algebras for their tensor, was more complex. This complexity led us to introduce the concept of \emph{pre-operations of a strong functor}. Pre-operations represent a natural extension of Manes' notion of operation of a monad~\cite{Manes69} and Plotkin and Power's notion of algebraic operation~\cite{PlotkinP01,PlotkinP03}. We chose to consider pre-operations over functors, not just on monads, to establish a connection between the operations of an algebraic monad and those of its signature functors. This approach allowed us to characterize the tensor bialgebras for the monads in terms of the tensor bialgebras for their associated signature functors, eliminating the need for a correspondence with a specific subclass of metric-enriched Lawvere theories.

Finally, we remark that quantitative equational theories, although not as general as
metric-enriched Lavwere theories, are a natural and simpler form of enriched equational
theory, which is still expressive enough to recover many examples of interest in computer
science (see~\cite{MardarePP:LICS16,BacciMPP18,MioV20}).  In this respect, it is pleasing
that also this simpler subclass of enriched theories is closed under sum and tensor.

\subsection*{Synopsis} We start by recalling some preliminary categorical definitions that
will be used in the rest of the paper (Section~\ref{sec:prelim}).  In
Section~\ref{sec:QDS}, we introduce the core definitions and results of the theory of
quantitative algebras. In Section~\ref{sec:examples}, we present several examples of
algebraic quantitative effects and present their axiomatic quantitative equational
theories. In Sections~\ref{sec:sum} and \ref{sec:tensor}, we develop the theory for the
sum and tensor of quantitative equational theories and show that such combinators
correspond to the categorical sum and tensor of quantitative effects as monads,
respectively. In each of these sections, we propose several nontrivial examples of
composition of quantitative effects.  Finally, in Section~\ref{sec:concl} we collect some conclusions and propose
possible future work.

In the Appendices~\ref{app:exmetric} and \ref{sec:EMet-lcp} we recall some technical results regarding the categories of metric spaces that we relevant to the result presented in this paper.

%%%%%%%%%%%%%%%%%%%%%%%%%%%%%%%%%%%%%%%%%%%%%%
\section{Preliminaries and Notation}
\label{sec:prelim}

In this paper, we deal with Eilenberg-Moore algebras of strong monads on the category of extended metric spaces. We assume familiarity with the basic notions of category theory, such as functors, natural transformations, and adjunctions (see~\cite{MacLane} for reference). 

In this section, for the sake of fixing notation, we recall some basic definitions regarding metric spaces, monads, and monoidal closed categories. As these definitions are standard, a reader who is familiar with these concepts can safely skip this section. 

\subsection{Categories of Extended Metric Spaces} An \emph{extended metric space} is a pair $(X,d_X)$ consisting of a set $X$ equipped with a distance function $d_X \colon X \times X \to [0,\infty]$ satisfying: (i) $d(x,y) =0$ iff $x=y$, (ii) $d_X(x,y) = d_X(y,x)$ and (iii) $d_X(x,z) \leq d_X(x,y) + d_X(y,z)$. Note that the distance function is allowed to have infinite values, so the sum of positive real numbers is extended to $[0, \infty]$ by canonically imposing that $\infty + r = r + \infty = \infty$, for all $r \in [0, \infty]$ (hence, $\infty$ is the top element w.r.t.\ the extension of the order $\leq$).

A sequence $(x_i)$ in $(X,d_X)$ \emph{converges} to $x \in X$ if $\forall \epsilon > 0,\exists N, \forall i \geq N, d_X(x_i,x) \leq \epsilon$.
A sequence $(x_i)$ is \emph{Cauchy} if
$\forall \epsilon > 0,\exists N, \forall i,j \geq N, d_X(x_i,x_j) \leq
\epsilon$.
If every Cauchy sequence converges, the extended metric space $(X,d_X)$ is said to be
\emph{complete}.  If a space is not complete it can be completed by a
well-known construction called \emph{Cauchy completion}.  We write
$\overline{(X,d_X)}$ for the completion of $(X,d_X)$. 

Let $(X,d_X)$, $(Y,d_Y)$ be extended metric spaces. A map $f \colon X \to Y$ is \emph{$c$-Lipschitz continuous}, with constant $c \geq 0$, if for all $x,x' \in X$, $c
\cdot d_X(x,x') \geq  d_Y(f(x),f(x'))$. 
If $c = 1$, the function is called \emph{non-expansive}, and if $0 \leq c < 1$
and $f$ maps $X$ to itself, it is called a \emph{contraction}.  Observe that Lipschitz continuous functions preserve convergence since they are continuous in the usual sense.

When the distance function is clear from the context, we will refer to the extended metric space $(X, d_X)$ simply as $X$. Throughout the rest of the paper, to simplify notation, we will adopt the convention of subscripting the distance function with the name of the space, i.e., $d_X$ for the space $X$.

\smallskip

The categories of metric spaces that we consider are $\Met$, with extended metric spaces
as objects and non-expansive maps as morphism, and its full subcategory $\CMet$ of complete extended metric spaces. These categories are complete and cocomplete,
\ie, have all limits and colimits (see \autoref{app:exmetric} for details). Moreover, $\CMet$ is a reflective subcategory of $\Met$, with reflection given by the Cauchy
completion functor $\CC \colon \Met \to \CMet$, mapping a metric space to its completion,
being the left adjoint to the embedding $\CMet \hookrightarrow \Met$.  

\subsection{Monads and their Algebras}
\label{sec:monadsPrelim}
A monad on a category $\cat{C}$ is a triple $(T, \eta, \mu)$ consisting of an endofunctor
$T \colon \cat{C} \to \cat{C}$ and two natural transformations: a \emph{unit} $\eta \colon
\textit{Id} \nat T$ and a \emph{multiplication} 
$\mu \colon TT \nat T$ that satisfy the laws
\begin{equation*}
\begin{tikzcd}%[column sep=large]
TX \arrow[r, "\eta T" ] \arrow[dr, "id"'] &
TTX \arrow[d, "\mu" ] & 
TX \arrow[l, "T \eta"'] \arrow[dl, "id"] \\
& TX
\end{tikzcd}
\qquad\qquad
\begin{tikzcd}
TTTX \arrow[r, "\mu T" ] \arrow[d, "T \mu"'] & 
TTX \arrow[d, "\mu"] \\
TTX \arrow[r, "\mu"'] &
TX 
\end{tikzcd}
\end{equation*}
respectively called the \emph{left/right unit laws} and \emph{multiplication law} for the monad $(T,\eta,\mu)$.
When the monad structure is clear from the context we will denote $(T,\eta,\mu)$ simply as $T$.

\smallskip
Given an endofunctor $H \colon \cat{C} \to \cat{C}$, the \emph{free monad on $H$} is a monad $H^*$ on $\mathbf{C}$ equipped with a natural transformation $\gamma \colon H \nat H^*$ that is initial among all such
pairs $(S, \lambda \colon H \nat S)$. 

\smallskip
A \emph{monad map} from a monad $(T, \eta, \mu)$ on to a monad $(H, \rho, \nu)$ on the same category is
a natural transformation $\sigma \colon T \nat H$ that makes the following diagrams commute,
\begin{equation*}
\begin{tikzcd}%[column sep=large]
X \arrow[r, "\eta" ] \arrow[dr, "\rho"'] &
TX \arrow[d, "\sigma" ] \\
& HX
\end{tikzcd}
\qquad\qquad
\begin{tikzcd}
TTX \arrow[r, "\sigma T" ] \arrow[d, "\mu"'] & 
THX \arrow[r, "\sigma H" ] &
HHX \arrow[d, "\nu"] \\
TX \arrow[rr, "\sigma"'] & &
HX 
\end{tikzcd}
\end{equation*}
If $\sigma \colon T \nat H$ is an epimorphism, then $H$ is a \emph{quotient} of $T$. If it is a monomorphism, then $T$ is a \emph{submonad} of $T$. If it is an isomorphism, the two monads are isomorphic. In the following, we consider monads to be the same up to isomorphism.

\smallskip

Let $F \colon \cat{C} \to \cat{C}$ be an endofunctor. An algebra of $F$ (or simply, $F$-algebra) is a pair $(A, a)$ consisting of an object $A$, called \emph{carrier}, and a morphism $a \colon FA \to A$ in $\cat{C}$, called \emph{$F$-algebra structure}. A morphism of $F$-algebras (or simply, $F$-homomorphism) from $(A,a)$ to $(B,b)$ is an arrow $f \colon A \to B$ in $\cat{C}$ making the square below commute
\begin{equation*}
\begin{tikzcd}
FA \arrow[r, "a" ] \arrow[d, "F f"'] & 
A \arrow[d, "f"] \\
FB \arrow[r, "b"'] &
B
\end{tikzcd}
\end{equation*}
The algebras of a functor $F$ and their homomorphisms form a category, denoted $\Alg{F}$. The category of $F$-algebras has an obvious forgetful functor $U^F \colon \Alg{F} \to \cat{C}$ mapping an $F$-algebra $(A,a)$ to its carrier $A$, hence forgetting the algebra structure. If the forgetful functor has a left adjoint $L^F \colon \cat{C} \to \Alg{F}$, then the algebra $L^F(A)$ obtained from an object $A \in \cat{C}$ is called the \emph{free $F$-algebra for $A$}. The monad $U^FL^F$ resulting from this adjunction is called \emph{algebraic monad} and corresponds to the free monad $F^*$. Observe that free monads are not necessarily algebraic, however, this holds when the category $\cat{C}$ has products~\cite{Barr1970}.

\smallskip
An \emph{Eilenberg-Moore} (EM) algebra for a monad $(T, \eta, \mu)$, is a $T$-algebra $(A, a)$ making the
two diagrams below commute
\begin{equation*}
\begin{tikzcd}%[column sep=large]
A \arrow[r, "\eta" ] \arrow[dr, "id"'] &
TA \arrow[d, "a" ] & 
TTA \arrow[l, "\mu"' ] \arrow[d, "Ta"] \\
& A &
TA \arrow[l, "a" ]
\end{tikzcd}
\end{equation*}
respectively called \emph{unit law} (left diagram) and \emph{multiplication law} (right diagram) for the $T$-algebra
$(A,a)$. The morphisms between EM algebras are the $T$-homomorphism of their $T$-algebras. 
The resulting category of EM algebras for the monad $T$ is called the \emph{Eilenberg-Moore category} for the monad $T$, and it is denoted by $\EMAlg{T}$.

The forgetful functor $U^T \colon \EMAlg{T} \to \cat{C}$ has a left adjoint $F^T \colon \cat{C} \to \EMAlg{T}$ associating the \emph{free} EM algebra $(TX, \mu_X)$ with the object $X \in \cat{C}$. By construction, the monad $U^TF^T$ arising from the adjunction is isomorphic to $T$. Moreover, $\EMAlg{T}$ has all limits which exist in $\cat{C}$, and they are created by the forgetful functor. The situation for colimits is more complicated, as colimits may not necessarily exist.

\subsection{Monoidal Closed Categories and Strong Functors}
A category is \emph{monoidal} when it comes equipped with a ``product'' structure.
In detail, a \emph{monoidal category} is 
%a structure $\langle \cat{V}, \mprod, I, \alpha, \lambda, \rho \rangle$ consisting of 
a category $\cat{V}$ with a bifunctor $\mprod \colon \cat{V} \times \cat{V} \to \cat{V}$, called \emph{monoidal product}%
\footnote{The standard symbol for the monoidal product is $\otimes$, however we prefer to denote it as $\mprod$ to avoid confusion with other tensorial operations we will deal with in this paper, specifically, the tensor of monads.}, a \emph{unit object} $I \in \cat{V}$, and three natural isomorphisms: 
(\emph{associator}) $\alpha_{V,W,Z} \colon V \mprod (W \mprod Z) \stackrel{\iso}{\longrightarrow} (V \mprod W) \mprod Z$,
(\emph{left unitor}) $\lambda_{V} \colon I \mprod V \stackrel{\iso}{\longrightarrow} V$, 
and (\emph{right unitor}) $\rho_{V} \colon V \mprod I \stackrel{\iso}{\longrightarrow} V$, 
subject to the coherence conditions
\begin{equation*}
\begin{tikzcd}
V \mprod (W \mprod (Y \mprod Z)) \arrow[r, "\alpha" ] \arrow[d, "id \,\mprod\, \alpha"'] &
(V \mprod W) \mprod (Y \mprod Z) \arrow[r, "\alpha" ] & ((V \mprod W) \mprod Y) \mprod Z \\
V \mprod ((W \mprod Y) \mprod Z) \arrow[rr, "\alpha"' ] & & 
(V \mprod (W \mprod Y)) \mprod Z \arrow[u, "\alpha \,\mprod\, id"' ]
\end{tikzcd}
\tag{\sc assoc}
\end{equation*}

\begin{equation*}
\begin{tikzcd}
V \mprod (I \mprod W) \arrow[rr, "\alpha" ] \arrow[dr, "id \,\mprod\, \lambda"'] & &
(V \mprod I) \mprod W \arrow[dl, "\rho \,\mprod\, id" ] \\
& V \mprod W
\tag{\sc unit}
\end{tikzcd}
\end{equation*}
expressing that the operation $\mprod$ is associative, with left/right identity.

A monoidal category is \emph{symmetric} when in addition it is equipped with a natural isomorphism 
(\emph{braiding}) $s_{V,W} \colon  V \mprod W \stackrel{\iso}{\longrightarrow} W \mprod V$ 
such that the following diagrams commute:
\begin{equation*}
\begin{tikzcd}[column sep=large]
V \mprod (W \mprod Z) \arrow[r, "\alpha" ] \arrow[d, "id \,\mprod\, s"'] &
(V \mprod W) \mprod Z \arrow[r, "s" ] & Z  \mprod (V \mprod W) \arrow[d, "\alpha"] \\
V \mprod (Z \mprod W) \arrow[r, "\alpha"' ] & (V \mprod Z) \mprod W  \arrow[r, "s \,\mprod\, id"' ] & 
(Z \mprod V) \mprod W
\end{tikzcd}
\end{equation*}
\begin{equation*}
\begin{tikzcd}
I \mprod V \arrow[rr, "s" ] \arrow[dr, "\lambda"'] & &
V \mprod I \arrow[dl, "\rho" ] \\
& V
\end{tikzcd}
\quad
\begin{tikzcd}
V \mprod W \arrow[rr, "s" ] \arrow[dr, "id"'] & &
W \mprod V \arrow[dl, "s" ] \\
& V \mprod W
\end{tikzcd}
\end{equation*}

A monoidal category is \emph{closed}, if has an internal \emph{hom-functor} $[ -,- ] \colon \cat{V} \times \cat{V} \to \cat{V}$, such that for every object $V \in \cat{V}$, $[V, - ] \colon \cat{V} \to \cat{V}$ is right adjoint to $(V \mprod -) \colon \cat{V} \to \cat{V}$. We will denote the \emph{counit} (or evaluation map) of the adjunction $(V \mprod -) \dashv [V, -]$ by $\ev^V \colon V \mprod [V, -] \nat \textit{Id}$ and the \emph{unit} (or co-evaluatation map) by 
$\coev^V \colon \textit{Id} \to [V, V \mprod - ]$.

\begin{exas} \label{ex:monoidalclosed}
The monoidal closed categories we will consider are $\cat{Set}$, $\cat{Met}$, and $\CMet$.
\begin{enumerate} 
\item $\cat{Set}$ is a symmetric monoidal closed category with Cartesian product $X \times Y$ as monoidal product and internal hom $[X, Y]$ given by the set of functions from $X$ to $Y$. Since the monoidal product
coincides with the categorical product, $\Set$ is \emph{Cartesian closed}.

\item \label{MetClosedStructure}
$\Met$ is a symmetric monoidal closed category, with monoidal product 
$X \mprod Y$ being the extended metric space with underlying set $X \times Y$ and 
distance function  $d_{X \mprod Y}((x,y)(x',y')) = d_X(x,x') + d_Y(y,y')$.
The internal hom $[X, Y]$ is given by the set of non-expansive maps from $X$ to $Y$ with $d_{[X,Y]}(f,g) = \sup_{x \in X} d_Y(f(x),g(x))$ (the point-wise supremum metric) as distance function. The the evaluation map $\ev^X_Y \colon X \mprod [X, Y] \to X$ is given by $\ev^X_Y(f,y)=f(y)$.
Note that $\mprod$ is not the categorical product in $\Met$, for which the distance function would have $\max$ in place of $+$, as one can show that $\Met$ is not Cartesian closed~\cite{Lawvere73}.

\item $\CMet$ has the same symmetric monoidal closed structure of 
$\Met$, as the monoidal product $\mprod$ defined above preserves 
Cauchy completeness.
\end{enumerate}
\end{exas}

\smallskip

%Let $\cat{V}$ be a closed category with monoidal product $\mprod \colon \cat{V} \times \cat{V} \to \cat{V}$. 
A functor $F \colon \cat{V} \to \cat{V}$ is \emph{strong} with \emph{monoidal strength} 
$\mathit{t}_{V,W} \colon V \mprod FW \to F(V \mprod W)$, if $t$ is a natural transformation 
satisfying the following coherence conditions w.r.t.\ the associator $\alpha$ and left unitor $\lambda$ of $\cat{V}$:
\begin{equation*}
\begin{tikzcd}
I \mprod FV \arrow[rd, "\lambda"'] \arrow[r, "t"] & 
F(I \mprod V) \arrow[d, "F\lambda"] \\
 & FV
\end{tikzcd}
\begin{tikzcd}
(U \mprod V) \mprod FW \arrow[dd, "t"'] \arrow[r,"\alpha"] 
	& U \mprod (V \mprod FW) \arrow[d, "U \mprod t "] \\
	& U \mprod F(V \mprod W) \arrow[d, "t"] \\
F((U \mprod V) \mprod W) \arrow[r, "F\alpha"] 
	& F(U \mprod (V \mprod W))
\end{tikzcd}
\end{equation*}
When $\cat{V}$ is symmetric, the dual strength $\hat{t}_{V,W} \colon FW \mprod V \to F(W \mprod V)$ is 
given by $\hat{t} = Fs \circ t \circ s$, where $s_{V,W} \colon V \mprod W \to W \mprod V$ is 
the \emph{braiding} of $\cat{V}$.

A natural transformation $\theta \colon F \nat G$ is said \emph{strong} 
if $F,G$ are strong functors with strengths $t, \sigma$, respectively, and the diagram below commutes
%$\sigma \circ (id \mprod \theta) = \theta \circ t$, 
\begin{equation*}
\begin{tikzcd}
V \mprod FW \arrow[r,"V \mprod \theta"] \arrow[d, "t"']
& V \mprod GW \arrow[d, "\sigma"] \\
F(V \mprod W) \arrow[r,"\theta"] & G(V \mprod W)
\end{tikzcd}
\end{equation*} 
meaning that $\theta$ interacts well with the strengths.

A monad $(T, \eta, \mu)$ with unit $\eta \colon Id \nat T$ and multiplication $\mu \colon TT \nat T$ is \emph{strong} if $T$ is a strong functor with strength $t$ such that the following diagrams commute
%$t \circ (id \mprod \eta) = \eta$ and $\mu \circ tt = t \circ (id \mprod \mu)$.
\begin{equation*}
\begin{tikzcd}
& V \mprod TW \arrow[dd, "t"] \\
V \mprod W \arrow[ru, "V \mprod \eta"] \arrow[rd, "\eta"'] \\
 & T(V \mprod W)
\end{tikzcd}
\quad
\begin{tikzcd}
U \mprod TTV \arrow[d, "t"'] \arrow[r,"U \mprod \mu"] 
	& U \mprod TV \arrow[dd, "t "] \\
T(U \mprod TV) \arrow[d, "Tt"'] \\ 
TT(U \mprod V) \arrow[r, "\mu"] 
	& T(U \mprod V)
\end{tikzcd}
\end{equation*}

Note that strong functors (resp.\ monads) on a symmetric monoidal closed category $\cat{V}$
are equivalent to $\cat{V}$-enriched functors  (resp.\ monads) on the self-enriched category $\cat{V}$~\cite{Kock72}.

\section{Quantitative Equational Theories}
\label{sec:QDS}

Quantitative equations were
introduced in~\cite{MardarePP:LICS16}.  In this framework,  equalities $t =_\e s$ are indexed by a
positive rational number, to capture the idea that $t$ is ``within $\e$''
of $s$.  This intuitive description is formalised in a manner analogous to traditional equational logic.  
In this section, we review this formalism.
 
Let $\Sigma$ be a signature of function symbols
$f \colon n \in \Sigma$ of arity $n \in \naturals$.  Let $X$ be a countable
set of variables, ranged over by $x,y,z, \dots$.  We write 
$\TT[\Sigma]{X}$ for the set of $\Sigma$-terms freely generated over $X$,
ranged over by $t,s,u,\ldots$.

A \emph{substitution of type $\Sigma$} is a function
$\sigma \colon X \to \TT[\Sigma]{X}$, canonically extended to
terms as $\sigma(f(t_1, \dots, t_n)) = f(\sigma(t_1), \dots, \sigma(t_n))$;
we write $\Sub[\Sigma]$ for the set of substitutions of type $\Sigma$.

A \emph{quantitative equation of type $\Sigma$} over $X$ is an expression
of the form $t =_\e s$, for $t,s \in \TT[\Sigma]{X}$ and $\e \in \prationals$.  
We use $\V[\Sigma,X]$ to denote the set of quantitative equations of 
type $\Sigma$ over $X$, and its subsets will be ranged over by $\Gamma, \Theta, \ldots$.
%
% Fix $X$ a countable set of \emph{metavariables}.  A \emph{quantitative  deduction system of type $\Sigma$} is a relation ${\vdash} \subseteq 2^{\E[\Sigma,X]} \times \E[\Sigma,X]$ from the powerset of $\E[\Sigma,X]$ to $\E[\Sigma,X]$ satisfying the following meta-axioms, for each $f \colon n \in \Sigma$ 
%
Let  $\E[\Sigma,X]$ be the set of \emph{conditional quantitative equations} on $\TT[\Sigma]{X}$, which are expressions of the form
\begin{equation*}
\{t_1 =_{\e_1} s_1, \dots, t_n =_{\e_n} s_n\} \vdash t =_\e s \,,
\end{equation*}
for arbitrary $s_i,t_i,s,t\in \TT[\Sigma]{X}$ and $\e_i,\e\in\prationals$. 
As in standard equational logic, we abbreviate $\emptyset \vdash t =_\e s$ to $\vdash t =_\e s$.

\begin{defi}[Quantitative Equational Theory]
A \emph{quantitative equational theory of type $\Sigma$ over $X$} is a set $\U \subseteq \E[\Sigma,X]$ of
conditional quantitative equations satisfying the following conditions, for arbitrary $x, y, z, x_i, y_i \in X$, terms $s,t\in \TT[\Sigma]{X}$, 
rationals $\e,\e' \in \prationals$, and $\Gamma, \Theta \subseteq \V[\Sigma,X]$,
\begin{align*} 
({ \Refl}) \,\, 
& \vdash x =_0 x \, \in \U, \\
({ \Symm}) \,\, 
& \{x =_\e y \} \vdash y =_\e x \in \U \,, \\
({ \Triang}) \,\, 
& \{x =_\e z, z =_{\e'} y \} \vdash x =_{\e+\e'} y \in \U \,, \\
({ \Weak}) \,\, 
& \{x =_\e y \} \vdash x =_{\e+\e'} y \in \U \,, \text{ for all $\e'>0$} \,, \\ 
({ \Nexp{f}}) \,
& \{x_i \,{=_\e}\, y_i \mid i =1 \dots n \} \,{\vdash}\, f(x_1,\dots, x_n) \,{=_\e}\, f(y_1,\dots, y_n) \in \U \,, 
\text{ for $f \colon n \in \Sigma$} \,, \\
({ \Cont}) \,\, 
& \text{If $\{x =_{\e'} y \mid \e'>\e\} \subseteq \U$, 
then $\vdash x =_\e y \in \U $} \,, \\
({ \Subst}) \,\,
& \text{If $\Gamma \vdash t =_\e s \in \U$, then $\sigma(\Gamma) \vdash \sigma(t) =_\e \sigma(s) \in \U$} \,, 
\text{ for $\sigma \in \Sub[\Sigma]$} \,, \\
({ \Assum}) \,\,
& \text{If $t =_\e s \in\Gamma$, then $\Gamma \vdash t =_\e s \in \U$} \,, \\
({ \Cut}) \,\, 
& \text{If $\set{\Gamma \vdash t
=_{\e'} s}{t =_{\e'} s \in \Theta} \subseteq \U$ and $\Theta \vdash t =_\e s \in \U$, then $\Gamma \vdash t =_\e s \in \U$} \,,
\end{align*}
where $\sigma(\Gamma) = \set{\sigma(t) =_\e \sigma(s)}{ t =_\e s \in \Gamma}$.
\end{defi}

The conditions (\Subst), (\Cut), (\Assum) are the usual deductive rules of equational logic.
The axioms (\Refl), (\Symm), (\Triang) correspond, respectively, to reflexivity,
symmetry, and the triangle inequality; (\Weak) represents inclusion of
neighbourhoods of increasing diameter; (\Cont) is the limiting property of a
decreasing chain of neighbourhoods with converging diameters; 
and (\Nexp{$f$}) expresses non-expansiveness of $f \in \Sigma$.

A set $A \subseteq \E[\Sigma,X]$ of conditional quantitative equations \emph{axiomatizes} a quantitative
equational theory $\U$, if $\U$ is the smallest quantitative equational
theory containing $A$.

The models of these theories, called \emph{quantitative $\Sigma$-algebras}, are
$\Sigma$-algebras in $\Met$.  
\begin{defi}[Quantitative Algebra] \label{def:quantitativealgebra}
  A \emph{quantitative $\Sigma$-algebra} is a tuple
  $\A = (A, \Sigma^{\A})$, where $A$ is an extended metric space and 
  $\Sigma^\A = \{f^\A \colon A^n \to A \mid f \colon n \in \Sigma \}$ is a
  set of non-expansive \emph{interpretations} %for the algebraic operators in $\Sigma$ 
  (\ie, $\max_{i} d_A(a_i, b_i) \geq d_A( f^\A(a_1, \dots, a_n), f^\A(b_1, \dots, b_n))$).
%  , \ie, satisfying 
%  the following, for all
%  $0 \leq i \leq n$ and $a_i, b_i \in A$,
%\begin{equation*}
% \max_{i} d_A(a_i, b_i) \geq d_A( f^\A(a_1, \dots, a_n), f^\A(b_1, \dots, b_n)) \,.
%\end{equation*}
\end{defi}

The morphisms between quantitative $\Sigma$-algebras are non-expansive
$\Sigma$-homomorphisms.  Quantitative $\Sigma$-algebras and their morphism
form a category, denoted by $\QA[\Sigma]$.

$\A = (A,\Sigma^\A)$ \emph{satisfies} the
conditional quantitative equation $\Gamma \vdash t =_\e s$ in $\E[\Sigma,X]$,
written $\Gamma \models_\A t =_\e s$, if for any assignment 
$\iota \colon X \to A$, the following implication holds
\begin{align*}
\big( \forall t' =_{\e'} s' \in \Gamma\,,  d_A(\iota(t'),\iota(s')) \leq \e'  \big)
%\text{ implies }
\Rightarrow
d_A(\iota(t),\iota(s)) \leq \e \,,
\end{align*}
where $\iota(t)$ is the homomorphic interpretation of $t$ in $\A$.  

A quantitative algebra $\A$ is said to \emph{satisfy} (or be a \emph{model}
for) the quantitative theory  
$\U[]$, if $\Gamma \models_\A t =_\e s$ whenever $\Gamma \vdash t
=_\e s \in \U[]$.  
We write $\KK[\Sigma]{\U}$ for the collection of models of a theory $\U$ of type $\Sigma$.  

\subsection{Free Monads on Quantitative Equational Theories}
To every signature $\Sigma$, one can associate a \emph{signature endofunctor}
 (also called $\Sigma$) on $\Met$ by:
\begin{equation*}
\Sigma X = \coprod_{f {:} n \in \Sigma} X^n \,.
\end{equation*}
It is easy to see that, by universality of the coproduct, 
quantitative $\Sigma$-algebras correspond to
 $\Sigma$-algebras for the functor $\Sigma$ in $\Met$, and the morphisms
between them to non-expansive homomorphisms of $\Sigma$-algebras.
In the rest of the paper, we will pass between these two points of view as convenient.

In~\cite{MardarePP:LICS16} it is shown that any quantitative theory $\U$ of
type $\Sigma$ induces a monad $T_{\U}$ on $\Met$, called the \emph{free
  monad on $\U$}. The result leading to its definition is
summarized in the following theorem.

\begin{thm}[Free Algebra~\cite{MardarePP:LICS16}] \label{th:freeQAlgebra}
The forgetful functor $\KK[\Sigma]{\U} \to \Met$ 
has a left adjoint.
\end{thm}
The left adjoint assigns to any $X \in \Met$ a \emph{free quantitative $\Sigma$-algebra} 
$(T_X, \psi^{\U}_X)$ satisfying the quantitative theory $\U$, from which one canonically obtains the monad
$(T_{\U}, \eta^{\U}, \mu^{\U})$, with functor $T_{\U} \colon \Met \to \Met$ mapping 
$X \in \Met$ to the carrier $T_X$ of the free algebra.

\smallskip
Directly from the universal property of the adjunction, we get that for any quantitative
$\Sigma$-algebra $(A, a) \in \KK[\Sigma]{\U}$ and non-expansive map
$\beta \colon X \to A$, there exists a unique homomorphism $h \colon T_{\U}X \to A$ of 
quantitative $\Sigma$-algebras making the diagram below commute
\begin{equation*}
\begin{tikzcd}
X \arrow[rd, "\beta"' ] \arrow[r, "\eta^{\U}_X"] & 
T_{\U}X \arrow[d, dashed, "h"]
& \Sigma T_{\U}X \arrow[l, "\psi^{\U}_X"'] \arrow[d,"\Sigma h"] \\
 & A & \Sigma A \arrow[l, "a"']
\end{tikzcd}
\end{equation*}
The map $h$ is called the \emph{homomorphic extension of $a$ along $\beta$}. 

Notice that, homomorphic extensions provide us with a way of defining maps from $T_{\U}X$, for generic $X \in \Met$.
For example, the multiplication $\mu^{\U} \colon T_{\U}T_{\U} \nat T_{\U}$ is defined at component $X$ as the
homomorphic extension of $\psi^{\U}_{X}$ along $id_{T_{\U}X}$ (\ie, the unique map such that $\mu^{\U}_X \circ \eta^{\U}_X = id_{T_{\U}X}$ and $\mu^{\U}_X \circ \psi^{\U}_{T_{\U}X} = \psi^{\U}_{X} \circ \Sigma \mu^{\U}_X$).

\begin{fact}[The Quantitative Term Monad]
In~\cite{MardarePP:LICS16}, the monad $T_{\U}$ has been characterized in the form of a ``quantitative term monad''. 

Concretely, $T_{\U}X$ is defined as the set of $\Sigma$-terms extended with constants in $X$ modulo $0$-provability from $\U$ and $\Gamma_X = \{ {} \vdash x =_\delta y \mid d_X(x,y) \leq \delta \}$ (\ie, two terms $t,s \in \TT[\Sigma]{X}$ are considered equal if $\vdash t =_0 s$ is provable from $\U$ and $\Gamma_X$). This set is endowed with the distance function
\begin{equation*}
  d_{T_{\U}X}(t,s) = \inf \{\e  \mid {} \vdash t =_\e s \text{ is provable from $\U$ and $\Gamma_X$} \} \,.
\end{equation*}
Intuitively, the distance between the terms $t$ and $s$ is the smallest $\e$ such that $\vdash t =_\e s$
is deducible by using quantitative equations from the theory $\U$ and axioms in $\Gamma_X$ over constants terms in $X$; if $\not\vdash t =_\e s$ (not provable) for any $\e \in \prationals$, the distance is $\infty$.

The unit and multiplication act as in a standard term monad: the unit $\eta^{\U}_X \colon X \to T_{\U}X$ interprets the elements of $X$ as terms; the multiplication $\mu^{\U}_X \colon T_{\U}T_{\U}X \to T_{\U}X$ takes a term over terms in $T_{\U}T_{\U}X$ and flattens it out into a single term $T_{\U}X $ by term composition. The key detail is that these maps are non-expansive w.r.t.\ the distance defined above.
\end{fact}   

In~\cite{BacciMPP18}, it is shown that whenever the quantitative 
theory $\U$ is \emph{basic},
\ie, it can be axiomatized by a set of conditional equations of the form 
\begin{equation*}
\{x_1 =_{\e_1} y_1, \dots, x_n =_{\e_n} y_n\} \vdash t =_\e s \,,
\end{equation*}
where $x_i, y_i \in X$ (\cf~\cite{MardarePP17}), then the EM algebras for $T_{\U}$
are in 1-1 correspondence with the quantitative algebras satisfying $\U$:
\begin{thm} \label{th:EilenbergMoore}
For any basic quantitative theory $\U$ of type $\Sigma$,
$\EMAlg{T_{\U}} \iso \KK[\Sigma]{\U}$.  
\end{thm}

\subsection{Completion of Quantitative Algebras}

Sometimes it is convenient to consider the quantitative $\Sigma$-algebras whose
carrier is a complete extended metric space.  This class of algebras forms a full
subcategory of $\QA$: the category of complete quantitative algebras, denoted $\CQA$.

Then, it is natural to ask whether the standard Cauchy completion of metric spaces lifts to a notion of Cauchy completion of quantitative algebras. This is done as follows:
\begin{defi}(Algebra Completion)
  The \emph{Cauchy completion of a quantitative $\Sigma$-algebra}
  $\A = (A,\Sigma^\A)$, is the quantitative $\Sigma$-algebra
  $\overline{\A} = (\overline{A}, \Sigma^{\overline{\A}})$, where
  $\overline{A}$ is the Cauchy completion of $A$ and
  $\Sigma^{\overline{\A}} \,{=}\, \{ f^{\overline{\A}} \colon
  \overline{A}^n \to \overline{A} \mid f \colon n \,{\in}\, \Sigma \}$ is
  such that for Cauchy sequences $(b^i_j)_j$ converging to $b^i \in \overline{A}$, 
  for $1 \leq i \leq n$, 
  %$f^{\overline{\A}}(b^1, \ldots, b^n) =\lim_j f^{\A}(b^1_j, \ldots, b^n_j)$.
  \begin{equation*}
 	f^{\overline{\A}}(b^1, \ldots, b^n) =\lim_j f^{\A}(b^1_j, \ldots, b^n_j) \,.
  \end{equation*}
\end{defi}

The above definition extends to a functor $\CC \colon \QA \to \CQA$, also called the Cauchy completion functor, mapping a quantitative algebra to its completion. As happens with metric spaces, this functor is the left adjoint to the embedding $\CQA \hookrightarrow \QA$.

\smallskip
Let $\U$ be a quantitative equational theory. Similarly to before, we consider the 
full subcategory of complete quantitative $\Sigma$-algebras that are models of $\U$, denoted by $\CC\KK{\Sigma, \U}$.
Then, given an algebra $\A$ satisfying $\U$, an interesting question is whether its completion $\overline{\A}$ is still a model for all the equations in $\U$. In other words, does the Cauchy completion functor restrict to 
$\CC \colon \KK{\Sigma, \U} \to \CC\KK{\Sigma,\U}$?

The answer is positive whenever $\U$ can be
axiomatized by a collection of \emph{continuous schemata} of quantitative equations, \ie, sets of quantitative equations of the form
\begin{align*}
\{x_i =_{\e_i} y_i \mid i=1..n\} \vdash t =_\e s \,, \quad
\text{ for all $\e \geq f(\e_1, \dots, \e_n)$,}
\end{align*}
where $f \colon \preals^n \to \preals$ is a continuous real-valued function, $\e, \e_i \in \prationals$, and $x_i, y_i \in X$. We call such a theory \emph{continuous}.

\begin{rem}
Asking for a theory to be continuous is necessary. For a counterexample, consider a signature having a single unary function symbol $g \colon 1$ and a theory $\U$ with axiom
\begin{equation*}
\{x =_2 y \} \vdash g(x) =_1 g(y) \,.
\end{equation*}
A quantitative algebra that is a model of $\U$ is given by the union of open intervals on the reals $[0,1) \cup (3,4]$ (with usual metric) by interpreting $g$ as the identity function. The Cauchy completion of this algebra has the closed set $[0,1] \cup [3,4]$ as the carrier and interprets $g$ again as the identity.
However, the axiom above is not satisfied in $[0,1] \cup [3,4]$ as this would require $|g(1) - g(3)| \leq 1$.
\end{rem}

When $\U$ is a continuous theory, the Cauchy completion functor $\CC \colon \KK{\Sigma, \U} \to \CC\KK{\Sigma,\U}$ is left adjoint to the functor embedding $\CC\KK{\Sigma,\U}$ into $\KK{\Sigma,\U}$.

Moreover, for this class of theories, a similar result to Theorem~\ref{th:freeQAlgebra} also holds.
\begin{thm}[Free Complete Algebra~\cite{MardarePP:LICS16}] \label{th:freeCQAlgebra}
For any continuous quantitative equational theory $\U$ of type $\Sigma$, 
the forgetful functor $\CC\KK[\Sigma]{\U} \to \CMet$ has a left adjoint.
\end{thm}
As a direct consequence of the above and Theorem~\ref{th:freeQAlgebra}, when $\U$ is continuous, for any
  $X \in \CMet$, quantitative $\Sigma$-algebra $(A, a)$ in
  $\CC\KK[\Sigma]{\U}$ and non-expansive map $\beta \colon X \to A$, there
  exists a unique homomorphism $h \colon \CC T_{\U} \to A$ making the following diagram commute
\begin{equation*}
\begin{tikzcd}
X \arrow[rd, "\beta"' ] \arrow[r, "\CC \eta^{\U}_X"] & 
\CC T_{\U}X \arrow[d, dashed, "h"]
& \Sigma \CC T_{\U}X \arrow[l, "\CC \psi^{\U}_X"'] \arrow[d,"\Sigma h"] \\
 & A & \Sigma A \arrow[l, "a"']
\end{tikzcd}
\end{equation*}
This, in particular, tells us that $(\CC T_{\U}X, \CC \psi^{\U}_X)$ is the \emph{free complete quantitative algebra} for an arbitrary metric space $X$, implying that $\CC T_{\U}$ is the free monad on $\U$ in $\CMet$.

\smallskip

Note that, by definition, continuous theories are basic. Thus, by essentially the same arguments of  Theorem~\ref{th:EilenbergMoore}, we have a 1-1 correspondence between the EM algebras for $\CC T_{\U}$ and the quantitative algebras satisfying a continuous theory $\U$.
\begin{thm} \label{th:CompleteEilenbergMoore}
For any continuous quantitative theory $\U$ of type $\Sigma$,
$\EMAlg{\CC T_{\U}} \iso \CC\KK[\Sigma]{\U}$.  
\end{thm}

%%%%%%%%%%%%%%%%%%%%%%%%%%%%%%%%%%%%%%%%%%%%%%
\section{Algebraic Presentation of Monads over Metric Spaces: Examples}\label{sec:examples} 

A presentation of a $\Set$ monad $T$ is an algebraic theory $(\Sigma, E)$ (\ie, a signature $\Sigma$ and 
a set $E$ of equations $s = t$ between $\Sigma$-terms) such that the full subcategory of the
universal algebras that satisfy all the equations in $E$ is isomorphic to 
the Eilenberg-Moore category $\EMAlg{T}$. If $T$ has a presentation $(\Sigma, E)$, then it is algebraic, because it is isomorphic to the (term) monad $T_E$ freely generated from the equations in $E$.

As in this paper we deal with monads on $\Met$, their presentations will be given in terms of \emph{quantitative algebraic theories} $(\Sigma, \U)$ (\ie, a signature $\Sigma$ and a quantitative equational theory $\U$ of type $\Sigma$) and, in complete analogy with the 
above, $(\Sigma, \U)$ is a presentation for $T$, if the category of quantitative algebras 
that are models of $\U$ is isomorphic to the Eilenberg-Moore category of $T$ 
(in short, $\KK{\Sigma,\U} \iso \EMAlg{T}$). 

In this section, we propose quantitative versions of several $\Set$ monads classically used as  
computational effects in programming languages and for each of them provide a quantitative equational presentation in the sense explained above. The computational effects we consider are: termination and exceptions (Section~\ref{sec:termination}), interactive input (Section~\ref{sec:controperators}), reading/writing (Section~\ref{sec:ReaderWriterMonads}), nondeterminism (Section~\ref{sec:Hausdorff}), and probabilistic choice (Section~\ref{sec:probchoice}).

%%%%%%%%%%%%%%%%%%%%%%%%%%%%%%%%%%%%%%%%%%%%%%
\subsection{Termination and Exceptions}
\label{sec:termination}
The monadic effect for termination in $\Set$ is given by the \emph{termination monad} (a.k.a.\ \emph{maybe monad}), denoted by $(- + 1)$, that maps a set $X$ to $X + 1$, where $+$ denotes the coproduct (hence, disjoint union) and $1 = \{ * \}$ is the terminal object in $\Set$, representing the effect of terminating the computation in an error state. The unit and multiplication are canonically defined from the universal property of the coproduct:
\begin{align*}
\inj_{l} \colon X \to (X + 1) &&& \textsc{unit} \\
[id_{X+1}, \inj_{r}] \colon ((X + 1) + 1) \to (X + 1) &&& \textsc{multiplication}
\end{align*}
where in the above $\inj_{l}$ and $\inj_r$ are the left and right canonical injections into the coproduct $X+1$.
The \emph{exception monad}, denoted by $(- + E)$, generalizes 
the above by mapping a set $X$ to $X + E$, where $E$ is a fixed set of exceptions. Intuitively, an effectful computation of this type, rather than just terminating in a generic error state, allows one to raise an exception $e \in E$ that represents extra information on the causes of the termination.
This monad has a straightforward algebraic presentation, with signature 
\begin{equation*}
\Sigma_E = \{ \textsf{raise}_e \colon 0 \mid e \in E\} \,,
\end{equation*}
having only nullary operation symbols (\ie, constants) $\textsf{raise}_e$, for each exception $e \in E$ and
equational theory containing only identities $t=t$ between terms (\ie, the trivial theory with no axioms on the constant symbols in the signature).

In the quantitative case, the corresponding exception monad on $\Met$ is still given by
$(- + E)$, with the only difference being that now $E$ is an extended metric space with
metric measuring the distance between exceptions. Computationally, this means that one may
measure the difference between the different types of terminations. This interpretation can be
useful, for example, in scenarios where exceptions carry the time stamp of the moment
they have been thrown, thus allowing one to compare program implementations by measuring 
the frequency of which exceptions are thrown.

For $E$ an extended metric space of exceptions, we define the \emph{quantitative algebraic theory of 
exceptions} over $E$, by taking the same signature as above, namely $\Sigma_E$, and adding to 
the theory the quantitative equations
\begin{equation*}
 {\vdash}\, \textsf{raise}_{e_1} \,{=_{\delta}}\, \textsf{raise}_{e_2} \,, \qquad \text{for $\delta \geq d_E(e_1,e_2)$}\,,
\end{equation*}
for any pair of exceptions $e_1,e_2 \in E$ and positive rational $\delta$. The r{\^o}le of this axiom is to lift to 
the set of terms the underlying metric of $E$. We denote this quantitative theory by $\mathcal{E}_E$.

\medskip
It is not difficult to show that for any $X \in \Met$, the quantitative $\Sigma_E$-algebra $(X + E, \phi_X)$
interpreting $\textsf{raise}_e \colon 0 \in \Sigma_E$ as $e \in X + E$ for each exception $e \in E$, formally defined by
\begin{align*}
  \phi_X \colon \Sigma_E (X+E) \to X + E
  &&
  \phi_X(\textsf{raise}_e) = e \,,
\end{align*}
is isomorphic to the free quantitative algebra in $\KK{\Sigma_E, \mathcal{E}_E}$. From this, we obtain:
\begin{thm} \label{th:isoExceptionMet}
The monads $T_{\mathcal{E}_E}$ and $(- + E)$ on $\Met$ are isomorphic. 
\end{thm}
As the quantitative theory $\mathcal{E}_E$ is basic, by Theorems~\ref{th:EilenbergMoore} and \ref{th:isoExceptionMet}, we have that $(\Sigma_E, \mathcal{E}_E)$ is a presentation of the exception monad $(- + E)$ on $\Met$ 
(\ie, $\EMAlg{(- + E)} \iso \KK{\Sigma_E, \mathcal{E}_E}$).

\medskip
The exception monad $(- + E)$ is well defined also in $\CMet$, the only difference being that one assumes 
$E$ to a complete metric space. 
As the theory $\mathcal{E}_E$ is continuous, by similar arguments to the above, $(\Sigma, \mathcal{E}_E)$ is a presentation of this monad in $\CMet$, that is 
$\EMAlg{(- + E)} \iso \CC\KK[\Sigma_E]{\mathcal{E}_E}$.
\begin{thm} \label{th:isoExceptionCMet}
The monads $\CC T_{\mathcal{E}_E}$ and $(- + E)$ on $\CMet$ are isomorphic. 
\end{thm}

%%%%%%%%%%%%%%%%%%%%%%%%%%%%%%%%%%%%%%%%%%%%%%
\subsection{Interactive Input}
\label{sec:controperators}

Interactive input on a (nonempty) finite set $I = \{i_1, \dots, i_n\}$ of symbols, can be expressed
by a $n$-ary operation $\textsf{input}(t_1, \dots, t_n)$ representing the computation that proceeds 
as $t_j$ on input $i_j$. In $\Set$, the corresponding monadic effect is given in terms of the free monad 
on $(-)^{|I|}$, with algebraic presentation given by the trivial equational theory with no axioms
on the input operations.

In the quantitative setting, one may wish the input operation to be contractive (\ie, $c$-Lipschitz continuous for some $0 < c < 1$) so that repeated input operations eventually converge to a fixed point on complete metric spaces (\cf\ Banach fixed point theorem). This can be expressed by means of the following 
quantitative equations
\begin{equation*}
 \{x_1 \,{=_\e}\, y_1, \dots, x_n \,{=_\e}\, y_n\} \,{\vdash}\, \textsf{input}(x_1,\dots, x_n) \,{=_{\delta}}\, \textsf{input}(y_1,\dots, y_n) \,, \qquad \text{for $\delta \geq c \e$}
\end{equation*}
expressing that the input operation is contractive (with contractive factor $c$). 

The corresponding quantitative monadic effect on $\Met$ (and $\CMet$ too) is given by the free monad on $c \cdot (-)^{|I|}$,
where $c \cdot -$ is the \emph{rescaling functor}, mapping a metric space $(X,d_X)$ to $(X, c \cdot d_X)$.

\medskip
The quantitative algebras for interactive inputs described above are a particular instance
of the \emph{algebras of contractive operators} from~\cite{BacciMPP18}, which we recall below.

\subsubsection{Algebras of contractive operators}
A signature of contractive operators $\Sigma$ is an (at most countable) collection of
function symbols $f$ with associated arity $n \in \naturals$ and \emph{contractive factor} $0 < c < 1$. We write this
as $f \colon \tupl{n,c} \in \Sigma$.
The quantitative theory for $\Sigma$, written $\O{\Sigma}$, is the
smallest theory satisfying, for each
$f \colon \tupl{n,c} \in \Sigma$, the quantitative equations
\begin{align*}
(\Lip{f})\,\,
& \{x_1 \,{=_\e}\, y_1, \dots, x_n \,{=_\e}\, y_n\} \,{\vdash}\, f(x_1,\dots, x_n) \,{=_{\delta}}\, f(y_1,\dots, y_n) \,,  \qquad \text{for $\delta \geq c \e$}\,.
\end{align*}
The axiom ($\Lip{f}$) is just asking the interpretation of $f$ to be $c$-Lipschitz continuous.

The quantitative algebras that are models for $\O{\Sigma}$ are called \emph{algebras of the contractive 
signature $\Sigma$}, and we denote their category as $\KK[\Sigma]{\O{\Sigma}}$.

\subsubsection{Monads of contractive operators}
\label{sec:ContractiveOp}

For a contractive signature $\Sigma$, we define a modification of the signature
endofunctor on $\Met$ by:
\begin{equation}
\tilde\Sigma X = \coprod_{f \colon \tupl{n,c} \in \Sigma} c \cdot X^n \,.
\label{eq:contractiveSignFunctor}
\end{equation}

It is not difficult to show that the quantitative $\Sigma$-algebras satisfying 
$\O{\Sigma}$ are in one-to-one correspondence with the algebras of $\tilde\Sigma$,
that is $\KK[\Sigma]{\O{\Sigma}} \iso \Alg{\tilde\Sigma}$. In virtue of this, we will pass between these two points of view as convenient, and say that
an algebra of $\tilde\Sigma$ satisfies $\O{\Sigma}$.

Next we show that the free monad $T_{\O{\Sigma}}$ is isomorphic to $\tilde\Sigma^*$, 
the free monad on $\tilde\Sigma$.
For this result, we first need some discussion about sufficient conditions for the existence of free monads
on an endofunctor.

\begin{rem} \label{rem:freeExists}
Given any endofunctor $H$ on a category $\mathbf{C}$, we
write $(\mu y. Hy, \alpha_H)$ for the initial $H$-algebra, if it
exists. If $\mathbf{C}$ has binary coproducts, the free \mbox{$H$-algebra}
on $X \in \mathbf{C}$ can be identified with
$(\mu y. (Hy + X), \alpha_{H + X})$, and the one exists if and only if the
other does. These free algebras exist if, for example, $\mathbf{C}$ is
locally countably presentable and $H$ has countable rank. In this case
the forgetful functor $U^H \colon \Alg{H} \to \mathbf{C}$ has a left adjoint,
mapping $\mathbf{C}$-objects to their corresponding free $H$-algebra.
\end{rem}

We see from Remark~\ref{rem:freeExists} that, if $\mathbf{C}$ has binary coproducts, then $H^*$ can
be identified with $\mu y. (Hy + -)$ and the former exists if and only if the
other does. We further see that if $\mathbf{C}$ is locally countably
presentable and $H$ has countable rank, then $H^*$ exists~\cite{Kelly1980}.
Moreover, as $H^*$ is algebraic, the Eilenberg-Moore category $\EMAlg{H^*}$ is isomorphic 
to the category $\Alg{H}$ of algebras of $H$ (see Section~\ref{sec:monadsPrelim}).

Therefore, since $\Met$ is locally countably presentable~\cite{AdamekMM12} (see also Appendix~\ref{sec:EMet-lcp}) and $\tilde\Sigma$ has countable rank, the free algebra for $\tilde\Sigma$ exists and so does the free monad $\tilde\Sigma^*$.

As $\EMAlg{\tilde\Sigma^*}$ and $\KK{\Sigma, \O{\Sigma}}$ are isomorphic
and $\O{\Sigma}$ is basic, by freeness of $T_{\O{\Sigma}}$ (Theorem~\ref{th:freeQAlgebra}) 
the following holds: 
\begin{thm} \label{th:isoContractiveMonad}
The monads $T_{\O{\Sigma}}$ and $\tilde\Sigma^*$ on $\Met$ are isomorphic. 
\end{thm}

The situation is similar in the category $\CMet$ of complete extended metrics.
As $\CMet$ has coproducts and finite products, and rescaling a metric by a factor $0 < c < 1$
preserves completeness, for any contractive signature $\Sigma$, the endofunctor 
$\tilde\Sigma$ defined as in \eqref{eq:contractiveSignFunctor} is well defined in $\CMet$.
Moreover, $\CMet$ is locally countably presentable~\cite{AdamekMM12} and, since $\tilde\Sigma$ has 
countable rank, by Remark~\ref{rem:freeExists} the free monad $\tilde\Sigma^*$ on $\CMet$ 
exists and is algebraic. 

Similar to the previous case, also this time the Eilenberg-Moore category 
$\EMAlg{\tilde\Sigma^*}$ is isomorphic to $\CC\KK{\Sigma, \O{\Sigma}}$.
As $\O{\Sigma}$ is a continuous quantitative theory, by Theorem~\ref{th:CompleteEilenbergMoore}
and repeating the same argument we used before, we obtain:
\begin{thm} \label{th:isoContraciveMonadCompletion}
The monads $\CC T_{\O{\Sigma}}$ and $\tilde\Sigma^*$ on $\CMet$ are isomorphic. 
\end{thm}

%%%%%%%%%%%%%%%%%%%%%%%%%%%%%%%%%%%%%%%%%%%%%%
\subsection{Reading/Writing}
\label{sec:ReaderWriterMonads}
The monadic effects for reading and writing in $\Set$ are respectively given by the so-called, reader and writer monads. These effects, respectively, allow a computation to read from a finite list of globally declared 
variables, and write on an output tape though to record annotations or just used as standard output. 
Their formal definitions are recalled below. 

Given a set $E$ of input values, the \emph{reader monad} on $\Set$, denoted by $(-)^E$, maps a set $X$ to 
$X^E$, the set all of functions from $E$ to $X$, and acts on morphism $f \colon X \to Y$ as 
$f^E(g) = f \circ g$, for all $g \in X^E$. The unit $\kappa_X \colon X \to X^E$ and multiplication 
$\zeta_X \colon (X^E)^E  \to X^E$ are respectively given as follows, for $x \in X$, $e \in E$, and $g \colon E \to X^E$
\begin{equation}
\begin{aligned}
\kappa_X(x)(e) &= x  && \textsc{unit}
\\
\zeta_X(g)(e) &= f(e)(e) && \textsc{multiplication}
\end{aligned}
\label{eq:readerUnitMult}
\end{equation}

Given a set $\Lambda$ of output values, equipped with a monoidal structure $(\Lambda, *, 0)$, the \emph{writer monad} on $
\Set$, denoted by $(\Lambda \times -)$, acts on sets $X$ as $\Lambda \times X$, where $\times$ denotes the product (hence, Cartesian product), and on morphisms $f \colon X \to Y$ as $(\Lambda \times f)(\alpha, x) = (\alpha, f(x))$, for $x \in X$ and $\alpha \in \Lambda$. The unit  $\tau \colon Id \nat (\Lambda \times -)$ and multiplication $\varsigma \colon (\Lambda \times (\Lambda \times -)) \nat (\Lambda \times -)$ are 
respectively given as follows, for $x \in X$ and $\alpha, \alpha' \in \Lambda$ 
\begin{equation}
\begin{aligned}
\tau_X(x) &= (0, x) \,, && \textsc{unit}
\\
\varsigma_X((\alpha, (\alpha', x))) &= (\alpha* \alpha', x) && \textsc{multiplication}
\end{aligned}
\label{eq:writerUnitMult}
\end{equation}

\medskip
In the quantitative case, one wishes to define analogous monads on the category $\Met$ of extended metric spaces. However, extra care has to be taken as the definition of the above monads crucially exploits the Cartesian closed structure of $\Set$, and we already have seen that $\Met$ is not Cartesian closed (Example~\ref{ex:monoidalclosed}).

\begin{rem} \label{rem:readerCC}
The reader monad is always well defined in a Cartesian closed category $\cat{C}$.
Fix an object $E \in \cat{C}$. The reader monad $(-)^E$ has unit and multiplication
respectively given by 
\begin{align*}
 X \iso X^1 \xrightarrow{\, X^! \,} X^E  && \text{and} &&
 (X^E)^E \iso X^{E \times E} \xrightarrow{\, X^\delta \,} X^E  \,,
\end{align*}
where $! \colon E \to 1$ is the unique map to the terminal object and $\delta \colon E \to E \times E$
the diagonal map $\delta = \tupl{id, id}$.
However, this definition does not generalise to arbitrary monoidal closed categories and $\Met$ 
is an example of such.  The specific problem with $\Met$ is that
$\delta \colon E \to E \mprod E$ is not well defined for arbitrary $E \in \Met$,
as non-expansiveness requires that
\begin{equation*}
d_E(e, e') \geq d_{E \mprod E}(\delta(e), \delta(e')) = d_E(e,e') + d_E(e,e') \,,
\end{equation*}
which holds only when $E$ has the discrete metric (that assigns infinite distance to any pair of distinct elements). From this, we see that a quantitative analogue of the reader 
monad can be obtained if we restrict our attention only to spaces $E$ with discrete metrics. 
\end{rem}

For a set $E$ denote by $\underline{E}$ the corresponding extended 
metric space equipped with discrete metric. 
The reader monad on $\Met$, denoted by $(-)^{\underline{E}}$, 
assigns to each $X \in \Met$ the internal hom $[E, X]$ of (necessarily non-expansive) maps 
from $\underline{E}$ to $X$ with point-wise supremum metric (\cf\ Example~\ref{ex:monoidalclosed}\eqref{MetClosedStructure}) and acts on morphisms $f \colon X \to Y$ as $f^E(g) = f \circ g$, for all $g \in [E,X]$.
The unit and multiplication are defined as in~\eqref{eq:readerUnitMult}, where non-expansiveness 
for the multiplication's components at $X \in \Met$ follows because $\underline{E}$ 
has discrete metric.

\medskip
As for a quantitative analogue of the writer monad, we will assume the set of output values $\Lambda$ to be 
an extended metric space and further require its monoid structure $(\Lambda, *, 0)$ to have  
multiplication $* \colon \Lambda \times \Lambda \to \Lambda$, satisfying the following condition 
\begin{equation}
d_\Lambda(\alpha * \beta, \alpha' * \beta') \leq d_\Lambda(\alpha,\alpha') + d_\Lambda(\beta,\beta') \,,
\label{*nonexp}
\end{equation}
for all $\alpha, \alpha', \beta, \beta' \in \Lambda$, that is, $*$ is a non-expansive map of type $\Lambda \mprod \Lambda \to \Lambda$ in $\cat{Met}$.

Then, the writer monad on $\Met$, denoted by $(\Lambda \mprod -)$, acts on objects $X \in \Met$ as 
 $(\Lambda \mprod X)$, where $\mprod$ denotes the monoidal product discussed in Example~\ref{ex:monoidalclosed}\eqref{MetClosedStructure}, and on morphisms $f \colon X \to Y$ as
$(\Lambda \mprod f)(\alpha, x) = (\alpha, f(x))$, and $\alpha \in \Lambda$. The unit and multiplication 
are defined as in~\eqref{eq:writerUnitMult}, where the assumption \eqref{*nonexp} 
is necessary for proving non-expansiveness for the multiplication's components.

Below we provide quantitative equational representations for these two $\Met$ monads. 

\subsubsection{Reader Algebras}
\label{sec:readeralgebras}
Let $E = \ens{e_1, \dots, e_n}$ be a finite set of input values of which we assume a fixed enumeration.
The \emph{quantitative reader algebras} of type $E$ are the
algebras for the signature 
\begin{equation*}
  \Sigma_{\R[E]} = \ens{ \rd \colon |E| }
\end{equation*}
having only one operator $\rd$ of arity equal to the number of the input values 
in $E$, and satisfying the following quantitative equations
\begin{align*}
{ (\Idem)} 
& \vdash x =_0 \rd(x, \dots, x) \,, \\
{ (\Diag)} 
& \vdash \rd(x_{1,1}, \dots, x_{n,n} ) =_0 
	\rd( \rd(x_{1,1}, \dots, x_{1,n}), \dots, \rd(x_{n,1}, \dots, x_{n,n}) ) \,.
\end{align*}
We call the quantitative theory induced by the equations above, written $\R[E]$ (or simply $\R[]$ when $E$ is clear), \emph{quantitative theory of reading computations}.

Intuitively, the term $\rd(t_1, \dots, t_n)$ 
can be interpreted as the computation that proceeds as $t_i$ after reading the value $e_i$ from its input.
So $\rd$ describes the operation of reading from an input with values in $E$.
The equation (\Idem) says that if we ignore the value of the input the reading of it is not observable;
(\Diag) says that the resulting computation after reading the input is the same no matter how many times we read it. 

\begin{rem}
For the binary case ($|E| = 2$) we can think of $\rd$ as an \emph{if-then-else} statement $b?(S,T)$
checking for the value of a fixed global Boolean variable $b$ and proceeding as $S$ when $b = \text{true}$,
and as $T$ otherwise. In this case, $(\Idem)$ and $(\Diag)$ express the standard program equivalences
\begin{align*}
  S \equiv b?(S,S) &&\text{and} &&
  b?(S,T) \equiv b?\big(b?(S,T'),b?(S',T)\big) \,. 
\end{align*}
We should also remark that (\Idem) and (\Diag) are purely equational judgements
and are the equations presenting the reader monad $(-)^E$ on $\Set$.
Shortly, we will see that they also provide a presentation for the reader monad $(-)^{\underline{E}}$ on $\Met$. This should not surprise, since having input symbols equipped with the discrete metric makes 
the two monads equivalent.
\end{rem}

Next we show that the reader monad $(-)^{\underline{E}}$ is isomorphic to the free monad $T_{\R[]}$ on $\R[]$. Consequently, as the theory $\R[]$ is basic, by Theorem~\ref{th:EilenbergMoore}, $\EMAlg{(-)^{\underline{E}}} \iso \KK[\Sigma_{\R[]}]{\R[]}$.
In other words, $(\Sigma_{\R[]}, \R[])$ is a quantitative equational presentation of the monad $(-)^{\underline{E}}$ on $\Met$. 
\begin{thm} \label{th:isoReaderMonad}
The monads $T_{\R[]}$ and $(-)^{\underline{E}}$ in $\Met$ are isomorphic.  
\end{thm}
\begin{proof}
We prove this statement by showing that, for each $X \in \Met$, $X^{\underline{E}}$ has a quantitative
algebraic structure that is free in $\KK{\Sigma_{\R[]}, \R[]}$ with universal natural arrow $\kappa_X$;
and show that the units and multiplications of the two monads coincide (up to iso).

For any $X \in \Met$, we define the quantitative 
$\Sigma_{\R[]}$-algebra $(X^{\underline{E}}, \rho_X)$ as follows, for arbitrary maps 
$f_1,\dots, f_n \colon \underline{E} \to X$
\begin{align*}
  \rho_X \colon \Sigma_{\R[]} X^{\underline{E}} \to X^{\underline{E}}
  &&
  \rho_X(\rd(f_1,\dots, f_n))(e_i) = f_i(e_i) \,.
\end{align*}

Next we show that it satisfies $\R[]$.
For convenience, let $\rd^\rho$ denote the interpretation of the operator symbol 
$\rd \colon n \in \Sigma_{\R[]}$
in the algebra $(X^{\underline{E}}, \rho_X)$.
Soundness for (\Nexp{\rd}) follows by the fact that $\rho_X$ is a well defined 
map in $\Met$ as shown below
\begin{align*}
d_{X^{\underline{E}}}(\rd^\rho(f_1, \dots, f_n), \rd^\rho(g_1, \dots, g_n)) 
&= \sup_{e_i} d_X(\rd^\rho(f_1, \dots, f_n)(e_i), \rd^\rho(g_1, \dots, g_n)(e_i)) \\
&= \sup_{e_i} d_X(f_i(e_i), g_i(e_i)) \\
&\leq \max_{j} \Big( \sup_{e_i \in E} d_X(f_j(e_i), g_j(e_i)) \Big) \\
&\leq \max_{j} d_{X^{\underline{E}}}(f_j, g_j) \,.
\end{align*}
Soundness for (\Idem) follows by definition of $\rho$ as, for all $e_i \in E$
\begin{equation*}
 \rd^\rho(f,\dots, f)(e_i) = f(e_i) \,.
\end{equation*}
Soundness for (\Diag) also follows by definition, as
\begin{align*}
\rd^\rho(\rd^\rho(f_{1,1},\dots, f_{1,n}),\dots, \rd^\rho(f_{n,1},\dots, f_{n,n}) )(e_i) 
&= \rd^\rho(f_{i,1},\dots, f_{i,n})(e_i) \\
&= f_{i,i}(e_i) \\
&= \rd^\rho(f_{1,1},\dots, f_{n,n})(e_i) \,.
\end{align*}

Now we prove freeness.
Let $(A,a) \in \KK{\Sigma_{\R[]}, \R[]}$ and let $\beta \colon X \to A$ be a non-expansive map. Define $h \colon X^{\underline{E}} \to A$ as follows, 
for arbitrary $f \colon {\underline{E}} \to X$
\begin{equation*}
 h(f) = a(\rd(\beta(f(e_1)), \dots, \beta(f(e_n)) )) \,.
\end{equation*}
As it is the composition of non-expansive maps, then also $h$ is non-expansive.
Next, we prove that $h$ is the only map making the diagram below commute.
\begin{equation*}
\begin{tikzcd}
X \arrow[rd, "\beta"' ] \arrow[r, "\kappa_X"] & 
X^{\underline{E}} \arrow[d, dashed, "h"]
& \Sigma_{\R[]} X^{\underline{E}} \arrow[l, "\rho_X"'] \arrow[d,"\Sigma_{\R[]}h"] \\
 & A & \Sigma_{\R[]}A \arrow[l, "a"']
\end{tikzcd}
\end{equation*}
Let $\rd^\rho$ and $\rd^a$ denote the 
interpretations of $\rd \colon n \in \Sigma_{\R[]}$, respectively, in the algebras $(X^{\underline{E}}, \rho_X)$
and $(A,a)$.
Let $x \in X$.  Then 
\begin{align*}
(h \circ \kappa_X)(x)
&= \rd^a(\beta(\kappa_X(x)(e_1)), \dots, \beta(\kappa_X(x)(e_n)) )
\tag{def.\ $h$} \\
&= \rd^a(\beta(x), \dots, \beta(x) )
\tag{def.\ $\kappa$} \\
&= \beta(x) \,.
\tag{\Idem}
\end{align*}
Let $f_1, \dots, f_n \colon \underline{E} \to X$.  Then
\begin{align*}
(h \circ \rho_X)(\rd(f_1, \dots, f_n))
&= \rd^a(\beta(f_1(e_1)), \dots, \beta(f_n(e_n)) ) \tag{def.\ $h$ and $\rho$} \\
&= \rd^a \Big( 
	\rd^a( \beta(f_1(e_1)), \dots, \beta(f_1(e_n)) ), \dots \\
&\phantom{\rd^a \Big( \rd^a( \beta}	
	\dots,  \rd^a( \beta(f_n(e_1)), \dots, \beta(f_n(e_1)) ) \Big)
\tag{\Diag} \\
&= \rd^a(h(f_1), \dots, h(f_n))  \tag{def.\ $h$} \\
&= (a \circ \Sigma_{\R[]}h)(\rd(f_1, \dots, f_n)) \,.  \tag{def.\ $\rd^a$ and $\Sigma_{\R[]}$}
\end{align*}
Hence $h$ is a $\Sigma_{\R[]}$-homomorphism, that is, $h \circ \rho_X = a \circ \Sigma_{\R[]} h$.

It remains to prove uniqueness. 
Assume there exists $g \colon X^{\underline{E}} \to A$ such that 
$g \circ \kappa_X = \beta$ and
$g \circ \rho_X = a \circ \Sigma_{\R[]}g$.
Next, we show $h = g$.  Notice first that for any $f \colon X^{\underline{E}} \to X$,
$f = \rd^\rho( \kappa_X(f(e_1)), \dots, \kappa_X(f(e_n)) )$, as for all $e_i \in E$, the following
holds:
\begin{align*}
f(e_i) 
&= \kappa_X(f(e_i))(e_i) \tag{def.\ $\kappa$}  \\
&= \rd^\rho( \kappa_X(f(e_1)), \dots, \kappa_X(f(e_n)) )(e_i) \,.  \tag{def.\ $\rho$}
\end{align*}
From the above we have that, for all $f \colon X^{\underline{E}} \to X$,
\begin{align*}
h(f) 
&= h(\rd^\rho( \kappa_X(f(e_1)), \dots, \kappa_X(f(e_n)) )) \\
&= \rd^a((h \circ \kappa)(f(e_1)), \dots, (h \circ \kappa)(f(e_1))) \tag{$h$ homo} \\
&= \rd^a(\beta(f(e_1)), \dots, \beta(f(e_1))) \tag{$h \circ \kappa = \beta$} \\
&= \rd^a((g \circ \kappa)(f(e_1)), \dots, (g \circ \kappa)(f(e_1))) \tag{$g \circ \kappa = \beta$} \\
&= g(\rd^\rho( \kappa_X(f(e_1)), \dots, \kappa_X(f(e_n)) )) \tag{$g$ homo} \\
&= g(f)
\end{align*}
Therefore, $g = h$.

By the proof of freeness above, the functors $(-)^{\underline{E}}$ and $T_{\R[]}$ are
isomorphic and the units of the two monads coincide (up to iso).  We are
left to prove that also the multiplications coincide (up to iso).  By the universal property of free algebras, 
this follows by showing that the following diagram commutes
\begin{equation*}
\begin{tikzcd}
X^{\underline{E}} \arrow[rd, "id"' ] \arrow[r, "\kappa_{X^{\underline{E}}}"] & 
(X^{\underline{E}})^{\underline{E}} \arrow[d, dashed, "\zeta_X"]
& \Sigma_{\R[]} (X^{\underline{E}})^{\underline{E}} \arrow[l, "\rho_{X^{\underline{E}}}"'] \arrow[d,"\Sigma_{\R[]}\zeta_X"] \\
 & X^{\underline{E}} & \Sigma_{\R[]}X^{\underline{E}} \arrow[l, "\rho_X"']
\end{tikzcd}
\end{equation*}
$\zeta_X \circ \kappa_X = id$ holds since $(-)^{\underline{E}}$ is a monad.  
The right square commutes as shown below
\begin{align*}
(\zeta_X \circ \rho_{X^{\underline{E}}})(\rd(g_1, \dots, g_n))(e_i) 
&= \rho_{X^{\underline{E}}}(\rd(g_1, \dots, g_n))(e_i)(e_i)    \tag{def.\ $\zeta$} \\
&= g_i(e_i)(e_i)  \tag{def.\ $\rho$} \\
&= \zeta_X(g_i)(e_i)   \tag{def.\ $\zeta_X$} \\
&= \rho_{X}(\rd(\zeta_X(g_1), \dots, \zeta_X(g_n)))(e_i)   \tag{def.\ $\rho$} \\
&= (\rho_{X} \circ \Sigma_{\R[]} \zeta_X)(\rd(g_1, \dots, g_n))(e_i)  \tag{def.\ $\Sigma_{\R[]}$}
\end{align*}
for arbitrary $g_1, \dots, g_n \colon \underline{E} \to X^{\underline{E}}$.
\end{proof}

Note that, the monad $(-)^{\underline{E}}$ is well defined also in $\CMet$. Indeed, 
as the functor $(-)^{\underline{E}}$ is isomorphic to the finite product $(-)^n$, for $n = |E|$,
it preserves Cauchy completeness and can be restricted to an endofunctor on $\CMet$.
We further observe that it is isomorphic to the composite 
\begin{equation*}
  \CMet \hookrightarrow \Met \xrightarrow{(-)^{\underline{E}}} \Met \xrightarrow{\CC} \CMet \,.
\end{equation*} 

Since $\R[]$ is a continuous quantitative theory, by Theorems~\ref{th:freeCQAlgebra} and 
\ref{th:isoReaderMonad} we obtain also the following isomorphism of monads.
\begin{thm} \label{th:isoReaderMonadCompletion}
The monads $\CC T_{\R[]}$ and $(-)^{{\underline{E}}}$ in $\CMet$ are isomorphic.  
\end{thm}

Consequently, by Theorem~\ref{th:CompleteEilenbergMoore}, $\EMAlg{(-)^{\underline{E}}} \iso \CC\KK[\Sigma_{\R[]}]{\R[]}$, meaning that $(\Sigma_{\R[]}, \R[])$ is a quantitative equational 
presentation also of the monad $(-)^{\underline{E}}$ on $\CMet$.

\subsubsection{Quantitative Writer Algebras}
\label{sec:writeralgebras}

Fix an extended metric space $\Lambda \in \Met$ of \emph{output values} having monoid 
structure $(\Lambda, *, 0)$ with multiplication $* \colon \Lambda \times \Lambda \to \Lambda$ satisfying~ \eqref{*nonexp}.

The \emph{quantitative writer algebras} of type $\Lambda$ are the
algebras for the signature 
\begin{equation*}
  \Sigma_{\Wr[\Lambda]} = \ens{ \wrt{\alpha} \colon 1 \mid \alpha \in \Lambda }
\end{equation*}
having a unary operator $\wrt{\alpha}$, for each output value $\alpha \in \Lambda$, 
and satisfying the following axioms
\begin{align*}
{ (\Zero)} \,\,
& \vdash x =_0 \wrt{0}(x) \,, \\
{ (\Mult)} \,\,
& \vdash \wrt{\alpha}( \wrt{\alpha'}(x) ) =_0  \wrt{\alpha * \alpha'}(x) \,, \\
{ (\Diff)} \,\,
& \{x =_\e x'\} \vdash \wrt{\alpha}(x) =_\delta  \wrt{\alpha'}(x')\,, \text{ for $\delta \geq d_\Lambda(\alpha, \alpha') + \e $} \,.
\end{align*}
The quantitative theory induced by the axioms above, written $\Wr[\Lambda]$ (or simply $\Wr[]$, when $\Lambda$ is clear), is called \emph{quantitative theory of writing computations}.

The term $\wrt{\alpha}(t)$ represents the computation that proceeds as $t$ after 
writing $\alpha$ on the output tape.  
The axiom (\Zero) says that writing the identity element $0$ is not observable on the tape;
(\Mult) says that consecutive writing operations are stored in the tape in the order of execution;
(\Diff) compares two computations w.r.t.\ the distance of their output values.

\medskip
Next we show that the writer monad $(\Lambda \mprod -)$ is isomorphic to the free monad $T_{\Wr[]}$ on $\Wr[]$. Consequently, as the theory $\Wr[]$ is basic, by Theorem~\ref{th:EilenbergMoore}, $\EMAlg{(\Lambda \mprod -)} \iso \KK[\Sigma_{\Wr[]}]{\Wr[]}$.
Thus, $(\Sigma_{\Wr[]}, \Wr[])$ is a quantitative equational presentation of the writer monad on $\Met$.
\begin{thm} \label{th:isoWriterMonad}
The monads $T_{\Wr[]}$ and $(\Lambda \mprod -)$ in $\Met$ are isomorphic.  
\end{thm}
\begin{proof}
We show that, for each $X \in \Met$, $(\Lambda \mprod X)$ carries a quantitative
algebraic structure that is free in $\KK{\Sigma_{\Wr[]}, \Wr[]}$ with universal arrow $\tau_X$;
and show that the units and multiplications of the two monads coincide (up-to iso).

For any $X \in \Met$, we define the quantitative 
$\Sigma_{\Wr[]}$-algebra $(\Lambda \mprod X, \omega_X)$ as follows, for arbitrary 
$\alpha,\alpha' \in \Lambda$ and $x \in X$
\begin{align*}
  \omega_X \colon \Sigma_{\Wr[]} (\Lambda \mprod X) \to \Lambda \mprod X \,,
  &&
  \omega_X(\wrt{\alpha}(\alpha', x)) = (\alpha * \alpha', x) \,.
\end{align*}

Next we show that it satisfies $\Wr[]$. 
Let $\wrt{\alpha}^{\omega}$ denote the interpretation
of the operation $\wrt{\alpha} \colon 1 \in \Sigma_{\Wr[]}$ in the algebra $(\Lambda \mprod X, \omega_X)$.
Proving the soundness for ($\Nexp{\wrt{\alpha}}$), for each $\alpha \in \Lambda$, is equivalent to showing 
that the map $\omega$ is well defined in $\Met$.
\begin{align*}
d_{(\Lambda \mprod  X)}(\wrt{\alpha}^{\omega}(\beta, x), \wrt{\alpha}^{\omega}(\beta', x' )) 
&=d_{(\Lambda \mprod  X)}((\alpha*\beta, x), (\alpha * \beta', x' )) \tag{def.\ $\omega$} \\
&=d_\Lambda(\alpha*\beta, \alpha * \beta') + d_X(x, x') \tag{def.\ $\mprod$} \\
&\leq d_\Lambda(\alpha, \alpha) + d_\Lambda(\beta, \beta') + d_X(x, x')
\tag{\ref{*nonexp}} \\
&= d_\Lambda(\beta, \beta') + d_X(x, x') \tag{metric} \\
&=d_{(\Lambda \mprod  X)}((\beta, x), (\beta', x')) \,.  \tag{def.\ $\mprod$}
\end{align*} 
We are left to prove that the algebra $((\Lambda \mprod X), \omega_X)$ satisfies
the axioms (\Zero), (\Mult), and (\Diff).  

Soundness for (\Zero) holds trivially as $(\alpha, x) = (0*\alpha,x)$
because $0$ is the identity element of the monoid $\Lambda$.
The soundness of (\Mult) follows directly by definition of $\omega$ as 
\begin{equation*}
\wrt{\alpha}^{\omega}( \wrt{\alpha'}^{\omega} (\beta, x)) = 
\wrt{\alpha}^{\omega}((\alpha' * \beta, x)) =
((\alpha * (\alpha' * \beta), x)) =
((\alpha * \alpha') * \beta, x)) =
\wrt{\alpha * \alpha'}^{\omega}((\beta, x))
\,.
\end{equation*}
Finally, the soundness for (\Diff) follows by
\begin{align*}
d_{(\Lambda \mprod  X)}(\wrt{\alpha}^{\omega}(\beta, x), \wrt{\alpha'}^{\omega}(\beta', x' )) 
&=d_\Lambda(\alpha*\beta, \alpha' * \beta') + d_X(x, x') \tag{def.\ $\omega$ \& $\mprod$} \\
%&=d_\Lambda(\alpha*\beta, \alpha * \beta') + d_\Lambda(\alpha*\beta', \alpha' * \beta') + d_X(x, x') 
%\tag{triang.} \\
&\leq d_\Lambda(\beta, \beta') + d_\Lambda(\alpha, \alpha') + d_X(x, x') 
\tag{\ref{*nonexp}} \\
&= d_{\Lambda}(\alpha, \alpha') + d_{(\Lambda \mprod  X)}((\beta,x), (\beta', x')) \,,
\tag{def.\ $\mprod$}
\end{align*}

Now we prove freeness. 
Let $(A,a)$ be a $\Sigma_{\Wr[]}$-algebra satisfying $\Wr[]$ and $b \colon X \to A$
a non-expansive map. We define $h \colon \Lambda \mprod X \to A$ as follows, 
for arbitrary $\alpha \in \Lambda$ and $x \in X$:
\begin{equation*}
 h((\alpha,x)) = a({\wrt{\alpha}}(b(x))) \,.
\end{equation*}
Non-expansiveness of $h$ follows by the fact that $(A, a)$ satisfies
the axiom (\Diff) as shown below, where $\wrt{\alpha}^{a}$ denotes 
the interpretation of $\wrt{\alpha} \colon 1 \in \Sigma_{\Wr[]}$ in $(A,a)$,
\begin{align*}
d_A(h((\alpha,x)), h((\alpha',x))) 
&= d_A(\wrt{\alpha}^{a}(b(x)), \wrt{\alpha'}^{a}(b(x'))) \tag{def.\ $h$} \\
&\leq d_\Lambda(\alpha, \alpha') + d_A(b(x),b(x'))  \tag{\Diff} \\
&\leq d_\Lambda(\alpha, \alpha') + d_X(x,x')  \tag{$b$ non-exp} \\
&= d_{\Lambda \mprod X}((\alpha,x),(\alpha',x')) \,.  \tag{def.\ $\mprod$}
\end{align*} 
Next, we prove that $h$ is the unique map such that the diagrams below commute.
\begin{equation*}
\begin{tikzcd}
X \arrow[rd, "b"' ] \arrow[r, "\tau_X"] & 
\Lambda \mprod X \arrow[d, dashed, "h"]
& \Sigma_{\Wr[]} (\Lambda \mprod X) \arrow[l, "\omega_X"'] \arrow[d,"\Sigma_{\Wr[]}h"] \\
 & A & \Sigma_{\Wr[]}A \arrow[l, "a"']
\end{tikzcd}
\end{equation*}
The triangle to the left commutes because for all $x \in X$
\begin{align*}
(h \circ \tau_X )(x) 
&= h((0,x)) \tag{def.\ $\tau$} \\
&= \wrt{0}^{a}(b(x)) \tag{def.\ $h$} \\
&= b(x) \,.  \tag{\Zero}
\end{align*}
Let $x \in X$ and $\alpha, \alpha' \in \Lambda$.  Then, 
\begin{align*}
(h  \circ \omega_X)({\wrt{\alpha}}(\alpha',x)) 
&= \wrt{\alpha * \alpha'}^{a}(b(x)) \tag{def.\ $h$ and $\omega$} \\
&= \wrt{\alpha}^{a}( \wrt{\alpha'}^{a}(b(x)) ) \tag{\Mult} \\
&= \wrt{\alpha}^{a}( h(\alpha', x) ) \tag{def.\ $h$} \\
&= (a \circ \Sigma_{\Wr[]}h)({\wrt{\alpha}}(\alpha',x))  
\tag{def.\ $\wrt{\alpha}^{a}$ and $\Sigma_{\Wr[]}$}\,.
\end{align*}
Thus, $h$ is a $\Sigma_{\Wr[]}$-homomorphism, \ie, $h \circ \omega_X = a \circ \Sigma_{\Wr[]}h$.

It remains to show uniqueness. 
Notice first that, $(\alpha, x) = \wrt{\alpha}^{\omega}(\tau(x))$, where $\wrt{\alpha}^{\omega} = \omega_X \circ {\wrt{\alpha}}$ denotes the interpretation of $\wrt{\alpha} \colon 1 \in \Sigma_{\Wr[]}$ in $(\Lambda \mprod X, \omega_X)$.
Indeed, the following holds
\begin{align*}
(\alpha, x) 
&= (\alpha * 0, x) \tag{$0$ identity} \\
&= \wrt{\alpha}^{\omega}(0,x) \tag{def.\ $\omega$} \\
&= \wrt{\alpha}^{\omega}(\tau(x)) \,.  \tag{def.\ $\tau$}
\end{align*}
Assume there exists $g \colon \Lambda \mprod X \to A$ such that 
$g \circ \tau_X = b$ and
$g \circ \omega_X = a \circ \Sigma_{\Wr[]}g$.  Then,
\begin{align*}
h((\alpha, x))
&= h(\wrt{\alpha}^{\omega}(\tau(x))) \\
&= \wrt{\alpha}^{a}(h(\tau(x))) \tag{$h$ homo} \\
&= \wrt{\alpha}^{a}(b(x)) \tag{$h \circ \tau = b$} \\
&= \wrt{\alpha}^{a}(g(\tau(x)))  \tag{$g \circ \tau = b$} \\
&= g(\wrt{\alpha}^{\omega}(\tau(x))) \tag{$g$ homo} \\
&= g((\alpha, x))
\end{align*}
Therefore, $h = g$.

By proof of freeness above, the functors $(\Lambda \mprod -)$ and $T_{\Wr[]}$ are
isomorphic and the units of the two monads coincide (up-to iso).  We are
left to prove that also the multiplications coincide (up-to iso).  By the universal property of free algebras, 
this follows by showing that the following diagram commutes
\begin{equation*}
\begin{tikzcd}
(\Lambda \mprod X) \arrow[rd, "id"' ] \arrow[r, "\tau_{\Lambda \mprod X}"] & 
(\Lambda \mprod (\Lambda \mprod X)) \arrow[d, dashed, "\varsigma_X"]
& \Sigma_{\Wr[]} (\Lambda \mprod (\Lambda \mprod X)) \arrow[l, "\omega_{\Lambda \mprod X}"'] \arrow[d,"\Sigma_{\Wr[]}\varsigma_X"] \\
 & \Lambda \mprod X & \Sigma_{\Wr[]}(\Lambda \mprod X) \arrow[l, "\omega_X"']
\end{tikzcd}
\end{equation*}
The triangle to the left holds as $(\Lambda \mprod -)$ is a monad.  
The right square commutes by
\begin{align*}
(\varsigma_X \circ \omega_{\Lambda \mprod X})({\wrt{\alpha}}(\alpha', (\alpha'', x)) ) 
&= \varsigma_X((\alpha * \alpha', (\alpha'', x))) \tag{def.\ $\omega$} \\
&= (\alpha * \alpha' * \alpha'', x) \tag{def.\ $\varsigma$} \\
&= \omega_X({\wrt{\alpha}}(\alpha' * \alpha'', x)) \tag{def.\ $\omega_X$} \\
&= \omega_X({\wrt{\alpha}}( \varsigma_X( \alpha', (\alpha'', x)) ) ) \tag{def.\ $\varsigma$} \\
&= (\omega_X \circ \Sigma_{\Wr[]} \varsigma_X)({\wrt{\alpha}}(\alpha', (\alpha'', x)) ) 
\tag{def.\ $\Sigma_{\Wr[]}$}
\end{align*}
for arbitrary $x \in X$ and $\alpha, \alpha', \alpha'' \in \Lambda$.
\end{proof}

If we assume the carrier of a monoid $(\Lambda, *, 0)$ to be a complete extended metric space $\Lambda$, the writer monad $(\Lambda \mprod -)$ 
is well defined also in $\CMet$.
We further observe that, as $\mprod$ preserves completeness, the underlying functor $(\Lambda \mprod -)$ 
is isomorphic to the composite 
\begin{equation*}
  \CMet \hookrightarrow \Met \xrightarrow{(\Lambda \mprod -)} \Met \xrightarrow{\CC} \CMet \,.
\end{equation*} 

Since $\Wr[]$ is a continuous quantitative theory, by Theorems~\ref{th:freeCQAlgebra} and 
\ref{th:isoWriterMonad} we obtain also the following isomorphism of monads.
\begin{thm} \label{th:isoWriterMonadCompletion}
The monads $\CC T_{\Wr[]}$ and $(\Lambda \mprod -)$ in $\CMet$ are isomorphic.  
\end{thm}

Thus, by Theorem~\ref{th:CompleteEilenbergMoore}, $(\Sigma_{\Wr[]}, \Wr[])$ is a quantitative equational 
presentation also of the monad $(\Lambda \mprod -)$ on $\CMet$.

%%%%%%%%%%%%%%%%%%%%%%%%%%%%%%%%%%%%%%%%%%%%%%
\subsection{Nondeterminism}
\label{sec:Hausdorff}
The monadic effect for nondeterminism in $\Set$ is given by the \emph{powerset monad}, denoted by $\mathcal{P}$, mapping a set $X$ to $\mathcal{P} X = \{ U \mid U \subseteq X \}$ and a function $f \colon X \to Y$ to $\mathcal{P} f(U) = \{ f(u) \mid u \in U \}$, for $U \in \mathcal{P} X$. 
The unit $\sigma \colon Id \nat \mathcal{P}$ and multiplication $\upsilon \colon \mathcal{P}\mathcal{P} \nat \mathcal{P}$ are given, for $x \in X$ and $S \in \mathcal{P}\mathcal{P} X$, by
\begin{equation}
\begin{aligned}
\sigma_X(x) &= \{ x \} && \textsc{unit} 
\\
\upsilon_X(S) &= \bigcup \{ U \mid U \in S \} && \textsc{multiplication}
\end{aligned}
\label{eq:PowesetUnitMult}
\end{equation}
%Another monads of interest are the submonads $\mathcal{P}_{f}$, of finite subsets, and $\mathcal{P}_{ne}$, 
%of non-empty finite subsets, that act on functions just like the powerset monad, and have unit and multiplication 
%defined as in~\eqref{eq:PowesetUnitMult}.

\medskip
In the quantitative setting, a natural candidate for a distance on subsets is the Hausdorff metric. Typically, the Hausdorff metric is defined on nonempty compact subsets (equivalently, nonempty closed bounded subsets). 
Nonemptyness is meant to ensure that the distance is always finite. However, as we are dealing with extended metric spaces the empty set need not be excluded.

The \emph{Hausdorff extended metric} on the set of all compact subsets of an extended metric space $X$ is defined, for arbitrary closed sets $U,V \subseteq X$ by
\begin{equation*}
\mathcal{H}(d_X)(U,V) = \max \left\{ \sup_{u \in U} d_X(u,V), \sup_{v \in V} d_X(v,U) \right\} \,,
\end{equation*}
where, $d_X(x,U) = \inf_{u \in U} d_X(x,U)$ denotes the distance from an element $x \in X$ to a set 
$U \subseteq X$ (we assume $\inf \emptyset = \infty$).

\smallskip
In this paper, we will consider monads for quantitative nondeterminism both in $\Met$ and $\CMet$.

On $\Met$, the \emph{finite (quantitative) powerset monad}, denoted by $\mathcal{P}_f$, maps an extended metric space $X$ to $\mathcal{P}_f X = \{ U \mid U \subseteq X, U \text{ finite} \}$ with Hausdorff metric (note that finite sets are compact) and acts on morphisms $f \colon X \to Y$ as $\mathcal{P}_f f (U) = \{ f(u) \mid u \in U \}$, for $U \in \mathcal{P}_f X$. The unit and multiplication are defined as in~\eqref{eq:PowesetUnitMult}. 
Another monad of interest on $\Met$ is the submonad $\mathcal{P}_{ne}$, 
of non-empty finite subsets, with the same unit and multiplication.

On $\CMet$ the \emph{compact subsets monad}, denoted by $\C$, maps a compacts extended metric space 
$X$ to $\C X = \{ U \mid U \subseteq X, U \text{ compact} \}$ with Hausdorff metric (as shown in~\cite[Lemma 3]{kuratowski1956m}, if the metric space $X$ is complete, then so is the metric space $\C X$), and acts on morphisms $f \colon X \to Y$ as $\C f (U) = \{ f(u) \mid u \in U \}$, for $U \in \C X$. 
The unit and multiplication are defined as in~\eqref{eq:PowesetUnitMult}. The fact that this is a well defined monad results from~\cite[Theorem 15]{BreugelHMW07}, which establishes the adjunction from which this monad is derived.

\begin{lem} \label{lem:complHausdorff}
The monads $\CC\mathcal{P}_f$ and $\C$ on $\CMet$ are isomorphic.
\end{lem}
\begin{proof}
Let $X \in \CMet$. Clearly, $\mathcal{P}_f X \subseteq \C X$, as finite subsets are compact.
Next, we show that $\mathcal{P}_f X$ is dense in $\C X$. Let $U \in \C X$ and $\epsilon > 0$.
Let $B_\epsilon(u) = \{ x \in X \mid d_X(u, x) < \epsilon \}$ be the open ball of radius $\epsilon$
centered in $u \in U$. As $\{ B_\epsilon(u) \mid u \in U \}$ is an open cover for $U$, by compactness of $U$,
there exists a finite subcover $\{ B_\epsilon(v) \mid v \in V \}$ for some $V = \{ v_0, \dots, v_n \} \subseteq U$.
In particular, we have that for any $u \in U$, exists $v \in V$ such that $d_X(u,v) < \epsilon$ 
(equivalently, $\sup_{u \in U} d_X(u,V) < \epsilon$). Thus
\begin{align*}
\mathcal{H}(d_X)(U,V) 
&= \max \left\{ \sup_{u \in U} d_X(u,V), \sup_{v \in V} d_X(v,U) \right\} \\
&= \sup_{u \in U} d_X(u,V) \tag{by $V \subseteq U$} \\
&< \epsilon
\end{align*}
So $\mathcal{P}_f X$ is dense in $\C X$. As $X$ is complete, convergence and Cauchy convergence coincide. Thus $\overline{\mathcal{P}_f X} \iso \C X$.
The correspondence between the units is trivial. The correspondence between multiplications follows because if 
$(V_i)$ is a sequence of finite subsets of $\C X$ converging to $S \in \C\C X$, then $(\bigcup \{ U \mid U \in V_i \})$ converges to $\bigcup \{ U \mid U \in S \}$.
\end{proof}

\subsubsection{Quantitative Join-Semilattice with Bottom}
\label{sec:semillattices}
In~\cite{MardarePP:LICS16} it was shown that the quantitative powerset monads considered above 
have an algebraic presentation in terms of a simple quantitative extension to the equational theory of
join-semilattices with bottom.

A quantitative join-semilattice with bottom is a quantitative algebra for the signature  
\begin{equation*}
  \Sigma_{\Semi} = \{ + \colon 2 , \zero \colon 0 \}
\end{equation*}
with a binary operator $+$ and a constant $\zero$ that satisfy the quantitative equations
\begin{align*}
{ (\Szero)} 
& \vdash x + \zero =_0 x \,, \\
{ (\Sone)} 
& \vdash x + x =_0 x \,, \\
{ (\Stwo)} 
& \vdash x + y =_0 y + x \,, \\
{ (\Sthree)} 
& \vdash (x + y) + z =_0 x + (y + z) \,, \\ 
{ (\Sfour)} \,
& \{ x \,{=_\e}\, y, x' \,{=_{\e'}}\, y' \} \,{\vdash}\, x + x' \,{=_{\max\{\e, \e' \}}}\, y + y'  \,.
\end{align*}
We denote by $\Semi$ the above quantitative theory \emph{of semilattices with bottom}.
The axioms (\Szero), (\Sone), (\Stwo), (\Sthree) are those of \emph{(join-)semilattices with bottom} and they are
essentially standard ``equational'' axioms. The truly quantitative equation is the last one, 
(\Sfour). 

\begin{rem}
Note that (\Sfour) is derivable from (\Nexp{+}) and (\Weak) and, conversely, (\Nexp{+}) is just an instance 
of (\Sfour) (when $\e = \e'$). Thus, in the quantitative theory $\Semi$, the axiom (\Sfour) is not necessary. We added it here to match the definition presented in~\cite{MardarePP:LICS16}.
\end{rem}
 
\medskip
For any $X \in \Met$, one can define a quantitative 
$\Sigma_{\Semi}$-algebra $(\mathcal{P}_f X, \phi_X)$ as follows, for arbitrary $U, V \in \mathcal{P}_f X$
\begin{align*}
  \phi_X \colon \Sigma_{\Semi}\mathcal{P}_f X \to \mathcal{P}_f X
  &&
  \begin{aligned}
  \phi_X(U+V) &= U \cup V \,, \\
  \phi_X({\zero}) &= \emptyset \,.
  \end{aligned}
\end{align*}
This quantitative algebra satisfies the quantitative theory $\Semi$, 
(\cf~\cite[Theorem~9.2]{MardarePP:LICS16}) and it is isomorphic to the 
free quantitative $\Sigma_{\Semi}$-algebra on $\Semi$ (\cf~\cite[Theorem~9.3]{MardarePP:LICS16}).

Thus, as shown in~\cite{MardarePP:LICS16}, $\mathcal{P}_f$ is isomorphic to the free monad $T_{\Semi}$ 
on the theory of quantitative join-semilattices with bottom.  
\begin{thm} \label{th:isofSemilatticesMonad}
The monads $T_{\Semi}$ and $\mathcal{P}_f$ on $\Met$ are isomorphic.  
\end{thm}

As a direct consequence of Lemma~\ref{lem:complHausdorff} and Theorem~\ref{th:isofSemilatticesMonad}
we obtain the following result, which, in combination with Theorem~\ref{th:CompleteEilenbergMoore}, tells us that  
$(\Sigma_{\Semi}, \Semi)$ is an algebraic presentation of the compact subsets monad $\C$ on complete metric spaces.
\begin{thm} \label{th:isoCompactSubsets}
The monads $\CC T_{\Semi}$ and $\C$ on $\CMet$ are isomorphic.  
\end{thm}

%%%%%%%%%%%%%%%%%%%%%%%%%%%%%%%%%%%%%%%%%%%%%%
\subsection{Probabilistic choice}
\label{sec:probchoice}
The monadic effect describing probabilistic choice (a.k.a., probabilistic nondeterminism) in $\Set$ is given by the
\emph{finitely supported distribution monad}, denoted by $\mathcal{D}$. This monad acts on sets $X$ as 
\begin{equation*}
\mathcal{D}X = \Big\{ p \colon X \to [0,1] \mid \sum_{x \in X} p(x) = 1, \supp{p} \text{ is finite} \Big\} \,.
\end{equation*}
\ie, the set of probability distributions $p$ with finite support $\supp{p} = \{ x \mid p(x) \neq 0 \}$ over $X$,
and acting on morphisms $f \colon  X \to Y$ as $\mathcal{D}(f)(p) = \sum_{x \in f^{-1}(y)} p(x)$, for 
$p \in \mathcal{D}X$.
The unit $\delta \colon Id \nat \mathcal{D}$ is given by the Dirac distribution function $\delta_X(x) = (x \mapsto 1)$, for $x \in X$, and the multiplication $m \colon \mathcal{D}\mathcal{D} \nat \mathcal{D}$ by 
$m_X(P)(x) = \sum_{p \in \supp{P}} P(p) \cdot p(x)$, for $P \in \mathcal{D}\mathcal{D}X$. 

\medskip
In the quantitative setting, given a metric space, one can more generally consider the set of Borel probability measures over it (those defined on the Borel $\sigma$-algebra induced by the metric).
There are several ways of measuring the ``difference'' between probability measures, \eg, using the Total Variation distance, Hellinger distance, Kullback--Leibler divergence, Jensen--Shannon divergence, etc. Here, however, we focus on one specific notion of distance: the \emph{Kantorovich metric}~\cite{Kantorovich42} (a.k.a., Wasserstein-1 metric, or Earth Mover's Distance). This distance has applications in optimisation and measure theory, as it is related to the concept
of transportation problem~\cite{Villani08} and metrizes weak-convergence of probability measures~\cite{Billingsley:convergence}.

Formally, the Kantorovich distance is defined on Radon measures of finite moment, but as we are dealing with distances that may take infinite values, we won't require the latter condition and consider instead integration on nonnegative extended real-valued measurable functions (\cf~\cite{Bar1995} for the formal definition
of the Lebesgue integration of extended real-valued functions).
In detail, a Borel probability measure $\mu$ on an extended metric space $X$ is \emph{Radon} if for 
any Borel set $E \subseteq X$, $\mu(E)$ is the supremum of $\mu(K)$ over all compact subsets 
$K$ of $E$. 

The \emph{Kantorovich extended metric} between Radon probability measures $\mu,\nu$ over 
an extended metric space $X$ is then defined as
\begin{equation*}
\K[d_X]{\mu,\nu} = \min_{\omega}  \lebint{d_X}{\omega} \,,
\end{equation*}
where $\omega$ runs over the set of all joint probability measures on $X \times X$ whose left and right 
marginals (= pushforwards along the projections) are, respectively, $\mu$ and $\nu$.

Examples of Radon probability measures are: (i) finitely supported
probability measures on any (extended) metric space, and (ii) generic Borel probability
measures over complete separable (extended) metric spaces.

\medskip
In this paper, we consider two distinct quantitative monads for probabilistic nondeterminism: one on $\Met$ and  one on $\CMet$.

On $\Met$, the \emph{finitely supported probability monad}, denoted by $\Pi$, assigns 
to an extended metric space $X$ the space $\Pi X$ of finitely supported Borel
probability measures with Kantorovich metric; and acts on morphisms $f \colon X \to Y$  
as $\Pi(f)(\mu) = \mu \circ f^{-1}$ (the pushforward of $f$), for any $\mu \in \Pi(X)$.
The unit $\delta \colon Id \Rightarrow \Pi$ and multiplication $m \colon \Pi\Pi \Rightarrow \Pi$, 
are given as follows, for $x \in X$, $\Phi \in \Pi \Pi X $, and  Borel subset $E \subseteq X$
\begin{equation}
\begin{aligned}
\delta_X(x) &= \delta_x \,, && \textsc{unit} 
\\
m_X(\Phi)(E) &= \lebint{v_E}{\Phi} && \textsc{multiplication} \,,
\end{aligned}
\label{eq:KantorovichUnitMult}
\end{equation}
where $\delta_x$ is the Dirac delta measure at $x$, and
$v_E \colon \Pi X \to [0,1]$ is the evaluation function, taking
$\mu \in \Pi X $ to $\mu(E) \in [0,1]$.

\smallskip
On $\CMet$, the \emph{Radon probability monad}, denoted by $\Delta$, maps a complete
extended metric space $X$ to the (complete) extended metric space $\Delta X$ of 
Radon probability measures with Kantorovich metric; and 
acts on morphisms $f \colon X \to Y$ as $\Delta(f)(\mu) = \mu \circ f^{-1}$, for
$\mu \in \Delta(X)$.  
The unit $\delta \colon Id \Rightarrow \Delta$ and multiplication $m \colon \Delta\Delta \Rightarrow \Delta$, 
are defined as in~\eqref{eq:KantorovichUnitMult}.

\smallskip
These two monads are related as follows
\begin{lem} \label{lem:complGiry}
The monads $\CC\Pi$ and $\Delta$ on $\CMet$ are isomorphic.
\end{lem}

\subsubsection{Interpolative Barycentric Algebras}
\label{sec:barycentric}
In~\cite{MardarePP:LICS16} it was shown that the quantitative probability monads considered above have an algebraic presentation in terms of a quantitative extension of barycentric algebras, which they called \emph{interpolative barycentric algebras}.

%We recall interpolative barycentric algebras from 
%and prove that the free monad induced by the interpolative barycentric quantitative equational 
%theory is isomorphic to the finite distribution monad.
Interpolative barycentric algebras are the quantitative algebras for the signature  
\begin{equation*}
  \Sigma_{\B} = \set{ +_e \colon 2 }{ e \in [0,1]}
\end{equation*}
with a binary operator $+_e$, for each $e \in [0,1]$ (a.k.a.\
\emph{barycentric signature}), and satisfying the quantitative equations
\begin{align*}
{(\Bone)} 
& \vdash x +_1 y =_0 x \,, \\
{ (\Btwo)} 
& \vdash x +_e x =_0 x \,, \\
{ (\SC)} 
& \vdash x +_e y =_0 y +_{1-e} x \,, \\
{ (\SA)} 
& \vdash (x +_e y) +_{e'} z =_0 x +_{ee'} (y +_{\frac{e' - ee'}{1 - ee'}} z) \,, 
\text{ for $e,e' \in [0,1)$} \,, \\ 
{ (\IB)} \,
& \{ x \,{=_\e}\, y, x' \,{=_{\e'}}\, y' \} \,{\vdash}\, x +_e x' \,{=_{\delta}}\, y +_e y', \, 
\text{ for $\delta \geq e \e + (1-e) \e'$} \,.
\end{align*}
The quantitative theory axiomatized by the quantitative equations above, written $\B$, is 
called \emph{interpolative barycentric quantitative equational theory}.
The axioms (\Bone), (\Btwo), (\SC), (\SA) are those of \emph{barycentric algebras} (a.k.a.\ 
\emph{abstract convex sets}) due to M.\ H.\ Stone~\cite{Stone49} where
(\SC) stands for \emph{skew commutativity} and (\SA) for \emph{skew
  associativity}; (\IB) is the \emph{interpolative barycentric axiom}
introduced in~\cite{MardarePP:LICS16}.  

\medskip
For any $X \in \Met$, one can define a quantitative 
$\Sigma_{\B}$-algebra $(\Pi X, \phi_X)$ as follows, for arbitrary $\mu, \nu \in \Pi X$
\begin{align*}
  \phi_X \colon \Sigma_{\B}\Pi X \to \Pi X
  &&
  \phi_X(\mu +_e \nu) = e  \mu + (1-e)  \nu \,.
\end{align*}
This quantitative algebra satisfies the interpolative barycentric theory
$\B$ (\cf~\cite[Theorem~10.4]{MardarePP:LICS16}) and is isomorphic to the 
free quantitative $\Sigma_{\B}$-algebra on $\B$ (\cf~\cite[Theorem~10.5]{MardarePP:LICS16}).

Thus, as shown in~\cite{MardarePP:LICS16}, $\Pi$ is isomorphic to the free monad $T_{\B}$ 
on the theory $\B$ of interpolative barycentric algebras.  
\begin{thm} \label{th:isofBarycentricMonad}
The monads $T_{\B}$ and $\Pi$ on $\Met$ are isomorphic.  
\end{thm}

As a direct consequence of Lemma~\ref{lem:complGiry} and Theorem~\ref{th:isofBarycentricMonad}
we obtain the following result, which, in combination with Theorem~\ref{th:CompleteEilenbergMoore}, tells us that  
$(\Sigma_{\B}, \B)$ is an algebraic presentation of the Radon probability monad $\Delta$ on complete metric spaces.
\begin{thm} \label{th:isoGiry}
The monads $\CC T_{\B}$ and $\Delta$ on $\CMet$ are isomorphic.  
\end{thm}

%%%%%%%%%%%%%%%%%%%%%%%%%%%%%%%%%%%%%%%%%%%%%%
\section{Sum of Quantitative Theories}
\label{sec:sum} 

In this section we develop the theory of the \emph{sum} (or disjoint union) of
quantitative equational theories and show it to correspond to the sum quantitative algebraic effects
whose presentation is given in terms of basic quantitative theories.

Our leading examples of the sum of quantitative effects are given by the combination of 
termination/exceptions with arbitrary quantitative effects; and the combination of 
interactive inputs (more generally, a collection of contractive operators) with arbitrary quantitative effects.
We conclude this section by showing how we can recover the
theory of quantitative Markov processes (\ie, the usual theory of Markov processes but now enriched with metric reasoning principles for the underlying probability measures) in terms of these two generic
combinators of quantitative effects.

\medskip
Let $\Sigma$, $\Sigma'$ be two disjoint signatures.  The
sum of two quantitative theories $\U$, $\U'$ of
respective types 
$\Sigma$ and $\Sigma'$, written $\U + \U'$, is the
smallest quantitative theory containing $\U$ and $\U'$.  Following
Kelly~\cite{Kelly1980}, we show that any model for $\U + \U'$ is a
\emph{$\tupl{\U, \U'}$-bialgebra}: a metric space $A$ with both a
$\Sigma$-algebra structure $\alpha \colon \Sigma A \to A$ satisfying $\U$
and a $\Sigma'$-algebra structure $\beta \colon \Sigma' A \to A$ satisfying
$\U'$.  Formally, let $\KK{(\Sigma, \U) \oplus (\Sigma', \U')}$ be the
category of $\tupl{\U, \U'}$-bialgebras with non-expansive maps 
preserving the two algebraic structures. Then, the following isomorphism
of categories holds.
\begin{prop} \label{prop:bialgebras}
$\KK{\Sigma + \Sigma', \U + \U'} \iso \KK{(\Sigma, \U) \oplus (\Sigma',
  \U')}$.
\end{prop}
\begin{proof}
The isomorphism is given by the following pair of functors 
\begin{equation*}
\begin{tikzcd}
\KK{\Sigma + \Sigma', \U + \U'} \arrow[rr, bend left=10, "H"]
&  & 
\KK{(\Sigma, \U) \oplus (\Sigma', \U')} \arrow[ll, bend left=10, "K"] 
\end{tikzcd}
\end{equation*}
defined, for an arbitrary quantitative $(\Sigma + \Sigma')$-algebra $(A,
\gamma)$ satisfying $\U + \U'$ and a 
$\tupl{\U, \U'}$-bialgebra $(B, \alpha, \beta)$, respectively as
\begin{align*}
H(A, \gamma) = (A, \gamma \circ in_l, \gamma \circ in_r) \,,
&&
K(B, \alpha, \beta) = (B, [\alpha, \beta]) \,, 
\end{align*}
where $[\alpha, \beta]$ is the unique map induced by $\alpha$ 
and $\beta$ by universality of the coproduct $\Sigma A + \Sigma' A$. 
On morphisms both functors map a morphism to itself; it is easy to see that
a homomorphism in one sense is also a homomorphism in the other. 

The fact that the functors are inverses is clear: $H \circ K = Id$ and
$K \circ H = Id$ follow immediately from the universal property of
coproducts.  It remains to show that the functors preserve the relevant quantitative equations, and thus are well defined.

To show that $H$ is well defined we need to prove that whenever
$(A, \gamma)$ satisfies $\U + \U'$, then $(A, \gamma \circ in_l)$ and
$(A, \gamma \circ in_r)$ satisfy $\U$ and $\U'$, respectively.  We will
prove only that $(A, \gamma \circ in_l)$ satisfies $\U$; the other
is similar.  
Let $\Gamma \vdash t =_\e s \in \U$ and $\iota \colon X \to A$ be an arbitrary assignment
of the variables.  Since $(A, \gamma)$ satisfies $\U + \U'$ and 
$\U \subseteq \U + \U'$, we have:
\begin{equation}
\big( \text{for all $t' =_{\e'} s' \in \Gamma$, }
  d_A(\iota^\sharp(t'),\iota^\sharp(s')) \leq \e'  \big)  
\text{ implies }
d_A(\iota^\sharp(t),\iota^\sharp(s)) \leq \e \,,
\label{eq:sat}
\end{equation}
where $\iota^\sharp \colon \TT{\Sigma+\Sigma',X} \to A$ is the homomorphic
extension of $\iota$ on $(A, \gamma)$.

Let $\iota^\flat \colon \TT{\Sigma,X} \to A$ is the homomorphic
extension of $\iota$ on $(A, \gamma \circ \inj_l)$. Next we show that $\iota^\flat = \iota^\sharp \circ i$, where $i$ is the canonical inclusion of $\Sigma$-terms into
$\TT{\Sigma+\Sigma',X}$. We proceed by induction on $\Sigma$-terms.
(Base case) Let $x \in X$. Then 
$(\iota^\sharp \circ i)(x) = \iota^\sharp(x) = \iota(x) = \iota^\flat(x)$.
(Inductive step) Let $f\colon n \in \Sigma$ and $t_1, \dots, t_n \in \TT{\Sigma,X}$. Then
\begin{align*}
    (\iota^\sharp \circ i)(f(t_1, \dots, t_n))
    &= \iota^\sharp(f(t_1, \dots, t_n)) 
    \tag{$f(t_1, \dots, t_n) \in \TT{\Sigma,X}$} \\
    &= (\gamma \circ \inj_l)(f(\iota^\sharp(t_1), \dots, \iota^\sharp(t_n))) 
    \tag{def. $\iota^\sharp$} \\
    &= \gamma(f(\iota^\sharp(t_1), \dots, \iota^\sharp(t_n))) 
    \tag{$f \in \Sigma$} \\
    &= \gamma(f((\iota^\sharp \circ i)(t_1), \dots, (\iota^\sharp \circ i)(t_n)))
    \tag{$f(t_1, \dots, t_n) \in \TT{\Sigma,X}$} \\
    &= \gamma(f(\iota^\flat(t_1), \dots, \iota^\flat(t_n)))
    \tag{ind. hp}\\
    &= \iota^\flat(f(t_1, \dots, t_n)) \,.
    \tag{def. $\iota^\flat$}
\end{align*}
Since $\Gamma \vdash t =_\e s$ has only occurrences of $\Sigma$-terms, \eqref{eq:sat} implies
\begin{equation*}
\big( \text{for all $t' =_{\e'} s' \in \Gamma$, }
  d_A((\iota^\sharp \circ i)(t'),(\iota^\sharp \circ i)(s')) \leq \e'  \big)  
\text{ implies }
d_A((\iota^\sharp \circ i)(t),(\iota^\sharp \circ i)(s)) \leq \e \,.
\end{equation*}
As $\iota^\flat = \iota^\sharp \circ i$, we can conclude that $(A, \gamma \circ in_l)$ satisfies $\Gamma \vdash t =_\e s$.  This argument is general so it applies to the whole theory $\U$.

For $K$, we need to show that whenever $(A, \alpha)$ satisfies $\U$ and
$(A, \beta)$ satisfies $\U'$, then $(A, [\alpha,\beta])$ satisfies
$\U + \U'$.  The argument resembles the one discussed earlier, so we will omit the 
more detailed explanation. By definition of disjoint union of quantitative
theories, it suffices to prove that $(A, [\alpha,\beta])$ is a model
for both $\U$ and $\U'$.  We consider only the former case; the other
is similarl. Let $\Gamma \vdash t =_\e s \in \U$ and
$\iota \colon X \to A$ be an arbitrary assignment of the variables.  Since
$(A, \alpha)$ satisfies $\U$, we have that 
\begin{equation}
\big( \text{for all $t' =_{\e'} s' \in \Gamma$, }
  d_A(\iota^\flat(t'),\iota^\flat(s')) \leq \e'  \big) 
\text{ implies }  
d_A(\iota^\flat(t),\iota^\flat(s)) \leq \e \,,
\label{eq:sat2}
\end{equation}
where $\iota^\flat \colon \TT{\Sigma+\Sigma',X} \to A$ is the homomorphic extension of 
$\iota$ on $(A, [\alpha,\beta])$. 

Let $\iota^\sharp \colon \colon \TT{\Sigma+\Sigma',X} \to A$ be the homomorphic extension of 
$\iota$ on $(A, \alpha)$. By universal property of the coproduct and definition of homomorphic
extension, we can show that $\iota^\sharp \circ i = \iota^\flat$, for $i$ the canonical inclusion of $\Sigma$-terms into $\TT{\Sigma+\Sigma',X}$.
Since $\Gamma \vdash t =_\e s$ contains only terms $\Sigma$-terms, from \eqref{eq:sat2} and $\iota^\sharp \circ i = \iota^\flat$ we get that $(A, [\alpha,\beta])$ satisfies
$\Gamma \vdash t =_\e s$. Again, this implies the result for all of $\U$.
\end{proof}

Let $T$, $T'$ be two monads on a category $\mathbf{C}$. An Eilenberg-Moore bialgebra for
\emph{$\tupl{T, T'}$} (or simply, $\tupl{T, T'}$-bialgebra) is an object $A \in \mathbf{C}$ with
Eilenberg-Moore algebra structures $\alpha \colon T A \to A$ and
$\beta \colon T'A \to A$. We write $\biEMAlg{T}{T'}$ for the category of
Eilenberg-Moore bialgebras for $\tupl{T, T'}$ with morphisms those 
in $\mathbf{C}$ preserving the two algebraic structures.

When the quantitative equational theories $\U$ and $\U'$ are basic, by
Theorem~\ref{th:EilenbergMoore}, we get a refinement of
Proposition~\ref{prop:bialgebras} as follows.
\begin{cor} \label{cor:bialgSimple}
For $\U, \U'$ basic quantitative theories,  $\KK{\Sigma + \Sigma', \U + \U'} \iso \biEMAlg{T_{\U}}{T_{\U'}}$.
\end{cor}
\begin{proof}
Immediate from Theorem~\ref{th:EilenbergMoore} and
Proposition~\ref{prop:bialgebras}. 
\end{proof}

The following result supports the construction of the sum of quantitative theories as 
a combinator of quantitative effects. It states that the free monad $T_{\U + \U'}$ 
on the sum $\U + \U'$ corresponds to the categorical sum (coproduct)
$T_{\U} + T_{\U'}$ of the free monads on $\U$ and $\U'$, respectively. 
This isomorphism of monads stands under the assumption that the sum is taken 
over basic quantitative theories.
\begin{thm} \label{th:sumofsimpleeqmonad}
If $\U, \U'$ are basic quantitative theories, then $T_{\U + \U'}$ is isomorphic to $T_{\U} + T_{\U'}$. 
\end{thm}
\begin{proof}
  By Corollary~\ref{cor:bialgSimple} and Theorem~\ref{th:freeQAlgebra} the
  obvious forgetful functor from $\biEMAlg{T_{\U}}{T_{\U'}}$ to $\Met$ has a
  left adjoint. The monad generated by this adjunction is isomorphic to
  $T_{\U + \U'}$. Thus, by~\cite{Kelly1980} (cf.\ also
  \cite[Proposition~2.8]{AdamekMBL12}), the monad $T_{\U + \U'}$ is 
  isomorphic to $T_{\U} + T_{\U'}$.
\end{proof}

\medskip
The above constructions do not use any specific property of the category
$\Met$, apart from requiring its morphisms to be non-expansive.  Thus, we can reformulate 
an alternative version of Theorem~\ref{th:sumofsimpleeqmonad} which is valid in $\CMet$, under
the assumption that the sum is taken over continuous quantitative theories.

Recall that continuous theories are basic.
Moreover, the disjoint union $\U + \U'$ of two
continuous quantitative theories $\U, \U'$ is also continuous, so that, by
Theorem~\ref{th:freeCQAlgebra}, the free monad on it in $\CMet$ is
$\CC T_{\U + \U'}$. Thus:
\begin{thm} \label{th:sumofsimpleeqmonadComplete}
If $\U, \U'$ are continuous theories, then $\CC T_{\U + \U'}$ is isomorphic to $\CC T_{\U} + \CC T_{\U'}$.
\end{thm}
\begin{proof}
By Theorem~\ref{th:freeCQAlgebra}, the monads $\CC T_{\U + \U'}$,
$\CC T_{\U}$, and $\CC T_{\U'}$ are, respectively, the free monads 
on $\U + \U'$, $\U$, and $\U'$ in $\CMet$. 

Similarly to Corollary~\ref{cor:bialgSimple}, one obtains that
$\CC\KK{\Sigma + \Sigma', \U + \U'}$ and $\biEMAlg{\CC T_{\U}}{\CC T_{\U'}}$
are isomorphic.  Thus, by Theorem~\ref{th:freeCQAlgebra} the
forgetful functor from $\biEMAlg{\CC T_{\U}}{\CC T_{\U'}}$ to $\Met$ has a
left adjoint, and the monad generated by this adjunction is isomorphic to
$\CC T_{\U + \U'}$.  Thus, by~\cite{Kelly1980} (cf.\ also
\cite[Proposition~2.8]{AdamekMBL12}), $\CC T_{\U + \U'}$ 
is the sum of $\CC T_{\U}$ and $\CC T_{\U'}$. 
%isomorphic to the sum of monads $\CC T_{\U} + \CC T_{\U'}$.
\end{proof}

\subsection{Sum with Exceptions}
\label{sec:sumWithException}
As remarked by Hyland, Plotkin, and Power (\cf\ \cite[Corollary~3]{HylandPP06}), Moggi's exception 
monad transformer, sending a monad $T$ to the composite $T(- + E)$ can be explained in terms of 
the sum of monads:
\begin{prop}[Sum with Exception Monad]
Given a category $\cat{C}$ with finite coproducts, an object $E$ of $\cat{C}$, and a monad $T$ on $\cat{C}$, 
the sum of the monads $(- + E)$ and $T$ exists and is given by a canonical monad structure on the composite 
$T(- + E)$.
\end{prop}

From the above result, in combination with Theorems~\ref{th:sumofsimpleeqmonad}, \ref{th:EilenbergMoore} and \ref{th:isoExceptionMet} we obtain an analogous transformer 
at the level of quantitative equational theories as follows.

\begin{cor} \label{cor:ExceptionTransfomerMet}
Let $\U$ be a basic quantitative equational theory.  
Then, $T_{\U}(- + E)$ is the free monad on the theory $\U + \mathcal{E}_E$ on $\Met$.
\end{cor}

Similarly, from Theorems~\ref{th:sumofsimpleeqmonadComplete} and \ref{th:isoExceptionCMet}, 
an analogous result holds also in $\CMet$.
\begin{cor} \label{cor:ExceptionTransfomerCMet}
Let $\U$ be a basic continuous equational theory.  
Then, $\CC T_{\U}(- + E)$ is the free monad on the theory $\U + \mathcal{E}_E$ on $\CMet$.
\end{cor}

\begin{exa}[Quantitative pointed convex semilattices]
Mio and Vignudelli~\cite{MioSV21,MioV20} while reasoning about the algebraic combination of quantitative nondeterminism (\cf\ Section~\ref{sec:Hausdorff}) and probabilistic choice (\cf\ Section~\ref{sec:probchoice}), considered the category of quantitative pointed convex semilattices and showed it is isomorphic to the Eilenberg-Moore category for $\hat{\C}(- + 1)$, \ie, the quantitative variant of the monad of \emph{(nonempty) convex sets of sub-probability distributions}.

This monad is just the composition of $\hat{\C}$, the (nonempty) convex sets of probability 
distribution monad~\cite{KeimelP16}, with the termination monad $(- + 1)$.
So, as $\hat{\C}$ is presented by the quantitative theory of convex semilattices, their result can be recovered as a simple application of Corollary~\ref{cor:ExceptionTransfomerMet} and Theorem~\ref{th:sumofsimpleeqmonad} and further extended on complete metric spaces by means of Corollary~\ref{cor:ExceptionTransfomerCMet} and Theorem~\ref{th:sumofsimpleeqmonadComplete}.
\end{exa}

\subsection{Sum with Interactive Inputs}
\label{sec:sumInput}
Now we consider the sum of generic quantitative effects $T$ with the monads $\tilde\Sigma^*$ of contractive 
operators, of which interactive inputs is a particular instance (\cf\ Section~\ref{sec:controperators}).

From Theorem~\ref{th:sumofsimpleeqmonad} we know that if $T$ has a quantitative algebraic presentation in terms of a basic theory $\U$, the sum exists, and, when starting with quantitative theories, we know how to describe it. But for the purposes of calculation, it is still convenient to have a more explicit construction of the sum qua monad, and Hyland et al.\ provided such a construction (\cf~\cite[Theorem~4]{HylandPP06}), which we recall below for convenience. The key fact used here is that the monad of contractive operators is described as the free monad on an endofunctor with countable rank, namely the contractive signature functor $\tilde\Sigma$ given in \eqref{eq:contractiveSignFunctor}.
\begin{thmC}[{\cite{HylandPP06}}] \label{th:sumofMonads}
Given an endofunctor $F$ and a monad $T$ on a category $\cat{C}$, if the free monads 
$F^*$ and $(FT)^*$ exist and are algebraic, then the sum of monads $T + F^*$ exists and is 
given by a canonical monad structure on the composite $T(FT)^*$.
\end{thmC}

As remarked in~\cite{HylandPP06}, when $\cat{C}$ is locally countably presentable and both 
$T$ and $F$ have countable rank, then $F^*$ and $(FT)^*$ exist and are algebraic.
Moreover, also $T(FT)^*$ has countable rank and so it is also the sum of monads $T + F^*$.

\medskip
We know that $\Met$ is locally countably presentable and that contractive signature functors 
$\tilde\Sigma$ have countable rank (\cf\ Section~\ref{sec:ContractiveOp}). Moreover, as recently proved 
by Ford et al.~\cite{FordMS21}, any quantitative theory $\U$ induces a 
monad $T_{\U}$ with countable rank. 

Therefore, from the discussion above and by Theorems~\ref{th:sumofsimpleeqmonad} and 
\ref{th:isoContractiveMonad}, we obtain the following characterization.
\begin{cor} \label{cor:ContractiveTransfomerMet}
Let $\U$ be a basic quantitative equational theory.  
Then, $T_{\U}(\tilde\Sigma T_{\U})^*$ is the free monad on the theory $\U + \O{\Sigma}$ on $\Met$.
\end{cor}

As observed in~\cite{HylandPP06}, the monad $T(FT)^*$ of Theorem~\ref{th:sumofMonads} is
simply another form of the generalised resumptions monad transformer of 
Cenciarelli and Moggi~\cite{Cenciarelli93}, sending $T$ to $\mu y. T(F y + -)$.
Hence, by the characterization above and guided by the same observations that lead
to~\cite[Corollary~2]{HylandPP06}, we obtain an analogous transformer at the level of quantitative equational 
theories as follows.
\begin{cor} \label{cor:ResumptionTransfomerMet}
Let $\U$ be a basic quantitative equational theory.  
Then, $\mu y. T_{\U}(\tilde\Sigma y + -)$ is the free monad on the theory $\U + \O{\Sigma}$ on $\Met$.
\end{cor}

Similarly, by Theorems~\ref{th:sumofsimpleeqmonadComplete} and \ref{th:isoContraciveMonadCompletion}, 
an analogous result holds also in $\CMet$.
\begin{cor} \label{cor:ResumptionTransfomerCMet}
Let $\U$ be a continuous quantitative equational theory.  
Then, $\mu y. \CC T(\tilde\Sigma y + -)$ is the free monad on the theory $\U + \O{\Sigma}$ on $\CMet$.
\end{cor}

\begin{rem}
It is worth remarking that using these ideas one obtains a modular description
of the monads of contractive operators. Let $\Sigma_1, \Sigma_2$ be two disjoint signatures of contractive
operators. It is clear that $\O{\Sigma_1 \cup \Sigma_2}$ is the same as the sum of theories $\O{\Sigma_1}$ and 
$\O{\Sigma_2}$. Hence, the sum $\tilde\Sigma_1^* + \tilde\Sigma_2^*$ is given by the free monad
$(\tilde\Sigma_1 + \tilde\Sigma_2)^*$, where we now mean the pointwise sum of functors.
\end{rem}

%%%%%%%%%%%%%%%%%%%%%%%%%%%%%%%
\subsection{The Algebras of Markov Processes}
\label{sec:markovprocesses}

In this section, we show how to obtain a quantitative equational axiomatization
of Markov processes with discounted bisimilarity metric~\cite{DesharnaisGJP04}
as the composition, via sum, of the following quantitative theories:

\begin{enumerate}%[label={\itshape(\alph*)}, fullwidth, itemsep=0.5ex]
\item 
\emph{The quantitative theory $\B$ of interpolative barycentric algebras},
used to express probabilistic nondeterminism with Kantorovich metric (Section~\ref{sec:probchoice});

\item 
\emph{The quantitative theory $\mathcal{E}_1$ of exceptions} over $1 = \{ * \}$, with $*$ 
as the only exception. This will be used to express termination (Section~\ref{sec:termination});

\item
\emph{The quantitative theory of contractive operators} (Section \ref{sec:controperators}).
In our case, we consider a signature $\Sigma_{\diamond} = \ens{\diamond \colon \tupl{1,c} }$ 
with a unary operator $\diamond$ with contractive factor $c \in (0,1)$.
This will be used to axiomatize the transition to a next state with discount factor $c$.
\end{enumerate}

\medskip\noindent
Formally, we define the \emph{quantitative theory of Markov processes} as
%, denoted by $\U_{\textbf{MP}}$, is given by
\begin{equation*}
  \U_{\textbf{MP}} = \B + \mathcal{E}_{1} + \O{\Sigma_{\diamond}} \,.
\end{equation*}
with signature $\Sigma_{\textbf{MP}} = \Sigma_{\B} \cup \Sigma_{1} \cup \Sigma_{\diamond}$ given 
as the disjoint union of those from its component theories. More explicitly,
\begin{equation*}
  \Sigma_{\textbf{MP}}  = \set{ +_e \colon 2}{ e \in [0,1]} \cup \{ \textsf{raise}_* \colon 0 \} \cup \{ \diamond \colon \tupl{1, c} \}
\end{equation*}
and $\U_{\textbf{MP}}$ has the following set of axioms 
\begin{align*} 
(\Bone)\,
& \vdash x +_1 y =_0 x \,, \\
(\Btwo)\, 
& \vdash x +_e x =_0 x \,, \\
(\SC)\,
& \vdash x +_e y =_0 y +_{1-e} x \,, \\
(\SA)\,
& \vdash (x +_e y) +_{e'} z =_0 x +_{ee'} (y +_{\frac{e' - ee'}{1 - ee'}} z) \,, \text{ for $e,e' \in [0,1)$} \,, 
\\ 
(\IB)\,
& \{ x \,{=_\e}\, y, x' \,{=_{\e'}}\, y' \} \,{\vdash}\, x +_e x' \,{=_{\delta}}\, y +_e y', \, \text{for $\delta \geq e \e + (1-e) \e'$,}
\\
(\Lip{\diamond})\,
& \{x =_\e y \} \vdash \diamond(x) =_{\delta} \diamond(y) \,,
\text{ for $\delta \geq c \e$} \,.
\end{align*}
Note that, the constant $\textsf{raise}_*$ has no explicit associated axiom since $\mathcal{E}_{1}$ is 
the trivial theory, corresponding to that for termination.

Intuitively, $\Sigma_{\textbf{MP}}$-terms (modulo $=_0$ provability) can be interpreted as equivalence classes of behaviours of Markov processes up to bisimilarity. The term $t +_e t'$ expresses convex combination of behaviours; $\textsf{raise}_*$ represents termination (or the deadlock behavior); and $\diamond(t)$ expresses
the ability to take a transition to the behaviour $t$.

\subsubsection{Markov Processes over Metric Spaces}
\label{sec:mp&bisim}

Following~\cite[Section~6]{BreugelHMW07}, we regard Markov processes as coalgebras on the category of
metric spaces, and slightly extending their approach to encompass the case 
when the bisimilarity distance is discounted by a factor $0 < c < 1$.

We consider two variants of Markov processes according to the type of their transition 
distribution functions%
\footnote{Note that the two types of coalgebras we are considering 
live in two different categories, $\Met$ and $\CMet$.}:
\begin{align*}
  X &\longrightarrow \Pi(c \cdot X + 1 )  \quad \text{in $\Met$} \,, \\
  X &\longrightarrow \Delta(c \cdot X + 1 )  \quad \text{in $\CMet$} \,,
\end{align*}
where $\Pi$ and $\Delta$ are the functors from Section~\ref{sec:probchoice}, mapping a 
metric space $X$ to a space of probability measures with Kantorovich metric.
The first variant is Markov processes with finitely supported transition probability distributions, commonly
regarded as \emph{Markov chains}. The second variant is Markov processes with Radon transition probability distributions. The use of the rescaling functor $(c \cdot -)$ is to express that transition functions are 
$c$-Lipschitz continuous, with contractive factor $0 \leq c \leq 1$. We will collectively refer to these two types of coalgebras structures as $c$-Markov processes.

In~\cite{BreugelHMW07}, van Breugel et al.\ characterized the bisimilarity distance 
on (labelled) Markov processes as the pseudometric induced by the unique homomorphism to the final 
coalgebra. We will do the same here by replicating their arguments in our specific setting.

\begin{prop}
The final coalgebras for $\Pi(c \cdot - + 1 )$ and $\Delta(c \cdot - + 1 )$ exist.
\end{prop}
\begin{proof}
As the categories $\Met$ and $\CMet$ are both complete and accessible (\cf\ Appendices~\ref{app:exmetric} and \ref{sec:EMet-lcp} for the formal definitions and proofs), the thesis follows by \cite[Theorem~8]{BreugelHMW07}, by showing that $\Pi(c \cdot - + 1 )$ and $\Delta(c \cdot - + 1 )$ 
are accessible functors (more precisely, $\aleph_1$-accessible). 

Notice that
$\Pi(c \cdot - + 1 )$ has a quantitative algebraic presentation in $\Met$ in terms of the theory $\B^c$
defined as $\B + \mathcal{E}_1$ where the axiom (\IB) (\cf\ Section~\ref{sec:barycentric}) is replaced by 
\begin{align*}
{ (\IB^c)} \,
& \{ x \,{=_\e}\, y, x' \,{=_{\e'}}\, y' \} \,{\vdash}\, x +_e x' \,{=_{\delta}}\, y +_e y', \, 
\text{ for $\delta \geq c(e \e + (1-e) \e')$} \,,
\end{align*}
that is, $T_{\B^c} \cong \Pi(c \cdot - + 1 )$ (the proof follows essentially identically to~\cite[Theorem~10.5]{MardarePP:LICS16}, which implies the isomorphism of monads). 
As~\cite{FordMS21} proved that the monads freely generated by a quantitative theory are $\aleph_1$-accessible,
we have that the final coalgebra for $\Pi(c \cdot - + 1 )$ exists.
Moreover, as $\CC T_{\B^c} \cong \Delta(c \cdot - + 1 )$, $\CC$ is $\aleph_1$-accessible, and $\aleph_1$-accessibility is closed under composition, we have that also $\Delta(c \cdot - + 1 )$ admits a final coalgebra.
\end{proof}

Then, the \emph{$c$-discounted bisimilarity pseudometric} 
on a $c$-Markov process $(X, \tau)$ is defined as the function $\dist \colon X \times X \to [0,\infty]$ given as
\begin{equation*}
  \dist(x,x') = d_{Z}(h(x), h(x')) \,,
\end{equation*}
where $h \colon X \to {Z}$ is the unique homomorphism to the 
final $c$-Markov process $(Z,\omega)$.

This distance has a characterization as the least fixed point of a 
monotone function on a complete lattice of $[0,\infty]$-valued pseudometrics.
\begin{prop} \label{prop:fixpointbisimdist}
The $c$-discounted bisimilarity pseudometric $\dist$ on $(X, \tau)$ 
is the \emph{unique} fixed point
of the following operator on the complete lattice of extended pseudometrics $d$ on $X$
with point-wise order $\sqsubseteq$, such that $d \sqsubseteq d_X$,
\begin{equation*}
  \Psi^c(d)(x,x') = \sup_{f} \left| \lebint{f}{\tau(x)} - \lebint{f}{\tau(x')} \right| \,,
\end{equation*}
with $f$ ranging over non-expansive positive $1$-bounded
real valued functions on $c \cdot X + 1$.
\end{prop}
\begin{proof}
Similar to the fixed point characterization given in \cite[Section~6]{BreugelHMW07}. The unicity of the fixed point follows by Banach fixed point theorem. Indeed, the set of extended real valued functions on $X \times X$
(which is a superset of the set of extended pseudometrics on $X$) can be turned into a complete Banach space 
by means of the sup-norm $|| f || = \sup_{x,x'} |f(x,x')|$ and $\Psi^c$ is a $c$-contractive operator on it.
\end{proof}

\subsubsection{Quantitative Algebraic Presentation}
\label{sec:presentationMP}
Here we relate $c$-Markov processes and their bisimilarity distance to the free algebras of $\U_{\textbf{MP}}$, both on $\Met$ and $\CMet$.

\paragraph{On Metric Spaces.}
We start by characterizing the monad $T_{\U_{\textbf{MP}}}$ on $\Met$. We do this in steps, by explaining the
contribution of each subtheory in the sum
\begin{equation*}
  \U_{\textbf{MP}} = \B + \mathcal{E}_{1} + \O{\Sigma_{\diamond}} \,.
\end{equation*}

(Step 1) First, note that $T_{\mathcal{E}_1} \iso (- + 1)$ is the \emph{maybe monad} (Theorem~\ref{th:isoExceptionMet}).
As $\B$ is basic, by Corollary~\ref{cor:ExceptionTransfomerMet} and Theorem~\ref{th:isofBarycentricMonad}, the free monad on $\B + \mathcal{E}_1$ is 
\begin{equation*}
T_{\B + \mathcal{E}_1} \iso T_\B(- + 1) \iso \Pi(- + 1) \,.
\end{equation*}
where $\Pi(- + 1)$ is the \emph{finitely supported sub-distribution monad} with functor assigning to
$X \in \Met$ the space of finitely supported Borel sub-probability measures with Kantorovich metric. 
Thus, $\B + \mathcal{E}_1$ axiomatizes finitely supported sub-probability distributions with Kantorovich metric.

(Step 2) The final step is to sum the above with the theory $\O{\Sigma_\diamond}$.
By Corollary~\ref{cor:ResumptionTransfomerMet}, the free monad on $\U_{\textbf{MP}} = \B + \mathcal{E}_1 + \O{\Sigma_\diamond}$ is 
\begin{equation*}
T_{\U_{\textbf{MP}}} 
\iso \mu y.  T_{\B + \mathcal{E}_1}(c \cdot y + - ) 
\iso \mu y.  \Pi(c \cdot y + 1 + -) \,,
\end{equation*}
where we implicitly applied the isomorphisms $c \cdot (A + B) \iso c \cdot A + c \cdot B$ and $1 \iso c \cdot 1$.
Explicitly, this means that the free monad on $\U_{\textbf{MP}}$
assigns to an arbitrary metric space $X \in \Met$ the \emph{initial solution} 
to the following functorial equation in $\Met$
\begin{equation}
  \MP \iso \Pi(c \cdot \MP + 1 + X) \,.
  \label{eq:functorialEquationMP}
\end{equation}

\medskip
Next we argue that $\U_{\textbf{MP}}$ axiomatizes the initial $c$-Markov process on $\Met$ with \mbox{$c$-discounted} bisimilarity metric.
Let $X = 0$ be the empty metric space (\ie, the initial object in $\Met$). Then \eqref{eq:functorialEquationMP} corresponds to the isomorphism on the initial $\Pi(c \cdot - + 1)$-algebra.
The isomorphism provides us also with a $\Pi(c \cdot - + 1)$-coalgebra
structure on $\MP[0]$ which, according to our interpretation, is a $c$-Markov process $(\MP[0], \tau_0)$.

The key observation is that the metric on $\MP[0]$ is the bisimilarity metric.
\begin{lem} \label{lem:bisimMetric}
$d_{\MP[0]}$ is the $c$-discounted bisimilarity metric on $(\MP[0], \tau_0)$.
\end{lem}
\begin{proof}
Isomorphisms in $\Met$ are isometries. Hence, by definition of $(\MP[0], \tau_0)$ and \eqref{eq:functorialEquationMP} 
\begin{equation*}
d_{\MP[0]}(x,x') = \K[d]{\tau_0(x), \tau_0(x')} \,,
\end{equation*}
where $d$ is the metric on $c \cdot \MP[0] + 1$. 
By Kantorovich-Rubinstein Duality~\cite[Theorem.~5.10]{Villani08}
\begin{equation*}
\K[d]{\tau_0(x), \tau_0(x')} = \sup_{f} \left| \lebint{f}{\tau_0(x)} - \lebint{f}{\tau_0(x')} \right| \,,
\end{equation*}
where $f$ ranges over non-expansive functions $f \colon c \cdot X + 1 \to [0,1]$. 
Thus, the thesis follows by Proposition~\ref{prop:fixpointbisimdist}.
\end{proof}

\begin{rem} 
For a less abstract description of $(\MP[0], \tau_0)$, notice that the elements of $\MP[0]$ 
are ground terms over the signature $\Sigma_{\textbf{MP}}$ modulo $=_0$ provability.
One can interpret a term as a pointed (or rooted) acyclic sub-probabilistic Markov chain up-to bisimilarity.
For example, the term $\diamond( \diamond(\textsf{raise}_*) +_{\frac{1}{2}} (\diamond(\diamond(\textsf{raise}_*)+_{\frac{1}{2}} \textsf{raise}_*) )$ corresponds to the sub-probabilistic Markov chain below
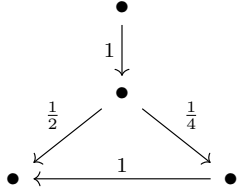
\begin{equation*}
\begin{tikzcd}
& \bullet \arrow[d, "1"' ] \\
& \bullet \arrow[dl, "\frac{1}{2}"'] \arrow[dr, "\frac{1}{4}"] \\
\bullet & & \bullet \arrow[ll, "1"' ]
\end{tikzcd}
\end{equation*}
and $\textsf{raise}_*$ corresponds to the deadlock process, with probability $0$ to move to any state.
\end{rem}

\paragraph{On Complete Metric Spaces.}
Now we characterize the monad $\CC T_{\U_{\textbf{MP}}}$ on $\CMet$. We do this by following the 
same steps as for the monad $T_{\U_{\textbf{MP}}}$ on $\Met$.

(Step 1) By Theorem~\ref{th:isoExceptionCMet}, $\CC T_{\mathcal{E}_1} \iso (- + 1)$ is the 
\emph{maybe monad}. As $\B$ is continuous, by Corollary~\ref{cor:ExceptionTransfomerCMet} and Theorem~\ref{th:isoGiry}, the completion of the free monad on $\B + \mathcal{E}_1$ is 
\begin{equation*}
\CC T_{\B + \mathcal{E}_1} \iso \CC T_\B(- + 1) \iso \Delta(- + 1) \,.
\end{equation*}
where $\Delta(- + 1)$ is the \emph{Radon sub-probability distribution monad} with Kantorovich metric.

(Step 2) In combination with the theory $\O{\Sigma_\diamond}$, by Corollary~\ref{cor:ResumptionTransfomerCMet}, the free completion monad on $\U_{\textbf{MP}} = \B + \mathcal{E}_1 + \O{\Sigma_\diamond}$ is given by
\begin{equation*}
\CC T_{\U_{\textbf{MP}}} 
\iso \mu y.  \CC T_{\B + \mathcal{E}_1}(c \cdot y + - ) 
\iso \mu y.  \Delta(c \cdot y + 1 + -) \,.
\end{equation*}
This means that also for the case of complete metric spaces the free monad on 
$\U_{\textbf{MP}}$ assigns to any arbitrary metric space $X \in \CMet$ the 
\emph{initial solution} of the following functorial equation in $\CMet$
\begin{equation}
  \MP \iso \Delta(1 + c \cdot \MP + X) \,.
  \label{eq:functorialEquationCMP}
\end{equation}

\medskip
Observe that the map $\omega_X \colon \MP \to \Delta(1 + c \cdot \MP + X)$
arising from the above isomorphism is a coalgebra structure for the functor
$\Delta(1 + c \cdot - + X)$ on $\CMet$. Next we show that $(\MP, \omega_X)$ is
actually the final coalgebra.

\begin{thmC}[{\cite[Section~7]{TuriR98}}] \label{th:uniquefixedpoint}
Every locally contractive endofunctor $H$ on $\CMet$ has a unique 
fixed point which is both an initial algebra and a final coalgebra for $H$.
\end{thmC}

Recall from Example~\ref{ex:monoidalclosed} that
the internal hom $[X, Y]$ in $\CMet$ is the set of non-expansive maps from $X$ to $Y$ with
point-wise supremum metric $d_{[X,Y]}(f,g) = \sup_{x \in X} d_Y(f(x),g(x))$.

An endofunctor $H$ on $\CMet$ is \emph{locally $c$-Lipschitz continuous} if 
for all $X,Y \in \CMet$, 
and non-expansive maps $f,g \colon X \to Y$,
\begin{equation*}
  d_{[HX, HY]}(H(f), H(g)) \leq c \cdot d_{[X,Y]}(f,g)  \,.
\end{equation*}
$H$ is \emph{locally non-expansive} if it is locally $1$-Lipschitz
continuous, and \emph{locally contractive} if it is locally $c$-Lipschitz
continuous, for some $0 \leq c < 1$.

Examples of locally contractive functors are the constant functors and the
rescaling functor $(c \cdot -)$, for $0 \leq c < 1$. Moreover, locally
contractiveness is preserved by products and coproducts and composition; and, if $H$ is locally non-expansive and $G$ is locally contractive, then $HG$ is locally contractive.

\begin{lem} \label{lm:Girylocallynonexp}
The endofunctor $\Delta$ on $\CMet$ is locally non-expansive.
\end{lem}
 \begin{proof}
 We need to check that for all $f,g \in \CMet(X,Y)$,
 \begin{equation}
   \sup_{x \in X} d_Y(f(x),g(x)) \geq \sup_{\mu \in \Delta X} \K[d_Y]{\Delta f(\mu), \Delta g(\mu)} \,.
   \label{eq:Deltalocallynonexpansive}
 \end{equation}
 Denote by $\Phi_Y$ be the set of non-expansive functions $k \colon Y \to [0,1]$, \ie, those functions such that $\forall y,y'.\, |k(y)-k(y')| \leq d_Y(y,y')$. Then, for any $\mu \in \Delta X$,
 \begin{align*}
 \K[d_Y]{\Delta f(\mu), \Delta g(\mu)} 
 &= \sup_{k \in \Phi_Y} \left| \lebint{k}{\Delta f(\mu)} - \lebint{k}{\Delta g(\mu)} \right|
 \tag{Kantorovich duality} \\
 &= \sup_{k \in \Phi_Y} \left| \lebint{k}{(\mu \circ f^{-1})} - \lebint{k}{(\mu \circ g^{-1})} \right|
 \tag{def.~$\Delta$} \\
 &= \sup_{k \in \Phi_Y} \left| \lebint{k \circ f}{\mu} - \lebint{k \circ g}{\mu} \right|
 \tag{change of var.} \\
 &= \sup_{k \in \Phi_Y} \left| \lebint{(k \circ f) - (k \circ g)}{\mu} \right|
 \tag{linearity of $\int$} \\
 &\leq \sup_{k \in \Phi_Y} \lebint{  \left| (k \circ f) - (k \circ g) \right|}{\mu} 
 \tag{subadd.\ of $| \cdot |$} \\
 &\leq \lebint{ d_Y \circ \tupl{f,g}}{\mu} 
 \tag{$k$ non-expansive} \\
 &\leq \lebint{ \sup_{x \in X} d_Y(f(x),g(x)) }{\mu} 
 \tag{monotonicity of $\int$} \\
 &= \sup_{x \in X} d_Y(f(x),g(x)) \,.
 \tag{$\mu$ probability measure}
 \end{align*}
 For the generality of $\mu \in \Delta X$, the above inequality implies \eqref{eq:Deltalocallynonexpansive}.
 \end{proof}
 
Thus, the following holds. 
\begin{thm} \label{th:finalDeltacoalgebra}
  $(\MP, \omega_X)$ is the final coalgebra for
  $\Delta( 1 + c \cdot Id + X)$ in $\CMet$.
\end{thm}
\begin{proof}
  This is a direct consequence of Theorem~\ref{th:uniquefixedpoint} and
  Lemma~\ref{lm:Girylocallynonexp}, since, $1 + c \cdot - + X$ is locally
  contractive and the composition of a locally contractive functor with a
  locally non-expansive one is locally contractive.
\end{proof}

Note that, when $X = 0$ is the empty metric space, the coalgebra of
this functor correspond to the final $c$-Markov process we have used in
Section~\ref{sec:mp&bisim} to characterize the $c$-discounted probabilistic
bisimilarity distance. When $X$ is not the empty space, we obtain coalgebraic
structures that can be interpreted as Markov process with $X$-labelled 
terminal states; one can view the labels in $X$ as describing different kind 
of termination of the process. 

Hence, in the light of Theorem~\ref{th:finalDeltacoalgebra}, we have
shown that for the case of complete metric spaces $\U_{\textbf{MP}}$
axiomatizes the $c$-discounted bisimilarity distance on the final Markov
process.

\begin{rem}
While by interpreting the theory $\U_{\textbf{MP}}$ over $\Met$ we can only characterize Markov processes 
that are acyclic, by doing it over $\CMet$ we obtain an algebraic representation of all bisimilarity classes as the elements of the final coalgebra. Thus, among others, we also recover Markov processes with cyclic structures as the limit of all their finite unfoldings. 
\end{rem}

%%%%%%%%%%%%%%%%%%%%%%%%%%%%%%%%%%%%%%%%%%%%%%
\section{Tensor of Quantitative Theories}
\label{sec:tensor} 

In this section, we consider the commutative combination of quantitative theories, their \emph{tensor}, 
that imposes mutual commutation of the operations from each theory. As such, it is properly coarser than the sum of two theories, which is just their unrestrained combination. 

The main theoretical result is that the free monad on the tensor of two basic theories corresponds to the categorical tensor of the free monads on the theories (\cf, Theorem~\ref{th:tensorofsimpleeqmonad}). In the proof given, we use the fact that the quantitative theories are \emph{basic}, as this allows us to exploit the correspondence between the algebras of a theory $\U$ and the EM-algebras of the monad $T_{\U}$ (Theorem~\ref{th:EilenbergMoore}). 

Our main examples of tensor of quantitative effects are given by the combination of 
reader and writer quantitative monads with arbitrary quantitative effects (Section~\ref{sec:tensorReaderWriter}).
We conclude the section by showing three nontrivial applications of tensorial combinations of quantitative theories by providing modular axiomatizations of labelled Markov processes (Section~\ref{sec:labelmarkovprocesses}), 
Mealy machines (Section~\ref{sec:mealyMachines}), 
and Markov decision processes (Section~\ref{sec:markovdecitionprocesses}), 
with their respective bisimilarity distances.

%%%%%%%%%%%%%%%%%%%%%%%%%%%%%%%%%%%%%%%%%%%%%%
\subsection{Tensor of Strong Monads on Metric Spaces} \label{sec:MonadTensor}

In this section, we recall the definition of (categorical) \emph{tensor of strong monads} on $\Met$ and $\CMet$. Our presentation is based on Manes~\cite{Manes69}, which addresses the case for $\Set$ monads. Given that the monoidal structures of $\Met$ and $\CMet$ are essentially identical, we will concentrate on $\Met$ in the following discussion, with the understanding that the same results hold for $\CMet$ as well.

\smallskip
Recall that $\Met$ is a symmetric monoidal closed category (\cf\ Example~\ref{ex:monoidalclosed}): the monoidal product $X \mprod Y$ has $X \times Y$ as underlying set and distance function given by $d_{X \mprod Y}((x,y)(x',y')) = d_X(x,x') + d_Y(y,y')$; 
the internal hom $[X,Y]$ is given by the set of non-expansive maps from $X$ to $Y$ with distance function $d_{[X,Y]}(f,g) = \sup_{x \in X} d_Y(f(x),g(x))$.

\bigskip
As $\Met$ is self-enriched, it has all $v$-fold \emph{powers} 
(or $v$-\emph{powers}) $X^v$, of any $v, X \in \Met$, defined as $X^v =
[v, X]$~\cite{Kelly82}.   
Moreover, $(-)^v \colon \Met \to \Met$ is a strong functor with strength 
$\xi_{X,Y} \colon X \mprod Y^v \to (X \mprod Y)^v$ obtained by currying 
\begin{equation*}
  v \mprod (X \mprod Y^v) \xrightarrow{\,\iso\,}  
  X \mprod (v \mprod Y^v) \xrightarrow{\,X \mprod \ev \,} X \mprod Y \,.
\end{equation*}

\medskip
%Such a functor lifts to the category of algebras for a strong functor as follows.
%\begin{defi}[Power algebra]
%\label{def:PowerAlgebra}
Let $F \colon \Met \to \Met$ be a strong functor with strength $t$.
The $v$-\emph{power} functor $(-)^v$ is lifted to $F$-algebras by mapping
$(A, a)$ to $(A,a)^v = (A^v, a^v \circ \sigma_A)$, 
where $\sigma \colon F(-^v) \nat (F-)^v$ is the strong natural transformation where, at each component $X \in \Met$, $\sigma_X$ is obtained by currying
\begin{equation*}
  v \mprod FX^v \xrightarrow{\,t \,}  
  F(v \mprod X^v) \xrightarrow{\,F\ev \,} FX \,.
  \label{eq:distributiveLaw}
\end{equation*}
We call $(A,a)^v$ the $v$-power of $(A, a)$. As the definition above is valid for generic extended metric spaces $v \in \Met$
we have that $F$-algebras are closed under powers of $\Met$-objects.

\begin{defi}[Pre-operation of a strong functor]
\label{def:operation}
Let $F \colon \Met \to \Met$ be a strong functor and $v \in \Met$.  A
\emph{$v$-ary pre-operation of $F$} is a strong natural transformation of type $(-)^v \nat F$.
\end{defi}

We denote by $\Op[F](v)$ the set of $v$-ary pre-operations and by $\Op[F]$ the collection of all pre-operations of $F$.
An assignment of $g \in \Op[F](v)$ to an $F$-algebra $(A,a)$ is the composite $a^g = a \circ g_A$.
We call $a^g$ an \emph{operation} of $(A,a)$.

\begin{prop} \label{prop:operations}
Let $(A,a), (B,b)$ be $F$-algebras of a strong endofunctor $F$ on $\Met$ 
and $f \colon A \to B$ a morphism in $\Met$.
Then, the following are equivalent:
\begin{enumerate}
  \item \label{operations1} 
  $f$ is an $F$-homomorphisms from $(A,a)$ to $(B,b)$;
  \item \label{operations2} 
  For every $v \in \Met$ and $g \in \Op[F](v)$, $f \circ a^g = b^g \circ f^v$.
\end{enumerate}
\end{prop}
\begin{proof}
$\eqref{operations1} \Rightarrow \eqref{operations2}$ follows by definition of $a^g$, $b^g$
and naturality of $g$.  As for $\eqref{operations2} \Rightarrow \eqref{operations1}$, note that since $\Met$
is a symmetric monoidal closed category, we have a 1-1 correspondence between strong and 
$\Met$-enriched endofunctors on $\Met$, and also between strong and $\Met$-enriched natural transformations~\cite{Kock72}.  Therefore, by (the weak form of) the enriched Yoneda lemma (\cf~\cite{Kelly82}), 
there exists a natural bijection between strong natural transformations $g \in \Op[F](A)$ and the (generalised) elements of $FA$, \ie, 
morphisms of the form $1 \to FA$, obtained via the composition
\begin{equation*}
1 \xrightarrow{\,j_A\,} A^A \xrightarrow{\,g_A\,} FA \,.
\end{equation*}
where $j_A(*) = id_A \in A^A$. Thus, for any $e \colon 1 \to FA$, there exists $\hat{e} \in \Op[F](A)$ such that $\hat{e}_A \circ j_A = e$.
As in $\Met$ the elements of an extended metric space $X$ are identifiable by maps of type $1 \to X$ (equivalently, $\Met(1,X)$ is a jointly epic family), to prove \eqref{operations1} it suffices to show 
that $f \circ a \circ e = b \circ Ff \circ e$, for all $e \colon 1 \to FA$:
\begin{align*}
    f \circ a \circ e
    &= f \circ a \circ \hat{e}_A \circ j_A \tag{def. $\hat{e}_A$}\\
    &= f \circ a^{\hat{e}} \circ j_A \tag{def. $a^{\hat{e}}$} \\
    &= b^{\hat{e}} \circ f^A \circ j_A \tag{by \eqref{operations2}} \\
    &= b \circ \hat{e}_B \circ f^A \circ j_A \tag{def. $a^{\hat{e}}$} \\
    &= b \circ Ff \circ \hat{e}_A \circ j_A \tag{nat. $\hat{e}$} \\
    &= b \circ Ff \circ e \,. \tag{def. $\hat{e}_A$}
\end{align*}
This concludes the proof.
\end{proof}

The above proposition indicates that $F$-algebras are precisely characterized by their operations.
In some situations, depending on the functor $F$, one gets the same characterization with much fewer operations. We identify this property with the following definition.
\begin{defi}[Exhaustive sets of pre-operations] 
A subset $\mathcal{E} \subseteq \Op[F]$ of pre-operations of a strong functor $F$ is \emph{exhaustive}, if for any $F$-algebras $(A,a)$, $(B,b)$ and 
$f \colon A \to B$ in $\Met$, the following are equivalent:
\begin{enumerate}
  \item $f$ is a $F$-homomorphisms from $(A,a)$ to $(B,b)$;
  \item For every $v \in \Met$ and $v$-ary pre-operation $g \in \mathcal{E}$, $f \circ a^g = b^g \circ f^v$.
\end{enumerate}
\end{defi}

Let $F, G$ be two strong endofunctors on $\Met$.  
A \emph{$\tupl{F, G}$-bialgebra} 
is a triple $(A, a, b)$ consisting of a metric space $A \in \Met$ with both a $F$-algebra structure $a \colon FA \to A$ and a $G$-algebra structure $b \colon GA \to A$.  A morphism of $\tupl{F, G}$-bialgebras 
is a non-expansive map that is simultaneously a $F$- and $G$-homomorphism.
Denote by $\biAlg{F}{G}$ the category of $\tupl{F, G}$-bialgebras.

\begin{prop} \label{prop:comm}
Let $(A,a,b)$ be a $\tupl{F, G}$-bialgebra.  
The following statements are equivalent:
\begin{enumerate}
  \item \label{comm1}
  For all $v \in \Met$ and $g \in \Op[F](v)$, $a^g$ is a $G$-homomorphism;
  \item \label{comm2}
  For all $w \in \Met$ and $h \in \Op[G](w)$, $b^h$ is a $F$-homomorphism.
\end{enumerate}
Diagrammatically:
\begin{equation*}
%\small
\begin{tikzcd}
GA^v \arrow[r, "\bar{b}"] \arrow[d, "G(a^g)"'] 
\arrow[dr, phantom, "\eqref{comm1}" gray]
& A^v  \arrow[d, "a^g"] \\
GA \arrow[r, "b"] & A
\end{tikzcd}
\quad\text{iff}\quad
\begin{tikzcd}
FA^w \arrow[r, "\bar{a}"] \arrow[d, "F(b^h)"'] 
\arrow[dr, phantom, "\eqref{comm2}" gray]
& A^w  \arrow[d, "b^h"] \\
FA \arrow[r, "a"] & A
\end{tikzcd}
\end{equation*}
where $(A,a)^w = (A^w, \bar{a})$ and $(A,b)^v = (A^v, \bar{b})$.
\end{prop}

Prior to presenting the proof of this statement, it is beneficial to introduce a technical result that will prove useful in subsequent discussions.
\begin{prop} \label{prop:Iso}
Let $(A,a)$ be a $F$-algebra of a strong endofunctor $F$ 
on $\Met$. Then, for any $v,w \in \Met$ and $g \in \Op[F](v)$ the following commute 
\begin{equation*}
\begin{tikzcd}
(A^v)^w \arrow[rr, "\chi", "\iso"'] \arrow[rd, "(a^g)^w"'] 
& & (A^w)^v \arrow[ld, "\bar{a}^g"] \\
& A^w
\end{tikzcd}
\end{equation*}
where $(A,a)^w = (A^w, \bar{a})$ and $\chi$ is the canonical isomorphism.  
\end{prop}
\begin{proof} %(of Proposition~\ref{prop:Iso})
By the universality of the counit $\ev \colon (w \mprod -) \nat Id$ 
of the adjunction $(w \mprod - ) \dashv (-)^w$
it suffices to show that the following two diagrams commute:
\begin{equation*}
\begin{tikzcd}
& w \mprod (A^v)^w \arrow[dl, "a^g \circ \ev"'] \arrow[d, "w \mprod (a^g)^w"] \\
A & w \mprod A^w \arrow[l, "\ev"]
\end{tikzcd}
\quad
\begin{tikzcd}
w \mprod (A^v)^w \arrow[d, "a^g \circ \ev"'] \arrow[r, "w \mprod \chi"] 
& w \mprod (A^w)^v \arrow[d, "w \mprod \bar{a}^g"]\\
A & w \mprod A^w \arrow[l, "\ev"]
\end{tikzcd}
\end{equation*}
The diagram to the left commutes by naturality of the counit $\ev$;
the one to the right commutes as follows, 
where $\xi$ and $t$ are respectively the strengths of $(-)^v$ and $F$
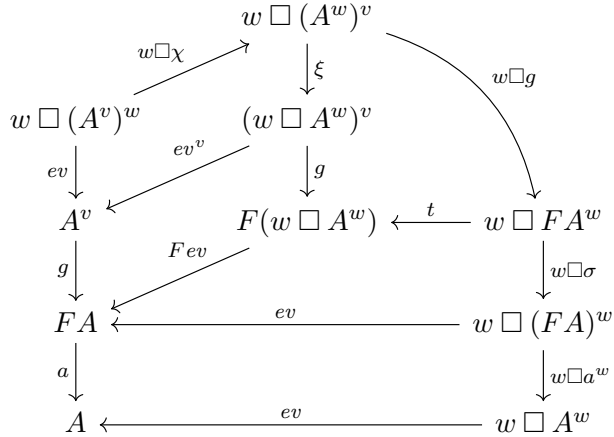
\begin{equation*}
\begin{tikzcd}
& w \mprod (A^w)^v \arrow[ddr, "w \mprod g", bend left] \arrow[d, "\xi"]
\\ w \mprod (A^v)^w \arrow[ur, "w \mprod \chi"] \arrow[d, "\ev"']
& (w \mprod A^w)^v \arrow[dl, "\ev^v"', near start] \arrow[d, "g"]
\\ A^v \arrow[d, "g"']
& F(w \mprod A^w) \arrow[dl, "F\ev"', near start] 
& w \mprod FA^w \arrow[l, "t"'] \arrow[d, "w \mprod \sigma"]
\\ FA \arrow[d, "a"']
& & w \mprod (FA)^w \arrow[ll, "\ev"'] \arrow[d, "w \mprod a^w"]
\\ A & & w \mprod A^w \arrow[ll, "\ev"']
\end{tikzcd}
\end{equation*}
by naturality of the counit $\ev$; definition of $\xi$ and $\chi$;
definition of the law $\sigma \colon F(-)^w \nat (F - )^w$; 
definition of $a^g$, $\bar{a}^g$; by $\bar{a} = a^w \circ \sigma$;
and because $g$ is strong.
\end{proof}

\begin{proof} (of Proposition~\ref{prop:comm})
$\eqref{comm1} \Rightarrow \eqref{comm2}$ By Proposition~\ref{prop:operations}, we prove \eqref{comm2}
by showing that for all $v \in \cat{V}$ and $g \in \Op[F](v)$, $b^h \circ \bar{a}^g = a^g \circ (b^h)^v$.
This is shown by the diagram below
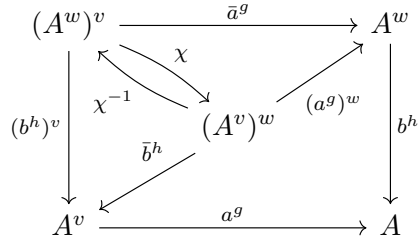
\begin{equation*}
\begin{tikzcd}
(A^w)^v \arrow[rr, "\bar{a}^g"] \arrow[rd, "\chi", bend left=10] \arrow[dd, "(b^h)^v"'] & & A^w \arrow[dd, "b^h"] \\
 & (A^v)^w \arrow[lu, "\chi^{-1}", bend left=10] 
 	\arrow[ru, "(a^g)^w"', near start] \arrow[ld, "\bar{b}^h"', near start] \\
A^v \arrow[rr, "a^g"] & & A
\end{tikzcd}
\end{equation*}
which commutes by Proposition~\ref{prop:Iso}, \eqref{comm1}, definition of $a^g$, and naturality of $g$. 
The implication $\eqref{comm2} \Rightarrow \eqref{comm1}$ is similar.
\end{proof}

\begin{defi}[Tensor algebra]
A \emph{tensor $\tupl{F, G}$-algebra} is a $\tupl{F, G}$-bialgebra $(A,a,b)$  that satisfies either of the equivalent conditions of Proposition~\ref{prop:comm}.
\end{defi}

In the case the functors $F$ and $G$ admit exhaustive sets of pre-operations, the conditions of Proposition~\ref{prop:comm} 
can be more conveniently expressed in the following way.
\begin{prop} \label{prop:exhaustiveComm}
Let $\mathcal{D}$ and $\mathcal{E}$ be exhaustive sets of pre-operations for $F$ and $G$, respectively. Then, $(A,a,b)$ is a tensor $\tupl{F, G}$-algebra iff it satisfies either of the equivalent conditions:
\begin{enumerate}
  \item \label{dComm1} 
  For all $g \in \mathcal{D}$, $a^g$ is a $G$-homomorphism;
  \item \label{dComm2} 
  For all $h \in \mathcal{E}$, $b^h$ is a $F$-homomorphism.
\end{enumerate}
\end{prop}
\begin{proof} %(of Proposition~\ref{prop:exhaustiveComm})
The equivalence of the statements \eqref{dComm1}, \eqref{dComm2} follows as 
in Proposition~\ref{prop:comm}, by using the density of $\D$ and $\E$ in lieu of Proposition~\ref{prop:operations}.

Let $(A,a,b)$ be a tensor $\tupl{F, G}$-algebra. Then, \eqref{dComm1} 
follows trivially because, $\mathcal{D}$ is a subset of pre-operations of $F$.
For the converse implication, assume \eqref{dComm1} and let $h \in \Op[G](w)$ for some 
$w \in \cat{V}$. We want to show that
\begin{equation*}
\begin{tikzcd}
F A^w \arrow[r, "\bar{a}"] \arrow[d, "F(b^h)"'] 
& A^w  \arrow[d, "b^h"] \\
\Sigma A \arrow[r, "a"] & A
\end{tikzcd}
\end{equation*}
commutes, where $(A,a)^{w} = (A^{w}, \bar{a})$.
Since $\mathcal{D}$ is exhaustive, it suffices to show that for 
all $v$-ary pre-operation $g \in \mathcal{D}$, $b^h \circ \bar{a}^g = a^g \circ (b^h)^{v}$. This follows by 
\begin{equation*}
\begin{tikzcd}
(A^w)^v \arrow[rr, "\bar{a}^g"] \arrow[rd, "\chi", bend left=10] \arrow[dd, "(b^h)^v"'] & & A^w \arrow[dd, "b^h"] \\
 & (A^v)^w \arrow[lu, "\chi^{-1}", bend left=10] 
 	\arrow[ru, "(a^g)^w"', near start] \arrow[ld, "\bar{b}^h"', near start] \\
A^v \arrow[rr, "a^g"] & & A
\end{tikzcd}
\end{equation*}
which commutes by Proposition~\ref{prop:Iso}, \eqref{dComm1}, 
definition of $a^g$, and naturality of $g$.
\end{proof}

Let $(T,\eta,\mu)$ be a strong monad on $\Met$.
Note that, as $T$ is a strong functor and the EM-algebras for $T$ are 
closed under powers of $\Met$-objects, all the results and definitions 
given in this section extends to EM-algebras for $T$.

\smallskip
Let $T$, $T'$ be two strong monads on $\Met$.  
A Eilenberg-Moore \emph{$\tupl{T, T'}$-bialgebra} is a triple $(A, a, a')$ consisting of an extended metric space $A \in \Met$ endowed with both a EM 
$T$-algebra structure $a \colon TA \to A$ and a EM $T'$-algebra structure 
$a' \colon T'A \to A$.  We say that a EM $\tupl{T, T'}$-bialgebra $(A,a,b)$
a EM tensor $\tupl{T, T'}$-algebra if it is so as a $\tupl{T, T'}$-bialgebra 
for the functors $T,T'$.
We denote by $\biEMAlg{T}{T'}$ the category of EM $\tupl{T, T'}$-bialgebras 
and by $\tensorAlg{T}{T'}$, the full subcategory of the EM 
$\tupl{T, T'}$-tensor algebras.

\begin{defi}[Tensor of monads]
If the forgetful functor $\tensorAlg{T}{T'} \to \cat{V}$ has left adjoint, then the
monad induced by the adjunction is the \emph{tensor} of $T, T'$, denoted $T \otimes T'$.
\end{defi}

Note that the tensor of monads does not necessarily exist
(see~\cite{BowlerGLS13} for counterexamples).
However, when it does $T \otimes T' \iso T' \otimes T$, 
as the categories of tensor biagebras 
$\tensorAlg{T}{T'}$ and $\tensorAlg{T'}{T}$ are isomorphic.

\begin{rem}[Discussion and related work]
Pre-operations of a strong functor $F$ are related to Plotkin and Power's \emph{algebraic operations}~\cite{PlotkinP01,PlotkinP03} in the sense that their assignment to $F$-algebras are the appropriate version of algebraic operations for functors.  Moreover, when considered over a strong monad $T$ they correspond to generic effects of type $I \to Tv$ (\ie, Kleisli maps of type $I \to v$, where $I$ is the unit object of the monoidal structure).  
The reason why we consider pre-operations over functors, and not just 
monads, is to relate the operations of an algebraic monad with those of 
its signature (see Section~\ref{sec:symbolicPreOp}).
\end{rem}

%%%%%%%%%%%%%%%%%%%%%%%%%%%%%%%%%%%%%%%%%%%%%%
\subsection{Tensor of Quantitative Theories}
\label{sec:TensorTheories} 

In this section, we develop the theory for the \emph{tensor} of quantitative 
equational theories. The main result is that the free monad on the tensor of 
two theories is the tensor of the monads on the theories. 

\smallskip
Let $\Sigma$, $\Sigma'$ be two disjoint signatures.  Following Freyd~\cite{Freyd66} (and~\cite{HylandPP06}),
we define the tensor of two quantitative equational theories $\U$, $\U'$ of
respective types $\Sigma$ and $\Sigma'$, written $\U \otimes \U'$, as the
smallest quantitative theory containing $\U$, $\U'$ and the quantitative equations
\begin{equation}
\vdash f( g(x^{1}_{1}, \dots, x^{1}_{m}), \dots, g(x^{n}_{1}, \dots, x^{n}_{m}) ) 
=_0  g( f(x^{1}_{1}, \dots, x^{n}_{1}), \dots, f(x^{1}_{m}, \dots, x^{n}_{m}) ) \,,
\label{eq:communtativeEq}
\end{equation}
for all $f \colon n \in \Sigma$ and $g \colon m \in \Sigma'$, expressing that the operations of one theory commute with the operations of the other.

\subsubsection{Symbolic Pre-operations} \label{sec:symbolicPreOp}

Towards our main result, we identify an exhaustive set of pre-operations for the free monads on quantitative equational theories which, in turn, will 
give us a simpler characterization for the tensor algebras for these monads (\cf~Proposition~\ref{prop:exhaustiveComm}).

First observe that any signature functor $\Sigma = \coprod_{f {:} n \in \Sigma} Id^n$ in $\Met$ is strong, as it is the coproduct of the strong functors 
$Id^n \iso (-)^{\underline{n}}$, where $\underline{n} \in \Met$ 
denotes the set $\ens{1, \dots, n}$ equipped with the discrete extended 
metric assigning infinite distance to distinct elements.
Moreover, the injections $\inj_f \colon (-)^{\underline{n}} \nat \Sigma$ are strong natural transformations, hence they are $\underline{n}$-ary pre-operations of $\Sigma$ (\cf~Definition~\ref{def:operation}).

\begin{prop} \label{prop:symbolicOperations}
$\mathcal{S}_\Sigma = \{ \inj_f  \mid f \colon n \in \Sigma \}$ is an exhaustive set of pre-operations of $\Sigma$.
\end{prop}
\begin{proof} %(of Proposition~\ref{prop:symbolicOperations})
Let $(A,a), (B,b)$ be $\Sigma$-algebras in $\Met$ and $h \colon A \to B$ 
a non-expansive map. We want to prove the equivalence of
\begin{enumerate}
  \item \label{symb1} 
  $f$ is a $\Sigma$-homomorphisms from $(A,a)$ to $(B,b)$;
  \item \label{symb2} 
  For every $f \colon n \in \Sigma$, $h \circ a^{\inj_f} = b^{\inj_f} \circ h^v$.
\end{enumerate}

$\eqref{symb1} \Rightarrow \eqref{symb2}$ follows by definition of $a^{\inj_f}$, $b^{\inj_f}$ and naturality
of $\inj_f \colon (-)^{\underline{n}} \nat \Sigma$. The implication
$\eqref{symb2} \Rightarrow \eqref{symb1}$ follows by the universality of the coproduct, as $\Sigma =  \coprod_{f {:} n \in \Sigma} Id^n$.
\end{proof}

In the following, the pre-operations in $\mathcal{S}_\Sigma$ will be called \emph{symbolic}, and to simplify the notation, for any $f \colon n \in \Sigma$ and $\Sigma$-algebra $(A,a)$, we write $a^f$ instead of $a^{\inj_f}$.

\medskip
Now we turn to study the pre-operations of the monad $T_{\U}$, for a quantitative equational theory $\U$ 
of type $\Sigma$. Firstly, observe that the monad $T_{\U}$ is strong, with strength 
\begin{equation*}
\zeta_{X,Y} \colon X \mprod T_{\U}Y \to T_{\U}(X \mprod Y) 
\end{equation*}
obtained by uncurrying the unique map $h_{X,Y}$ that, by Theorem~\ref{th:freeQAlgebra}, makes the following diagram commute
\begin{equation*}
\begin{tikzcd}
Y \arrow[rd, "\beta_{X,Y}"', near start, bend right] \arrow[r, "\eta^{\U}_Y"] & [-3ex]
T_{\U}Y \arrow[d, dashed, "h_{X,Y}"'] 
 & [5ex] \Sigma T_{\U}Y \arrow[l, "\psi^{\U}_Y"'] \arrow[d,"\Sigma h_{X,Y}"] \\
 & (T_{\U}(X \mprod Y))^X & \Sigma (T_{\U}(X \mprod Y))^X 
 \arrow[l, "\gamma"']
\end{tikzcd}
\end{equation*}
where $\beta_{X,Y}$ is the currying of 
$\eta^{\U}_{X \mprod Y} \colon X \mprod Y \to T_{\U}(X \mprod Y)$ and
$(T_{\U}(X \mprod Y)^X, \gamma)$ is the $X$-power of the free quantitative
$\Sigma$-algebra on $X \mprod Y$ satisfying $\U$.

Since a monad is strong iff both its unit and multiplication are strong natural transformations,
both $\eta^{\U}$, $\mu^{\U}$ are strong.
Moreover, also $\psi^{\U} \colon \Sigma T_{\U} \nat T_{\U}$ is strong.

Thus any pre-operation $g \in \Op[\Sigma](v)$ can be tuned into
a pre-operation of $T_{\U}$ as the composite 
\begin{equation*}
 (-)^{v} \xrightarrow{\,g\,} \Sigma
 			      \xrightarrow{\,\Sigma \eta^{\U}\,} \Sigma T_{\U}
			      \xrightarrow{\,\psi^{\U}\,} T_{\U} \,.
\end{equation*}

Moreover, when the theory $\U$ is basic, by Theorem~\ref{th:EilenbergMoore}, 
we can turn any exhaustive set of pre-operations of $\Sigma$ into an exhaustive set of pre-operations of $T_{\U}$.
\begin{prop} \label{prop:symbolicOperationsMonad}
Let $\U$ be a basic quantitative equational theory of type $\Sigma$. Then if $\mathcal{D}\subseteq \Op[\Sigma]$ is exhaustive, so is $\{ \psi^{\U} \circ \Sigma \eta^{\U} \circ g \mid g \in \mathcal{D} \} \subseteq \Op[T_{\U}]$.
\end{prop}
\begin{proof} %(of Proposition~\ref{prop:symbolicOperationsMonad})
$(A,a), (B,b)$ be $T_{\U}$-algebras and $h \colon A \to B$ a non-expansive map.
We want to prove the equivalence of
\begin{enumerate}
  \item \label{symbterm1} 
  $h$ is a $T_{\U}$-homomorphism from $(A,a)$ to $(B,b)$;
  \item \label{symbterm2} 
  For every $v$-ary pre-operation $g \in \mathcal{D}$, 
  $h \circ a^{(\psi^{\U} \circ \Sigma \eta^{\U} \circ g)} = b^{(\psi^{\U} \circ \Sigma \eta^{\U} \circ g)} \circ h^v$.
\end{enumerate}

$\eqref{symbterm1} \Rightarrow \eqref{symbterm2}$ follows by 
definition of $a^{(\psi^{\U} \circ \Sigma \eta^{\U} \circ g)}$, $b^{(\psi^{\U} \circ \Sigma \eta^{\U} \circ g)}$ 
and naturality of $\psi^{\U} \circ \Sigma \eta^{\U} \circ g$.
For the converse implication, recall that the isomorphism of categories from Theorem~\ref{th:EilenbergMoore},
maps a $T_{\U}$-algebra $(A,a)$ to 
the $\Sigma$-algebra $(A, a \circ \psi^{\U}_A \circ \Sigma\eta^{\U}_A)$ and morphisms essentially to themselves.
Thus $\eqref{symbterm2} \Rightarrow \eqref{symbterm1}$ follows by density of $\mathcal{D}$
and definition of $a^{(\psi^{\U} \circ \Sigma \eta^{\U} \circ g)}$, $b^{(\psi^{\U} \circ \Sigma \eta^{\U} \circ g)}$.
\end{proof}

When $\U$ is a basic theory, by combining Propositions~\ref{prop:symbolicOperations} and \ref{prop:symbolicOperationsMonad}, we easily obtain an exhaustive set 
of pre-operations also for the monad $T_{\U}$.
\begin{cor} \label{cor:symOperations}
$\mathcal{S}_{T_{\U}} = \{ \psi^{\U} \circ \Sigma \eta^{\U} \circ \inj_f \mid f \colon n \in \Sigma \}$ is a exhaustive set of pre-operations of $T_{\U}$, whenever $\U$ is a basic quantitative equational theory.
\end{cor}

Also the pre-operations in $\mathcal{S}_{T_{\U}}$ will be called \emph{symbolic} and we simplify the notation
by writing $a^{\tupl{f}}$ instead of $a^{(\psi^{\U} \circ \Sigma \eta^{\U} \circ \inj_f)}$, for $f \colon n \in \Sigma$ and $(A,a) \in \Alg{T_{\U}}$.

As an immediate consequence of Corollary~\ref{cor:symOperations} and Proposition~\ref{prop:exhaustiveComm}, we obtain the following simpler characterization for tensor $\tupl{T_{\U}, T_{\U'}}$-algebras.
\begin{cor} \label{cor:termComm}
Let $\U$, $\U'$ be basic quantitative theories respectively of type $\Sigma$, $\Sigma'$. Then, $(A,a,b)$ is a tensor $\tupl{T_{\U}, T_{\U'}}$-algebra iff
it satisfies either of the equivalent conditions
\begin{enumerate}
  \item \label{tComm1}
  For all $f \colon n \in \Sigma$, $a^{\tupl{f}}$ is a $T_{\U'}$-homomorphism;
  \item \label{tComm2}
  For all $g \colon n \in \Sigma'$, $b^{\tupl{g}}$ is a $T_{\U}$-homomorphism.
\end{enumerate}
\end{cor}

%%%%%%%%%%%%%%%%%%%%%%
\subsubsection{Tensor of Free Monads on Quantitative Theories}

Let $\U, \U'$ be basic quantitative theories respectively of type $\Sigma, \Sigma'$.
We show that any model for $\U \otimes \U'$ is a \emph{$\tupl{\U \otimes \U'}$-bialgebra}: 
an extended metric space $A$ with both a $\Sigma$-algebra structure 
$a \colon \Sigma A \to A$ satisfying $\U$ and a $\Sigma'$-algebra structure 
$b \colon \Sigma' A \to A$ satisfying $\U'$ and respecting the diagrammatic condition below, 
for all $f \colon n \in \Sigma$ and $g \colon m \in \Sigma'$
\begin{equation}
\begin{tikzcd}
A^{\underline{n}} \arrow[r, "a^f"] & A & A^{\underline{m}} \arrow[l, "b^g"'] \\
(A^{\underline{m}})^{\underline{n}} \arrow[rr, "\chi", "\iso"'] \arrow[u, "(b^g)^{\underline{n}}"] 
& & (A^{\underline{n}})^{\underline{m}} \arrow[u, "(a^f)^{\underline{m}}"'] 
\end{tikzcd}
\label{eq:commDiagram}
\end{equation}

Formally, we denote by $\KK{(\Sigma, \U) \otimes (\Sigma', \U')}$ the
category of $\tupl{\U \otimes \U'}$-bialgebras, with morphisms the non-expansive 
homomorphisms preserving both algebraic structures.  Then, the following isomorphism
of categories holds.
\begin{prop} \label{prop:bialgebrasTensor}
$\KK{\Sigma + \Sigma', \U \otimes \U'} \iso \KK{(\Sigma, \U) \otimes (\Sigma', \U')}$, for $\U, \U'$ basic theories.  
\end{prop}
\begin{proof} %(of Proposition~\ref{prop:bialgebrasTensor})
The isomorphism is given by the pair of functors 
\begin{equation*}
\begin{tikzcd}
\KK{\Sigma + \Sigma', \U \otimes \U'} \arrow[rr, bend left=5, "H"]
&  & 
\KK{(\Sigma, \U) \otimes (\Sigma', \U')} \arrow[ll, bend left=5, "K"] 
\end{tikzcd}
\end{equation*}
defined, for a $(\Sigma + \Sigma')$-algebra $(A,
a)$ satisfying $\U \otimes \U'$ and a 
$\tupl{\U \otimes \U'}$-bialgebra $(B, b, b')$, respectively as
\begin{align*}
H(A, a) = (A, a \circ \inj_l, a \circ \inj_r) \,,
&&
K(B, b, b') = (B, [b, b']) \,, 
\end{align*}
where $[b, b'] \colon \Sigma B + \Sigma' B \to B$ is the unique map induced by $b$ 
and $b'$ by the universality of the coproduct. 
Both functors are identity on morphisms; it is easy to see that
a homomorphism in one sense is also a homomorphism in the other.
 
The pair of functors above is the restriction of the isomorphic pair of  
functors used in the proof of~\cite[Proposition~4.1]{BacciMPP18}.
Thus, to show $H$ and $K$ are well defined we are just left to deal with checking 
that the restriction conditions on the subcategories are preserved both ways.

As for $H$, we prove that whenever $\A = (A, a)$ satisfies the quantitative 
equation in \eqref{eq:communtativeEq}, then $(A, a \circ \inj_l, a \circ \inj_r)$ satisfies 
the commutativity of the diagram in \eqref{eq:commDiagram}. 
This follows as, for all $f \colon n \in \Sigma$ and $g \colon m \in \Sigma'$, by definition of algebraic 
interpretation $(-)^\A$, we have
\begin{gather*}
f^\A = a \circ \inj_l \circ \inj_f = (a \circ \inj_l)^f \,, \\
g^\A = a \circ \inj_r \circ \inj_g = (a \circ \inj_r)^g \,.
\end{gather*}
Thus, the satisfiability \eqref{eq:communtativeEq} 
coincides with the commutativity of the diagram in \eqref{eq:commDiagram}.

For $K$ we need to show that whenever 
$(B, b, b')$ satisfies the commutativity of the diagram in \eqref{eq:commDiagram}, 
then $\A = (A, [b,b'])$ satisfies \eqref{eq:communtativeEq}. This follows as, 
for all $f \colon n \in \Sigma$ and $g \colon m \in \Sigma'$, by definition of algebraic 
interpretation $(-)^\A$, we have
\begin{gather*}
f^\A = [b,b'] \circ \inj_l \circ \inj_f = (b)^f \,, \\
g^\A = [b,b'] \circ \inj_r \circ \inj_g = (b')^g \,.
\end{gather*}
Thus, the commutativity of the diagram in \eqref{eq:commDiagram}
coincides with the satisfiability of \eqref{eq:communtativeEq}.
\end{proof} 

Moreover, by adapting the isomorphism of Theorem~\ref{th:EilenbergMoore} and
exploiting the fact that symbolic pre-operations are exhaustive (\cf~Corollary~\ref{cor:termComm}) the following is also true.

\begin{prop} \label{prop:tensorAlgebras}
$ \KK{(\Sigma, \U) \otimes (\Sigma', \U')} \iso \tensorAlg{T_{\U}}{T_{\U'}}$, for $\U, \U'$ basic theories.  
\end{prop}
\begin{proof} %(of Proposition~\ref{prop:tensorAlgebras})
Recall the isomorphism of categories from Theorem~\ref{th:EilenbergMoore}
\begin{equation*}
\begin{tikzcd}
\EMAlg{T_{\U}} \arrow[rr, bend left=10, "H"]
&  & 
\KK[\Sigma]{\U} \arrow[ll, bend left=10, "K"] 
\end{tikzcd}
\end{equation*}
mapping morphisms to themselves and on objects acting as
follows: for $(A,a) \in \EMAlg{T_{\U}}$ and $(B, b) \in
\KK[\Sigma]{\U}$, 
\begin{align*}
  H(A, a) = (A, a \circ \psi^{\U}_A \circ \Sigma\eta^{\U}_A) \,,
  &&
  K(B, b) = (B, b_\flat) \,,
\end{align*}
where $b_\flat \colon T_{\U}B \to B$ is the unique map that, by
Theorem~\ref{th:freeQAlgebra}, 
satisfies the equations $b_\flat \circ \eta^{\U}_B = id_B$ and 
$b_\flat \circ \psi^{\U}_B = b \circ \Sigma b_\flat$.
(for the details on the proof \cf\ \cite[Theorem~4.2]{BacciMPP18}).

Next we show that the obvious point-wise extension of the above functors 
on the categories of bialgebras $\KK{(\Sigma, \U) \otimes (\Sigma', \U')}$ and $\tensorAlg{T_{\U}}{T_{\U'}}$
is an isomorphism of categories.

Clearly, since $H$ and $K$ are inverse with each other, so are their point-wise extensions. We are left to prove is that $H$ and $K$ are well defined.

Let $(A,a,b) \in \tensorAlg{T_{\U}}{T_{\U'}}$. We need to check that 
condition \eqref{eq:commDiagram} is satisfied by 
$(A, a \circ \psi^{\U}_A \circ \Sigma\eta^{\U}_A, b \circ \psi^{\U}_A \circ \Sigma\eta^{\U}_A)$.
Let $(A,b)^{\underline{n}} = (A^{\underline{n}}, \bar{b})$. 
By Corollary~\ref{cor:termComm} and Propositions~\ref{prop:symbolicOperations},  \ref{prop:symbolicOperationsMonad}, we have 
that the bottom square diagram below commutes for all $f \colon n \in \Sigma$ and all $g \colon m \in \Sigma'$,
while the top commute by Proposition~\ref{prop:Iso}:
\begin{equation*}
\begin{tikzcd}
(A^{\underline{m}})^{\underline{n}} \arrow[d, "\chi"'] \arrow[dr, "(b^{\tupl{g}})^{\underline{n}}", bend left]\\
(A^{\underline{n}})^{\underline{m}} \arrow[r, "\bar{b}^{\tupl{g}}"] \arrow[d, "(a^{\tupl{f}})^{\underline{m}}"'] 
& A^{\underline{n}}  \arrow[d, "a^{\tupl{f}}"] \\
A^{\underline{m}} \arrow[r, "b^{\tupl{g}}"] & A
\end{tikzcd}
\end{equation*}
Since $a^{\tupl{f}} = (a \circ \psi^{\U}_A \circ \Sigma\eta^{\U}_A)^f$ and 
$b^{\tupl{g}} = (b \circ \psi^{\U}_A \circ \Sigma\eta^{\U}_A)^g$, the above diagram proves 
that condition \eqref{eq:commDiagram} holds.

Let $(A,a,b) \in \KK{(\Sigma, \U) \otimes (\Sigma', \U')}$. We need to show that $(A,a_\flat,b_\flat)$
is a tensor $\tupl{T_{\U}, T_{\U'}}$-bialgebra. 
By Corollary~\ref{cor:termComm}, it is sufficient 
to prove that the following diagram commutes for all $g \colon m \in \Sigma'$, 
\begin{equation}
\begin{tikzcd}
T_{\U} A^{\underline{m}} \arrow[r, "\overline{a_\flat}"] \arrow[d, "T_{\U}b_\flat^{\tupl{g}}"'] 
& A^{\underline{m}}  \arrow[d, "b_\flat^{\tupl{g}}"] \\
T_{\U} A \arrow[r, "a_\flat"] & A
\end{tikzcd}
\label{eq:goal}
\end{equation}
where $(A,a_\flat)^{\underline{m}} = (A^{\underline{m}}, \overline{a_\flat})$. 

Toward proving \eqref{eq:goal}, first notice that the diagram below commutes 
for all $f \colon n \in \Sigma$ and $g \colon m \in \Sigma'$ 
\begin{equation}
\begin{tikzcd}
(A^{\underline{m}})^{\underline{n}} \arrow[rr, "\bar{a}^f"] 
	\arrow[rd, "\chi^{-1}", bend left=10] \arrow[dd, "(b^g)^{\underline{n}}"'] 
 & & A^{\underline{m}} \arrow[dd, "b^g"] \\
 & (A^{\underline{n}})^{\underline{m}} \arrow[lu, "\chi", bend left=10] 
 	\arrow[ru, "(a^f)^{\underline{m}}"', near start] \\
A^{\underline{n}} \arrow[rr, "a^f"] \arrow[ur, phantom, "\eqref{eq:commDiagram}" gray] & & A
\end{tikzcd}
\label{eq:bg-homo}
\end{equation}
for $(A,a)^{\underline{m}} = (A^{\underline{m}}, \bar{a})$ and $(A,b)^{\underline{n}} = (A^{\underline{n}}, \bar{b})$. Indeed, the bottom commutes because $(A,a,b)$ satisfies \eqref{eq:commDiagram}, and the top triangle does by Proposition~\ref{prop:Iso}. Thus, by Propositions~\ref{prop:symbolicOperations}, \ref{prop:symbolicOperationsMonad} and \eqref{eq:bg-homo} we have that $b^g$ is a $\Sigma$-homomorphism from $(A^{\underline{m}}, \bar{a})$ to $(A,a)$.
Moreover, as shown below, $b^g = b_\flat^{\tupl{g}}$:
\begin{align*}
    b^g 
    &= b \circ in_g \tag{def. $b^g$}\\
    &= b_\flat \circ \psi^{\U'}_A \circ \Sigma'\eta^{\U'}_A  \circ in_g \tag{$HK = Id$} \\
    &= b_\flat^{\tupl{g}} \,. \tag{def. $b_\flat^{\tupl{g}}$}
\end{align*}

Going back to proving \eqref{eq:goal}, by Theorem~\ref{th:freeQAlgebra}, it suffices to show 
that both $b_\flat^{\tupl{g}} \circ \overline{a_\flat}$ and $a_\flat \circ T_{\U}b_\flat^{\tupl{g}}$ are the (unique) 
homomorphic extension of $a$ along $b_\flat^{\tupl{g}}$.
This is shown by the following diagrams
\begin{equation*}
\begin{tikzcd}
A^{\underline{m}} \arrow[r, "\eta^{{\U}}"] \arrow[d, "b_\flat^{\tupl{g}}"']
& T_{\U} A^{\underline{m}}  \arrow[d, "T_{\U} b_\flat^{\tupl{g}}"] 
& & \Sigma T_{\U} A^{\underline{m}}  \arrow[ll, "\psi^{\U}"'] \arrow[d, "\Sigma T_{\U}b_\flat^{\tupl{g}}"] 
\\
A \arrow[r, "\eta^{\U}"] \arrow[dr, "id"'] 
& T_{\U} A \arrow[d, "a_\flat"] 
& & \Sigma T_{\U}A \arrow[ll, "\psi^{\U}"'] \arrow[d, "\Sigma a_\flat"] 
\\
& A
& & \Sigma A \arrow[ll, "a"']
\end{tikzcd}
%\quad
\begin{tikzcd}
A^{\underline{m}} \arrow[r, "\eta^{{\U}}"] \arrow[dr, "id"'] \arrow[ddr, "b_\flat^{\tupl{g}}"', bend right]
& T_{\U} A^{\underline{m}}  \arrow[d, "\overline{a_\flat}"] 
& & \Sigma T_{\U} A^{\underline{m}}  \arrow[ll, "\psi^{\U}"'] \arrow[d, "\Sigma \overline{a_\flat}"] \\
& A^{\underline{m}} \arrow[d, "b_\flat^{\tupl{g}}"] 
& & \Sigma A^{\underline{m}} \arrow[d, "\Sigma b_\flat^{\tupl{g}}"] \arrow[ll, "\overline{a}"'] \\
& A
& & \Sigma A \arrow[ll, "a"']
\end{tikzcd}
\end{equation*}
that commute by definition of $a_\flat$; naturality of 
$\eta^{\U}$, $\psi^{\U}$; since $(A^{\underline{m}}, \overline{a_\flat}) = (K(A,a))^{\underline{m}} = K((A,a)^{\underline{m}}) = (A^{\underline{m}}, \overline{a}_\flat)$; and because $b_\flat^{\tupl{g}}$ is a $\Sigma$-homomorphism from $(A^{\underline{m}}, \bar{a})$ to $(A,a)$.
\end{proof}

By combining the above two propositions we get the main theorem of this section.
\begin{thm} \label{th:tensorofsimpleeqmonad}
Let $\U, \U'$ be basic quantitative theories.  Then, the monad $T_{\U \otimes \U'}$ in $\Met$
is the tensor of monads $T_{\U} \otimes T_{\U'}$.
\end{thm}
\begin{proof}
  By Propositions~\ref{prop:bialgebrasTensor} and~\ref{prop:tensorAlgebras} the
  forgetful functor from $\tensorAlg{T_{\U}}{T_{\U'}}$ to $\Met$ has a
  left adjoint and the monad generated by this adjunction is isomorphic to
  $T_{\U \otimes \U'}$.  Thus, by definition of tensor of monads, 
  $T_{\U \otimes \U'} \iso T_{\U} \otimes T_{\U'}$.
\end{proof}

The above results do not require any specific property of 
the $\Met$, apart from requiring the morphisms to be non-expansive maps. Thus, when the quantitative equational theories are continuous, we can reformulate  Theorem~\ref{th:tensorofsimpleeqmonad} to be valid in $\CMet$.

\begin{thm} \label{th:tensorofsimpleeqmonadComplete}
Let $\U, \U'$ be continuous quantitative theories.  
Then, $\CC T_{\U \otimes \U'}$ in $\CMet$ is the tensor of monads $\CC T_{\U} \otimes \CC T_{\U'}$.
\end{thm}
\begin{proof} %(of Theorem~\ref{th:tensorofsimpleeqmonadComplete})
The tensor $\U \otimes \U'$ of continuous theories is also continuous, so that, by 
Theorem~\ref{th:freeCQAlgebra}, the free monad on it in $\CMet$ is $\CC T_{\U \otimes \U'}$.
Moreover, by exploiting the universal property of~Theorem~\ref{th:freeCQAlgebra}, we 
can refactor the proofs of Propositions~\ref{prop:bialgebrasTensor} and \ref{prop:tensorAlgebras} 
to obtain the isomorphism $\CC\KK{\Sigma + \Sigma', \U \otimes \U'} \iso \tensorAlg{\CC T_{\U}}{\CC T_{\U'}}$. 
Thus, by definition of tensor of monads, $\CC T_{\U \otimes \U'} \iso \CC T_{\U} \otimes \CC T_{\U'}$.
\end{proof}

%%%%%%%%%%%%%%%%%%%%%%%%%%%%%%%%%%%%%%%%%%%%%%
\subsection{Tensor with Reader/Writer Effects}
\label{sec:tensorReaderWriter}
As an example of commutative combination of effects we consider the operation of tensoring a generic 
quantitative theory with the quantitative reader and writer theories, respectively.  Similarly to Hyland et al.~\cite{HylandLPP07}, we show that these operations corresponds, at the level of monads, to the so called reader and writer monad transformers of Moggi and Cenciarelli~\cite{Moggi91,Cenciarelli93}.

\subsubsection*{Reader Monad Transformer}
Let $T$ be a strong monad with strength $t$ and $E$ a finite set.

The strength $t$ gives rise to a distributive law of the monad $T$ over the monad $(-)^{\underline{E}}$
\begin{equation*}
\lambda_X \colon TX^{\underline{E}} \nat (TX)^{\underline{E}}
\end{equation*} 
obtained by currying $T\ev_X^{\underline{E}} \circ t_{{\underline{E}},X^{\underline{E}}}$.
As distributive laws induce a notion of monad composition~\cite{Beck69}, Moggi's \emph{reader monad 
transformer} 
\begin{equation*}
T \mapsto (T-)^{\underline{E}}
\end{equation*} 
is also available in $\Met$. The following says that we can recover this monad transformer
as the operation of tensoring with the reader monad.
\begin{thm}[Tensoring with Reader Monad] \label{th:readercomposition}
Let $T$ be a strong monad.  Then, $T \otimes (-)^{\underline{E}}$ exists and is given 
as the monad composition $(T-)^{\underline{E}}$.
\end{thm}
\begin{proof}
Recall that the composite $(T-)^{\underline{E}}$ is the monad that arises 
from the adjunction with the forgetful functor $\distAlg{\lambda} \to \Met$, 
where $\distAlg{\lambda}$ denotes the full subcategory of 
EM $\tupl{T, (-)^{\underline{E}}}$-bialgebras $(A,a,b)$ satisfying the
commutativity of the diagram 
\begin{equation}
\begin{tikzcd}
TA \arrow[r, "a"] & A & A^{\underline{E}} \arrow[l, "b"'] \\
T(A^{\underline{E}}) \arrow[rr, "\lambda"] \arrow[u, "Tb"] 
& & (TA)^{\underline{E}} \arrow[u, "a^{\underline{E}}"'] 
\end{tikzcd}
\label{eq:lambdaDiagram}
\end{equation}
The bialgebras satisfying \eqref{eq:lambdaDiagram} are called, $\lambda$-bialgebras for the law $\lambda \colon T(-^{\underline{E}}) \nat (T-)^{\underline{E}}$ (see e.g., \cite{Beck69}).
We show that the category of $\lambda$-bialgebras is identical to 
the category of tensor $\tupl{T \otimes (-)^{\underline{E}}}$-bialgebras, 
that is, that the commutativity of the diagram above corresponds to either one of 
the equivalent conditions from Proposition~\ref{prop:comm}.

One direction is easy, as if we assume $(A,a,b)$ to be a tensor 
$\tupl{T \otimes (-)^{\underline{E}}}$-bialgebra, then~\eqref{eq:lambdaDiagram}
is just \eqref{comm2} from Proposition~\ref{prop:comm}
for $h = id \in \Op[(-)^{\underline{E}}](\underline{E})$ as, by definition of $\underline{E}$-power algebra,
$(A, a)^{\underline{E}} = (A^{\underline{E}}, a^{\underline{E}} \circ \lambda_A)$.

For the converse direction, assume \eqref{eq:lambdaDiagram} holds and let $g \in \Op[T](v)$,
for some $v \in \Met$.  Then, asking that $a^g$ is a $(-)^{\underline{E}}$-homomorphism (\ie, condition \eqref{comm1} from Proposition~\ref{prop:comm}) corresponds to 
the commutativity of the following diagram, as $(A, b)^v = (A^v, b^v \circ \sigma_A)$ and 
$(A, a)^{\underline{E}} = (A^{\underline{E}}, a^{\underline{E}} \circ \lambda_A)$:
\begin{equation*}
\begin{tikzcd}
(A^v)^{\underline{E}} \arrow[r, "\sigma"] \arrow[d, "g^{\underline{E}}"'] 
& (A^{\underline{E}})^v \arrow[r, "b^v"] \arrow[d, "g"']  
& A^v \arrow[d, "g"'] \\
(TA)^{\underline{E}} \arrow[dr, "a^{\underline{E}}"', near start, bend right] 
& T(A^{\underline{E}}) \arrow[r, "Tb"] \arrow[d, "a^{\underline{E}} \circ \lambda"'] 
& TA \arrow[d, "a"] \\
& A^{\underline{E}} \arrow[r, "b"] & A 
\end{tikzcd}
\end{equation*}
The bottom-right square is \eqref{eq:lambdaDiagram}, so commutes by hypothesis; 
the top-right square commutes by naturality of $g$; and finally, the left diagram 
commutes by Proposition~\ref{prop:Iso} as, by definitions of the strengths of $(-)^v$ and 
$(-)^{\underline{E}}$, $\sigma \colon (A^v)^{\underline{E}} \nat (A^{\underline{E}})^v$ 
coincides with the canonical isomorphism (denoted as $\chi$ in Proposition~\ref{prop:Iso}).

Therefore, as the two categories of bialgebras coincide, by definition of tensor 
of monads, $T \otimes (-)^{\underline{E}} = (T-)^{\underline{E}}$.  
\end{proof}

By using the above result in combination with Theorem~\ref{th:tensorofsimpleeqmonad}, we obtain an
analogous transformer at the level of quantitative equational theories as follows.
\begin{cor} \label{cor:readertransfomer}
Let $\U$ be a basic quantitative equational theory.  
Then, $(T_{\U} -)^{\underline{E}}$ is the free monad on the theory $\U \otimes \R[]$ in $\Met$.
\end{cor}

Moreover, as $\R[]$ is a continuous theory, by Theorems~\ref{th:isoReaderMonadCompletion}, \ref{th:tensorofsimpleeqmonad}, and \ref{th:readercomposition}, we obtain the following variant of the 
quantitative reader theory transformer on continuous  theories.
\begin{cor}
Let $\U$ be a continuous quantitative theory.  
Then, $(\CC T_{\U} - )^{\underline{E}}$ is the free monad on the theory $\U \otimes \R[]$ in $\CMet$.
\end{cor}

\subsubsection*{Writer Monad Transformer}
Let $T$ be a strong monad with strength $t$ and $(\Lambda, *, 0)$ a monoid 
structure with $\Lambda \in \Met$, unit $0 \in \Lambda$, and non-expansive multiplication $* \colon \Lambda \times \Lambda \to \Lambda$.

The strength $t$ gives rise to a canonical distributive law of the monad 
$(\Lambda \mprod -)$ over $T$ as
\begin{equation*}
t_{\Lambda, -} \colon (\Lambda \mprod T-) \nat T(\Lambda \mprod -) \,.
\end{equation*}  
So the composite $T(\Lambda \mprod -)$ acquires a canonical monad structure via the above distributive law~\cite{Beck69}, and we 
obtain the following version of Moggi's \emph{writer monad transformer} in $\Met$:
\begin{equation*}
T \mapsto T(\Lambda \mprod -) \,.
\end{equation*}

Hyland et al.~in~\cite{HylandPP06} observed that Moggi's writer monad transformer
can be equivalently recovered as the operation of tensoring with the writer monad.

\begin{thm}[Tensoring with Writer Monad~{\cite[Theorem~12]{HylandPP06}}] \label{th:writercomposition}
Let $T$ be a strong monad with countable rank.  Then, the monad composition $T(\Lambda \mprod -)$ 
is given as $T \otimes (\Lambda \mprod -)$.
\end{thm}

As any quantitative theory $\U$ induces a 
monad $T_{\U}$ with countable rank (\cf~Ford et al.~\cite{FordMS21}), by combining the above with Theorems~\ref{th:tensorofsimpleeqmonad} and \ref{th:isoWriterMonad}, 
we get an analogous transformer at the level of quantitative equational theories as follows:
\begin{cor} \label{cor:writerTheoryTransformer}
Let $\U$ be a basic quantitative theory.  
Then, $T_{\U}(\Lambda \mprod -)$ is the free monad on the theory $\U \otimes \Wr[]$ in $\Met$.
\end{cor}

As $\Wr[]$ is also a continuous quantitative theory, by similar arguments as before, 
we obtain the following variant of quantitative writer theory transformer on continuous theories.
\begin{cor}
Let $\U$ be a continuous quantitative theory.  
Then, $\CC T_{\U}(\Lambda \mprod -)$ is the free monad on the theory $\U \otimes \Wr[]$ in $\CMet$.
\end{cor}

%%%%%%%%%%%%%%%%%%%%%%%%%%%%%%%%%%%%%%%%%%%%%%
\subsection{The Algebras of Labeled Markov Processes}
\label{sec:labelmarkovprocesses}
In this section, we provide a quantitative equational axiomatization of labelled Markov processes with their discounted bisimilarity metric~\cite[Section~6]{BreugelHMW07}. 

\subsubsection{Labelled Markov Processes over Metric Spaces}
\label{sec:lmp&bisim}
Let $A$ be a finite set of action labels. 
As in~\cite[Section~6]{BreugelHMW07}, we regard $A$-labelled Markov processes 
over extended metric spaces as coalgebras on the category of metric spaces.
In detail, we consider two variants of labelled Markov processes:
\begin{align*}
  X &\longrightarrow \Pi(c \cdot X + 1 )^{\underline{A}}  \quad \text{in $\Met$} \,, \\
  X &\longrightarrow \Delta(c \cdot X + 1 )^{\underline{A}}  \quad \text{in $\CMet$} \,,
\end{align*}
where $\Pi$ and $\Delta$ are the functors from Section~\ref{sec:probchoice}, mapping a 
metric space $X$ to a space of probability measures with Kantorovich metric. We will collectively refer to these coalgebras as labelled $c$-Markov processes.

Similarly to~Section~\ref{sec:mp&bisim}, the use of the rescaling functor $(c \cdot -)$ is 
to encompass the case where the probabilistic bisimilarity distance is discounted 
by a factor $0 < c < 1$. This will not change the essence of the results 
from~\cite{BreugelHMW07} that are used in this section to characterize 
the probabilistic bisimilarity metric. 

In~\cite{BreugelHMW07}, van Breugel et al.\ characterized the bisimilarity distance 
on labelled Markov processes as the pseudometric induced by the unique homomorphism to the final coalgebra. Specifically, the \emph{$c$-discounted bisimilarity pseudometric} on a labelled $c$-Markov process $(X, \tau)$ is obtained as the function $\dist \colon X \times X \to [0,1]$ given as
\begin{equation*}
  \dist(x,x') = d_{Z}(h(x), h(x')) \,,
\end{equation*}
where $h \colon X \to {Z}$ is the unique homomorphism to the final labelled $c$-Markov 
process $(Z,\omega)$.

This distance has a characterization as the least fixed point of a 
monotone function on a complete lattice of $1$-bounded pseudometrics.
\begin{propC}[{\cite[Theorem~40]{BreugelHMW07}}] \label{prop:fixpointbisimdistLabelled}
The $c$-discounted bisimilarity pseudometric $\dist$ on $(X, \tau)$ 
is the \emph{unique} fixed point
of the following operator on the complete lattice of extended pseudometrics $d$ on $X$
with point-wise order $\sqsubseteq$, such that $d \sqsubseteq d_X$,
\begin{equation*}
  \Psi^c(d)(x,x') = \sup_{a \in A} \; \sup_{f} \left| \lebint{f}{\tau(x)(a)} - \lebint{f}{\tau(x')(a)} \right| \,,
\end{equation*}
with $f$ ranging over non-expansive positive $1$-bounded
real valued functions on $c \cdot X + 1$.
\end{propC}

\subsubsection{Quantitative Algebraic Presentation}
\label{sec:presentationLMP}

We provide a quantitative equational theory that 
axiomatizes (the monad of) $A$-labelled Markov processes with 
$c$-discounted bisimilarity metric. We do this by extending the axiomatization of (unlabelled) Markov processes from Section~\ref{sec:markovprocesses}  with a new ``reading'' operator used to describe 
the reaction to the choice of a label from a finite set $A$ of action labels.
As expected, the reading operations will be axiomatized by the theory $\R[A]$ of reading computations (\cf~Section~\ref{sec:readeralgebras}). 

\medskip
Formally,  for $A = \ens{a_1, \dots, a_n}$ we define the quantitative theory of labelled Markov processes as the following combination of quantitative theories,
\begin{equation*}
  \U_{\textbf{LMP}} = ((\B + \mathcal{E}_{1}) \otimes \R[A]) + \O{\Sigma_{\diamond}} \,.
\end{equation*}
with signature $\Sigma_{\textbf{MP}} = \Sigma_{\B} \cup \Sigma_{1} \cup \Sigma_{\R[A]} 
\cup \Sigma_{\diamond}$ given 
as the disjoint union of those from its component theories. Explicitly, 
\begin{equation*}
  \Sigma_{\textbf{LMP}}  = \set{ +_e \colon 2}{ e \in [0,1]} \cup 
  \{ \textsf{raise}_* \colon 0 \} \cup \{ \rd \colon |A| \} \cup \{ \diamond \colon \tupl{1, c} \}
\end{equation*}
theory $\U_{\textbf{LMP}}$ is given by the following set of axioms 
\begin{align*} 
(\Bone)\,
& \vdash x +_1 y =_0 x \,, \\
(\Btwo)\, 
& \vdash x +_e x =_0 x \,, \\
(\SC)\,
& \vdash x +_e y =_0 y +_{1-e} x \,, \\
(\SA)\,
& \vdash (x +_e y) +_{e'} z =_0 x +_{ee'} (y +_{\frac{e' - ee'}{1 - ee'}} z) \,, \text{ for $e,e' \in [0,1)$} \,, 
\\ 
(\IB)\,
& \{ x \,{=_\e}\, y, x' \,{=_{\e'}}\, y' \} \,{\vdash}\, x +_e x' \,{=_{\delta}}\, y +_e y', \, \text{for $\delta \geq e \e + (1-e) \e'$,}
\\
(\Idem)\,
& \vdash x =_0 \rd(x, \dots, x) \,, \\
(\Diag)\,
& \vdash \rd(x_{1,1}, \dots, x_{n,n} ) =_0 
	\rd( \rd(x_{1,1}, \dots, x_{1,n}), \dots, \rd(x_{n,1}, \dots, x_{n,n}) )
\\
(\Com)\,
& \vdash \rd(x_1 +_e y_1, \dots, x_n +_e y_n) =_0 \rd(x_1, \dots, x_n) +_e \rd(y_1, \dots, y_n) \,, 
\\
(\Lip{\diamond})\,
& \{x =_\e y \} \vdash \diamond(x) =_{\delta} \diamond(y) \,,
\text{ for $\delta \geq c \e$} \,.
\end{align*}
Note that, the constant $\textsf{raise}_*$ has no explicit associated axiom since $\mathcal{E}_{1}$ is the trivial theory and  (\Idem) already implies the commutativity axiom required by tensoring with $\R[A]$.

Intuitively, $\Sigma_{\textbf{LMP}}$-terms (modulo $=_0$ provability) should be interpreted as equivalence classes of behaviours of labelled Markov processes up-to bisimilarity. The term $t +_e t'$ expresses convex combination of behaviours; $\textsf{raise}_*$ represents termination; $\rd(t_1, \dots, t_n)$ is used to express that $t_i$ is the selected behaviour after the choice of the action label $a_i \in A$; and $\diamond(t)$ expresses the ability of taking a transition 
to the behaviour $t$.

\paragraph*{On Metric Spaces}
\label{sec:labelmarkovprocessesMet}

We characterize the monad $T_{\U_{\textbf{LMP}}}$ on $\Met$  in steps, by explaining
the contribution of the different theories in
\begin{equation*}
  \U_{\textbf{LMP}} = ((\B + \mathcal{E}_{1}) \otimes \R[A]) + \O{\Sigma_{\diamond}} \,.
\end{equation*}

(Step 1) As shown in Section~\ref{sec:presentationMP}, $T_{\B + \mathcal{E}_1}$
is the \emph{finitely supported sub-distribution monad}
\begin{equation*}
T_{\B + \mathcal{E}_1} \iso \Pi(- + 1) \,.
\end{equation*}
Thus, $\B + \mathcal{E}_1$ axiomatizes finitely supported sub-distributions with Kantorovich metric.

(Step 2) By Theorem~\ref{th:tensorofsimpleeqmonad} and Corollary~\ref{cor:readertransfomer},
we further get the monad isomorphism
\begin{equation*}
T_{(\B + \mathcal{E}_1) \otimes \R[A]} \iso 
\Pi(1+ -) \otimes (-)^{\underline{A}} \iso 
(\Pi(1+ -))^{\underline{A}} \,,
\end{equation*}
saying that tensoring with the theory $\R[A]$ of reading computations corresponds
to axiomatically adding the capability of reacting to the choice of 
an action label.

(Step 3) The final step is to sum the above with the theory $\O{\Sigma_\diamond}$.
Then, by Corollary~\ref{cor:ResumptionTransfomerMet}, the monad on 
$\U_{\textbf{LMP}}$ is 
\begin{equation*}
T_{\U_{\textbf{LMP}}} 
\iso \mu y.  T_{(\B + \mathcal{E}_1) \otimes \R[A]}(c \cdot y + - ) 
\iso \mu y.  \Pi(c \cdot y + 1 + -)^{\underline{A}} \,,
\end{equation*}
where we implicitly applied the isomorphisms $c \cdot (A + B) \iso c \cdot A + c \cdot B$ and $1 \iso c \cdot 1$.

Explicitly, this means that the free monad on $\U_{\textbf{LMP}}$
assigns to an arbitrary metric space $X \in \Met$ the \emph{initial solution} 
of the following functorial equation in $\Met$
\begin{equation*}
  \FMP[X] \iso (\Pi(c \cdot \FMP[X] + 1 + X))^{\underline{A}} \,.
\end{equation*}
In particular, when $X = 0$ is the empty metric space (\ie, the initial object in $\Met$)
the above corresponds to the isomorphism on the initial $(\Pi(c \cdot - + 1))^{\underline{A}}\,$-algebra.
The isomorphism gives us also a $(\Pi(c \cdot - + 1))^{\underline{A}}\,$-coalgebra
structure $\tau_0 \colon \FMP[0] \to (\Pi(c \cdot \FMP[0] + 1))^{\underline{A}}$ on $\FMP[0]$.

The key observation is that the metric of $\FMP[0]$ is the bisimilarity metric.
\begin{lem} \label{lem:bisimMetricLabel}
$d_{\FMP[0]}$ is the $c$-discounted probabilistic bisimilarity metric on $(\FMP[0], \tau_0)$.
\end{lem}
\begin{proof}
Similar to Lemma~\ref{lem:bisimMetric}.
\end{proof}

\paragraph*{On Complete Metric Spaces}

Since all the quantitative theories considered are continuous, we can replicate 
the same steps also while interpreting the theory $\U_{\textbf{LMP}}$ over complete metric spaces, obtaining 
the monad
\begin{equation*}
\CC T_{\U_{\textbf{LMP}}} \iso \mu y. \Delta(c \cdot y + 1 + -)^{\underline{A}} \,.
\end{equation*}

By following similar arguments to Section~\ref{sec:presentationMP}, one can 
prove that the the functorial equation $\FMP[X] \iso \Delta(c \cdot \FMP[X] + 1 + X)^{\underline{A}}$ has a unique solution.
By applying the monad above on $X = 0$ we recover the carrier of the final $(\Delta(c \cdot - + 1))^{\underline{A}}\,$-coalgebra, equipped with $c$-discounted probabilistic bisimilarity metric. 

%%%%%%%%%%%%%%%%%%%%%%%%%%%%%%%%%%%%%%%%%%%%%%
\subsection{The Algebras of Mealy Machines}
\label{sec:mealyMachines}

In a similar spirit to the axiomatization of labelled Markov processes, here we
provide a quantitative axiomatization of Mealy machines with their (coalgebraically defined) discounted bisimilarity metric.

\subsubsection{Mealy machines over Metric Spaces}
\label{sec:mealyAutomata&bisim}

Informally, \emph{Mealy machines} are deterministic automata with outputs. 
Formally, they are tuples $(X, I, \Lambda, t, o)$ consisting of a set of states $X$, a 
finite set $I = \ens{i_1, \dots, i_n}$ of \emph{inputs}, a set $\Lambda$ of 
\emph{outputs}, a \emph{transition function} $t \colon X \times I \to X$, and an \emph{output function} $o \colon X \times I \to \Lambda$.

These structures are clearly $\Set$ coalgebras for the functor $(\Lambda \times -)^I$~\cite{Rutten06,SilvaBBR13}. In order to give a coalgebraic definition of a 
bisimilarity metric for Mealy machines, we will interpret them as coalgebras $(X,\tau)$ on categories of metric spaces. Specifically
\begin{align*}
  \tau \colon X &\longrightarrow (c \cdot X \mprod \Lambda)^{\underline{I}}  \quad 
  \text{in $\Met$/$\CMet$} \,,
\end{align*}
where  $0 < c < 1$ and we assume $\Lambda$ to be a complete metric space of outputs with a monoid structure.
The rescaling functor $(c \cdot -)$ is used to obtain a discounted bisimilarity 
distance. When we want to emphasize the r\^{o}le of the discount factor
we call these coalgebras $c$-Mealy machines.

Similarly to \cite{BreugelHMW07}, we define the the \emph{$c$-discounted bisimilarity pseudometric} on a $c$-Mealy machine $(X, \tau)$ as the pseudometric 
induced by the unique homomorphism to the final coalgebra. That is,
\begin{equation*}
  \dist(x,x') = d_{Z}(h(x), h(x')) \,,
\end{equation*}
where $h \colon X \to {Z}$ is the unique homomorphism to the 
final $c$-Mealy machine $(Z,\omega)$.

A concrete characterization of the final $c$-Mealy machine can be obtained as in~\cite{Rutten06}. We don't repeat the argument here as it is not necessary for our technical development, which requires only its existence. 

This distance has a characterization as the least fixed point of a 
monotone function on a complete lattice of $[0,\infty]$-valued pseudometrics.
\begin{prop} \label{prop:fixpointbisimdistMealy}
The $c$-discounted bisimilarity pseudometric $\dist$ on $(X, \tau)$ 
is the \emph{unique} fixed point
of the following operator on the complete lattice of extended pseudometrics $d$ on $X$
with point-wise order $\sqsubseteq$, such that $d \sqsubseteq d_X$,
\begin{equation*}
  \Psi^c(d)(x,x') = \sup_{i \in I} \left( c \cdot d(x_i, x'_i) + d_\Lambda(\lambda_i, \lambda'_i)  \right)  \,,
\end{equation*}
where $\tau(i)(x) = (x_i, \lambda_i)$ and $\tau(i)(x') = (x'_i, \lambda'_i)$.
\end{prop}
\begin{proof}
The uniqueness of the fixed point follows by Banach fixed point theorem. Indeed, the set of extended real valued functions on $X \times X$
(which is a superset of the set of extended pseudometrics on $X$) can be turned into a complete Banach space 
by means of the sup-norm $|| f || = \sup_{x,x'} |f(x,x')|$ and $\Psi^c$ is a $c$-contractive operator on it. Moreover, $\dist = \lim_{n\to \infty} (\Psi^c)^n(\mathbf{0})$,
where $\mathbf{0}$ is the constantly $0$ pseudometric. Since $\Psi^c$ is a monotone operator, $(\Psi^c)^n(\mathbf{0}) \sqsubseteq (\Psi^c)^{n+1}(\mathbf{0})$. 
Moreover, $\Psi^c$ maps pseudometrics into pseudometrics. As pseudometrics are closed under point-wise suprema, $\dist$ is a pseudometric.
\end{proof}

\subsubsection{Quantitative Algebraic Presentation}
\label{sec:presentationMealy}

Next we provide a quantitative equational theory that axiomatizes (the monad of) 
Mealy machines with $c$-discounted bisimilarity metric.
As we did already in the previous sections we will do this by combining simpler theories via of sum and tensor. The basic theories we use are:
\begin{enumerate}
\item 
\emph{The quantitative theory $\R[I]$ of reading computations}
will be used to axiomatize the reaction to the choice of an input symbol $i \in I$
(\cf~Section~\ref{sec:readeralgebras});

\item 
\emph{The quantitative theory $\Wr[\Lambda]$ of writing computations} will be used 
to describe the action of outputing a symbol $\alpha \in \Lambda$. (\cf~Section~\ref{sec:writeralgebras}). In our axiomatic interpretation, we assume 
$\Lambda$ to have a monoid structure and outputs to be recorded in an ``output tape'' by means of writing operations.

\item
\emph{The quantitative theory of contractive operators} $\O{\Sigma_{\diamond}}$
with signature $\Sigma_{\diamond} = \ens{\diamond \colon \tupl{1,c} }$ will be used to axiomatize the transition to a next state with discounting factor $0 < c < 1$ (\cf\ Section \ref{sec:controperators}).
\end{enumerate}

\medskip
Formally, for a finite set of inputs $I = \ens{i_1, \dots, i_n}$ and complete metric space $\Lambda$ of outputs with monoid structure $(\Lambda, 0, *)$, we define the quantitative theory of Mealy machines as the following combination of quantitative theories,
\begin{equation*}
  \U_{\textbf{MM}} = (\R[I] \otimes \Wr[\Lambda]) + \O{\Sigma_{\diamond}} \,.
\end{equation*}
with signature $\Sigma_{\textbf{MM}} = \Sigma_{\R[I]} \cup \Sigma_{\Wr[\Lambda]}
\cup \Sigma_{\diamond}$ given 
as the disjoint union of those from its component theories. Explicitly, 
\begin{equation*}
  \Sigma_{\textbf{MM}}  = \{ \rd \colon |I| \} \cup \{ \wrt{\alpha} \colon 1 \mid \alpha \in \Lambda  \} \cup \{ \diamond \colon \tupl{1, c} \}
\end{equation*}
and the theory $\U_{\textbf{MM}}$ is given by the following axioms 
\begin{align*}
(\Idem)\,
& \vdash x =_0 \rd(x, \dots, x) \,, 
\\
(\Diag)\,
& \vdash \rd(x_{1,1}, \dots, x_{n,n} ) =_0 
	\rd( \rd(x_{1,1}, \dots, x_{1,n}), \dots, \rd(x_{n,1}, \dots, x_{n,n}) )
\\
(\Zero) \,
& \vdash x =_0 \wrt{0}(x) \,, 
\\
(\Mult) \,
& \vdash \wrt{\alpha}( \wrt{\alpha'}(x) ) =_0  \wrt{\alpha * \alpha'}(x) \,, 
\\
(\Diff) \,
& \{x =_\e x'\} \vdash \wrt{\alpha}(x) =_\delta  \wrt{\alpha'}(x')\,, 
\text{ for $\delta \geq d_\Lambda(\alpha, \alpha') + \e $} \,,
\\
(\Com)\,
& \vdash \rd(\wrt{\alpha}(x_1), \dots, \wrt{\alpha}(x_n)) =_0 
\wrt{\alpha}( \rd(x_1, \dots, x_n) ) \,, 
\\
(\Lip{\diamond})\,
& \{x =_\e y \} \vdash \diamond(x) =_{\delta} \diamond(y) \,,
\text{ for $\delta \geq c \e$} \,.
\end{align*}

Intuitively, $\Sigma_{\textbf{MM}}$-terms (modulo $=_0$ provability) should be interpreted as equivalence classes of behaviours of Mealy machines up-to bisimilarity. The term $\rd(t_1, \dots, t_n)$ is used to express that $t_k$ is the selected behaviour after reading input $i_k \in I$; $\wrt{\alpha}(t)$ is the term expressing behaviour of writing the output $\alpha \in \Lambda$ in the output tape; and $\diamond(t)$ expresses the ability of taking a transition to the behaviour $t$.

\paragraph*{On Metric Spaces}
\label{sec:MealyMachinesMet}

We characterize the monad $T_{\U_{\textbf{MM}}}$ on $\Met$  in steps, by explaining
the contribution of the different theories in $\U_{\textbf{MM}}$.

(Step 1) As shown in Section~\ref{sec:readeralgebras}, $T_{\R[I]}$
is the \emph{reader monad}
\begin{equation*}
T_{\R[I]} \iso (-)^{\underline{I}} \,.
\end{equation*}
Thus, $\R[I]$ axiomatizes the space of functions with domain the set $I$.

(Step 2) By Theorem~\ref{th:tensorofsimpleeqmonad} and Corollary~\ref{cor:writerTheoryTransformer} (equivalently, Corollary~\ref{cor:readertransfomer}), we further get the monad isomorphisms
\begin{equation*}
T_{\R[I] \otimes \Wr[\Lambda]} \iso (-)^{\underline{I}} \otimes (- \mprod \Lambda) \iso 
(- \mprod \Lambda)^{\underline{I}} \,,
\end{equation*}
saying that tensoring with the theory $\Wr[\Lambda]$ of writing computations corresponds to axiomatically adding the capability of writing an output symbol after
reading an input action.

(Step 3) By summing the above theories with the theory $\O{\Sigma_\diamond}$,
by Corollary~\ref{cor:ResumptionTransfomerMet}, we get that the free monad on 
$\U_{\textbf{MM}}$ is 
\begin{equation*}
T_{\U_{\textbf{MM}}} 
\iso \mu y.  T_{\R[I] \otimes \Wr[\Lambda]}(c \cdot y + - ) 
\iso \mu y.  ((c \cdot y  + -) \mprod \Lambda)^{\underline{I}} \,.
\end{equation*}

Explicitly, the free monad on $\U_{\textbf{MM}}$
assigns to an arbitrary metric space $X \in \Met$ the \emph{initial solution} 
of the following functorial equation in $\Met$
\begin{equation*}
  \MM[X] \iso (c \cdot \MM[X] + X) \mprod \Lambda)^{\underline{I}} \,.
\end{equation*}
In particular, when $X = 0$ is the empty metric space
the above corresponds to the isomorphism of the initial $(c \cdot - \mprod \Lambda)^{\underline{I}}\,$-algebra.
From this we recover a $(c \cdot - \mprod \Lambda)^{\underline{I}}\,$-coalgebra
structure $\tau_0 \colon \MM[0] \to (c \cdot \MM[0] \mprod \Lambda)^{\underline{I}} $ on $\MM[0]$, whence a $c$-Mealy machine.

\begin{lem} \label{lem:bisimMetricMealy}
$d_{\MM[0]}$ is the $c$-discounted probabilistic bisimilarity metric on $(\MM[0], \tau_0)$.
\end{lem}
\begin{proof}
Similar to Lemma~\ref{lem:bisimMetric}.
\end{proof}

\paragraph*{On Complete Metric Spaces}

As the quantitative theories considered are continuous, we can replicate 
the same steps also while interpreting the theory $\U_{\textbf{MM}}$ over complete metric spaces, obtaining 
the monad
\begin{equation*}
\CC T_{\U_{\textbf{MM}}} 
\iso \mu y.  ((c \cdot y  + -) \mprod \Lambda)^{\underline{I}} \,.
\end{equation*}

By following similar arguments to Section~\ref{sec:presentationMP}, one can 
prove that the the functorial equation $\MM[X] \iso (c \cdot \MM[X] + X) \mprod \Lambda)^{\underline{I}}$ has a unique solution in $\CMet$.
Hence, by applying the monad above on $X = 0$ we recover the carrier of the final 
$(c \cdot - \mprod \Lambda)^{\underline{I}}\,$-coalgebra, equipped with $c$-discounted probabilistic bisimilarity metric.

%%%%%%%%%%%%%%%%%%%%%%%%%%%%%%%%%%%%%%%%%%%%%%
\subsection{The Algebras of Markov Decision Processes with Rewards}
\label{sec:markovdecitionprocesses}
In this section we provide a quantitative equational axiomatization of Markov decision processes with rewards and their (coalgebraically defined) discounted bisimilarity metric. The axiomatization is obtained by extending that of labelled Markov processes
from Section~\ref{sec:labelmarkovprocesses} by adding the ability to record the rewards associated with a specific probabilistic decision.

\subsubsection{Markov Decision Processes over Metric Spaces}
\label{sec:mdp&bisim}

Informally, Markov decision processes are labelled Markov processes where
each choice of action label (decision) is associated with a probabilistic reward.
Formally, as in~\cite{BreugelHMW07}, we regard them as coalgebras on the category of extended metric spaces. In detail, we consider two variants of Markov decision processes:
\begin{align*}
  X &\longrightarrow \Pi(c \cdot X \mprod \reals )^{\underline{A}}  \quad \text{in $\Met$} \,, \\
  X &\longrightarrow \Delta(c \cdot X \mprod \reals)^{\underline{A}}  \quad \text{in $\CMet$} \,,
\end{align*}
where $\Pi$ and $\Delta$ are the functors from Section~\ref{sec:probchoice}.
For convenience, the rescaling functor $(c \cdot - )$ is used to account of a 
discount factor on the bisimilarity metric and the functor $(- \mprod \reals)$
is to give a metric interpretation to the combination with the reward structure.

\begin{rem}
In~\cite{Puterman2005} a Markov decision process is defined as a tuple $(S, p(\cdot | s,a), r(s, a))$ with a Markov kernel $p \colon S \times A \to \Delta(S)$ and randomised reward function $r \colon S \times A \to \Delta(\reals)$.
Our coalgebraic representation is the natural generalisation over metric spaces, where the 
randomness of the Markov kernel and reward function is combined as a probability measure on $(c \cdot S \mprod \reals)$, by regarding $\reals$ and $S$ as extended metric spaces.
\end{rem}

Similarly to Section~\ref{sec:mp&bisim}, one can show that the final coalgebra for 
the functors $\Pi(c \cdot - \mprod \reals )^{\underline{A}}$ in $\Met$ and 
$\Delta(c \cdot - \mprod \reals)^{\underline{A}}$ in $\CMet$ exists, thus we define 
the $c$-discounted probabilistic bisimilarity distance on a Markov decision process 
$(X, \tau)$ as the pseudometric 
\begin{equation*}
  \dist(x,x') = d_{Z}(h(x), h(x')) 
\end{equation*}
induced by the unique homomorphism $h \colon X \to {Z}$ to the final coalgebra.

Also in this case, the probabilistic bisimilarity distance can be given a fixed point
characterization.
\begin{prop}
The $c$-discounted bisimilarity pseudometric $\dist$ on $(X, \tau)$ 
is the \emph{unique} fixed point
of the following operator on the complete lattice of extended pseudometrics $d$ on $X$
with point-wise order $\sqsubseteq$, such that $d \sqsubseteq d_X$,
\begin{equation*}
  \Psi^c(d)(x,x') = \sup_{a \in A} 
  	\sup_{f} \left| \lebint{f}{\tau(x)(a)} - \lebint{f}{\tau(x')(a)} \right| \,,
\end{equation*}
with $f$ ranging over non-expansive positive $1$-bounded
real valued functions on $c \cdot X \mprod \reals$.
\end{prop}

%%%%%%%%%%%%%%%%%%%%%%%%%%%%%%%%%%%%%%%%

\subsubsection{Quantitative Algebraic Presentation}
\label{sec:presentationMDP}

We provide a quantitative axiomatization
of Markov decision processes with rewards equipped with discounted bisimilarity metric.
As the construction is similar to Section~\ref{sec:labelmarkovprocesses}, we avoid repeating the details of 
each step of the monad characterization. 

%We do this by extending what we have seen in Section~\ref{sec:labelmarkovprocesses} 
%with the theory of writing computations to encode rewards.

\medskip
Let $A = \ens{a_1, \dots, a_n}$ be a finite set of actions and $(\reals, +, 0)$ be the standard monoid structure on the reals. We define the quantitative theory $\U_{\textbf{MDP}}$ of Markov decision processes with real-valued rewards as the following combination of quantitative theories,
\begin{equation*}
  \U_{\textbf{MDP}} = ((\B \otimes \Wr[\reals]) \otimes \R[A]) + \O{\Sigma_{\diamond}} \,.
\end{equation*}
with signature $\Sigma_{\textbf{MDP}} = \Sigma_{\B} \cup \Sigma_{\Wr[\reals]} \cup \Sigma_{\R[A]} \cup \Sigma_{\diamond}$ given 
as the disjoint union of those from its component theories. 

\paragraph*{On Metric Spaces and Complete Metric Spaces}
\label{sec:markovdecisionprocesses}

Similarly to what we have done in for labelled Markov processes, 
we relate Markov decision processes and their $c$-discounted
probabilistic bisimilarity pseudometric with the free monads on the 
theory $\U_{\textbf{MDP}}$ in $\Met$ and $\CMet$.  

\smallskip
The only step that changes in the characterization of the monad $T_{\U_{\textbf{MDP}}}$ in $\Met$, 
regards the combination of theories $\B \otimes \Wr[\reals]$, which is dealt using Corollary~\ref{cor:writerTheoryTransformer}. Thus, similarly to Section~\ref{sec:labelmarkovprocesses} we get 
\begin{equation*}
T_{\U_{\textbf{MDP}}} = 
T_{((\B \otimes \Wr[\reals]) \otimes \R[A]) + \O{\Sigma_{\diamond}}} \iso 
\mu y.  \Pi((c \cdot y + - ) \mprod \reals)^{\underline{A}} \,.
\end{equation*}
The metric on the initial solution for the functorial fixed point definition 
corresponds to the $c$-discounted probabilistic bisimilarity (pseudo)metric on
its coalgebra structure.

\smallskip
Similar considerations apply also when interpreting the theories in the category $\CMet$
of complete metric spaces, as the argument follows without issues because $\reals$ 
a complete metric space.  Thus we obtain the following 
characterization for the monad:
\begin{equation*}
\CC T_{\U_{\textbf{LMP}}} \iso 
\mu y.  \Delta((c \cdot y + - ) \mprod \reals)^{\underline{A}} \,.
\end{equation*}
As the fixed point solution in $\CMet$ is unique, $\CC T_{\U_{\textbf{LMP}}}0$ is an algebraic characterization of the final $\Delta((c \cdot - ) \mprod \reals)^{\underline{A}}\,$-coalgebra with probabilistic bisimilarity metric.

%%%%%%%%%%%%%%%%%%%%%%%%%%%%%%%%%%%%%%%%
\section{Conclusions}
\label{sec:concl}

We studied the disjoint and commutative combinations of quantitative effects, respectively
as the sum and tensor of their quantitative equational theories.  The key results are
Theorems~\ref{th:sumofsimpleeqmonad} and \ref{th:tensorofsimpleeqmonad}, asserting that
the sum and tensor of two quantitative theories corresponds to the categorical sum and
tensor, respectively, of their free monads.  In addition to these general results, we
provide quantitative analogues Moggi's monad transformers for exceptions, resumption,
reader, and writer.

We illustrate the applicability of our theoretical development with the axiomatizations
four coalgebraic bisimilarity metrics: for Markov processes, labeled Markov processes,
Mealy machines, and Markov decision processes. Apart from the intrinsic interest in their
quantitative equational presentation as effects, what is particularly pleasant is the
systematic compositional way with which one can obtain quantitative axiomatizations of
different variants of coalgebraic structures by just combining theories as new basic
ingredients.

An example that escapes our compositional treatment via sum and tensor is the
combination of probabilities and non-determinism as illustrated
in~\cite{MioV20}.  A possible future work in this direction is to extend the
combination of theories with another operator: the distributive tensor (see~\cite[Section~6]{HylandP06}).
Following an intuition similar to Cheng~\cite{Cheng_2020}, we claim that these correspond in a suitable way to Garner's weak distributive law~\cite{Garner20}.  
Our claim seems promising in the light of the work~\cite{GoyP20,BonchiS20} 
which consider equational axiomatizations combining probabilities and non-determinism.

\section*{Acknowledgments}
  \noindent Giorgio Bacci wishes to acknowledge fruitful discussions with Dexter Kozen and Ugo Dal Lago during a workshop at the Bellairs Research Institute in Barbados.
  Radu Mardare was supported by the EPSRC-UKRI grant EP/Y000455/1 ``A correct-by-construction approach to approximate computation''.
  Prakash Panangaden’s research was supported by a grant from NSERC.

%% Bibliography
\bibliographystyle{alphaurl}
\bibliography{biblio}

\appendix

\section{Limits and Colimits of Extended Metric Spaces} 
\label{app:exmetric}

Limits and colimits in the category $\Met$ of extended metric spaces and non-expansive maps are defined
similarly to those in $\Set$, at least for the part of their underlying set. Some care, though, should to taken 
in the definition of the distance function.

From an abstract point of view, the reason is that the forgetful functor $U \colon \Met \to \Set$, sending
an extended metric space $(X,d_X)$ to its underlying set $X$, is faithful (\ie, homsets are mapped injectively), and has as left adjoint the functor $\mathit{Disc} \colon \Set \to \Met$ that assigns to 
each set $X$ the \emph{discrete} extended metric space $\underline{X}$ giving distance $\infty$ to every distinct pair of 
elements. Therefore $U$ preserves all limits which may exist in $\Met$ (this is \emph{why} the underlying
set of product spaces is the Cartesian product of their underlying sets).

Moreover, $\Met$ is a full reflective subcategory of $\PMet$, the category of extended pseudometric spaces (\ie, a relaxation of extended metric spaces where different elements $x \neq y$ can be assigned distance $d(x,y) = 0$),
with reflection mapping a pseudometric space $(X,d_X)$ into its quotient modulo the equivalence 
$x \cong y$ iff $d_X(x,y) = 0$. Since $\PMet$ is cocomplete with colimits constructed similarly to $\Set$ also $\Met$ is cocomplete and its colimits are just simple quotiented versions of those in $\Set$.

Although the abstract argument above is enough to prove completeness and cocompleteness of $\Met$, in the proof below we give a direct concrete construction of its limits and colimits.
\begin{prop} \label{prop:emetbicomplete}
$\Met$ is a complete and cocomplete category.
\end{prop}
\begin{proof}

Let $D \colon \mathcal{I} \to \Met$ be a small diagram, and let $D(i) = (X_i, d_i)$, for each object $i \in \mathcal{I}$.
Let $U \colon \Met \to \Set$ be the standard forgetful functor, sending $(X,d_X)$ to $X$. Clearly,
also $UD \colon \mathcal{I} \to \Set$ is a small diagram. 
We show completeness and cocompleteness separately:

\textbf{Completeness:} Let $(f_i \colon L \to X_i)_{i \in \mathcal{I}}$ be the limit cone to $UD$. We define 
$d_L \colon L \times L \to [0, \infty]$ as follows, for arbitrary $x,y \in L$
\begin{equation*}
  d_L(x, y) = \sup_{i \in \mathcal{I}} d_i(f_i(x), f_i(y)) \,,
\end{equation*}
and claim that this is an extended metric%
\footnote{Note that the definition of $d_L$ makes sense since the supremum exists in $[0, \infty]$; 
this would not be true for standard (finite) metrics taking values in $[0, \infty)$.}.

Let $x,y,z \in L$. 
Identity of indiscernible follows by
\begin{align*}
  d_L(x,y) = 0 
  &\iff \sup_{i \in \mathcal{I}} d_i(f_i(x), f_i(y)) = 0 \tag{def.\ $d_L$} \\
  &\iff \forall i \in \mathcal{I}, \,  d_i(f_i(x), f_i(y)) =  0  \tag{$d_i$ positive} \\
  &\iff \forall i \in \mathcal{I}, \, f_i(x) = f_i(y)  \tag{$d_i$ metric} \\
  &\iff  x = y \,,  \tag{$(f_i)_{i \in \mathcal{I}}$ limit cone }
\end{align*}
symmetry by
\begin{align*}
  d_L(x,y)  
  &= \sup_{i \in \mathcal{I}} d_i(f_i(x), f_i(y)) \tag{def.\ $d_L$} \\
  &= \sup_{i \in \mathcal{I}} d_i(f_i(y), f_i(x)) \tag{$d_i$ metric} \\
  &= d_L(y,x) \,,  \tag{def.\ $d_L$}
\end{align*}
and triangular inequality by
\begin{align*}
  d_L(x,z) +  d_L(z,y) 
  &= \sup_{i \in \mathcal{I}} d_i(f_i(x), f_i(z)) + \sup_{i \in \mathcal{I}} d_i(f_i(z), f_i(y)) \tag{def.\ $d_L$} \\
  &\geq \sup_{i \in \mathcal{I}} \big( d_i(f_i(x), f_i(z)) + d_i(f_i(z), f_i(y)) \big) \tag{$\sup$} \\
  &\geq \sup_{i \in \mathcal{I}} d_i(f_i(x), f_i(y)) \tag{$d_i$ metric} \\
  &= d_L(x,y) \,.  \tag{def.\ $d_L$}
\end{align*}
With this metric all $f_i$ are non-expansive functions. Indeed we have, for all $i \in \mathcal{I}$ and $x, y \in L$
\begin{equation*}
 d_i(f_i(x), f_i(y)) \leq \sup_{i \in \mathcal{I}} d_i(f_i(x), f_i(y)) = d_L(x,y) \,.
\end{equation*}
Since the forgetful functor $U \colon \Met \to \Set$ is faithful, the non-expan\-siveness of
the maps $f_i$ implies that $(f_i \colon (L,d_L) \to (X_i, d_i))_{i \in \mathcal{I}}$ is a cone to $D$.
Next we show that this is actually the limiting cone. 

Let $(h_i \colon (H,d_H) \to (X_i,d_i))_{i \in \mathcal{I}}$ be a cone to $D$. Then $(h_i \colon H \to X_i)_{i \in \mathcal{I}}$ is 
a cone to $UD$. Since $(f_i \colon L \to X_i)_{i \in \mathcal{I}}$ is the limit cone 
to $UD$, there exists a unique function $g \colon H \to L$ in $\Set$ satisfying 
$f_i \circ g = h_i$, for all $i \in \mathcal{I}$. We finish our proof by showing that $g$ is a non-expansive function.
By non-expansiveness of the $h_i$'s we have that, for all $i \in \mathcal{I}$ and $a,b \in H$, 
$d_i(h_i(a), h_i(b)) \leq d_H(a,b)$, and thus also
\begin{align*}
d_L(g(a), g(b)) 
&= \sup_{i \in \mathcal{I}} d_i(f_i(g(a)), f_i(g(b))) \tag{def.\ $d_L$} \\
&= \sup_{i \in \mathcal{I}} d_i(h_i(a), h_i(b)) \tag{$f_i \circ g = h_i$} \\
&\leq d_H(a,b) \,. \tag{$h_i$ non-expansive}
\end{align*}
Thus we conclude that 
$(f_i \colon (L,d_L) \to (X_i, d_i))_{i \in \mathcal{I}}$ is a limit cone to $D$.

\textbf{Cocompleteness:} Let $(f_i \colon X_i \to L)_{i \in \mathcal{I}}$ be the colimit cocone to $UD$. We define 
$d_L \colon L \times L \to [0, \infty]$, for arbitrary $x,y \in L$, as follows:
\begin{equation*}
 d_L(x,y) = \sup_{d \in M_L} d(x,y) \,,
\end{equation*}
where $M_L$ is the set of all extended pseudometrics $d$ on $L$ making 
all $f_i$'s non-expansive functions $f_i \colon (X_i,d_i) \to (L,d)$. 
We claim that this is an extended pseudometric.
Since all $d \in M_L$ are pseudometrics, we can derive immediately that $d_L(x,x) = 0$ and
$d_L(x,y) = d_L(y,x)$, for all $x,y \in L$. Moreover, for all $x,y,z \in L$, we have
\begin{align*}
  d_L(x,z) +  d_L(z,y) 
  &= \sup_{d \in M_L} d(x,z) + \sup_{d \in M_L} d(z,y) \tag{def.\ $d_L$} \\
  &\geq \sup_{d \in M_L} d(x,z) + d(z,y) \tag{$\sup$} \\
  &\geq \sup_{d \in M_L} d(x,y) \tag{$d$ metric} \\
  &= d_L(x,y) \,.  \tag{def.\ $d_L$}
\end{align*}
Moreover, for all $i \in \mathcal{I}$ and $x, y \in X_i$
\begin{align*}
 d_i(f_i(x), f_i(y)) 
 &\leq \sup_{d \in M_L} d(x, y) \tag{def.\ $M_L$} \\
 &=  d_L(x, y) \,. \tag{def.\ $d_L$}
\end{align*}
Thus, all the functions $f_i$ are non-expansive w.r.t.\ the pseudometric $d_L$.

Now we turn the extended pseudometric space $(L, d_L)$ into an extended metric space 
$(C, d_C)$ by taking the quotient modulo the equivalence $x \cong y$ iff $d_L(x,y) = 0$. 
The extended metric $d_C \colon C \times C \to [0, \infty]$ is given by  
\begin{equation*}
  d_C([x], [y]) = d_L(x,y) 
\end{equation*}
for all $x,y \in L$, where $[\cdot] \colon L \to C$ denote the quotient map w.r.t.\ $\cong$.
Note that $d_C$ is well defined because by triangular inequality of $d_L$ the definition 
above is independent of the choice of the representative $x$ of the $\cong$-equivalence 
class $[x]$ in $C$.

From the non-expansiveness of the maps $f_i$ we have that also $[\cdot] \circ f_i$ are non-expansive.
Thus, since the forgetful functor $U \colon \Met \to \Set$ is faithful, $([\cdot] \circ f_i \colon (X_i, d_i) \to (C,d_C))_{i \in \mathcal{I}}$ is a cocone to $D$ in $\Met$.
Next, we show that this is the colimiting cocone. 

Let $(h_i \colon (X_i,d_i) \to (H,d_H))_{i \in \mathcal{I}}$ be a cocone to $D$. Then $(h_i \colon X_i \to H)_{i \in \mathcal{I}}$ is 
a cocone to $UD$. Since $(f_i \colon X_i \to L)_{i \in \mathcal{I}}$ is the colimit cocone 
to $UD$, there exists a unique function $g \colon L \to H$ in $\Set$ satisfying 
$g \circ f_i = h_i$, for all $i \in \mathcal{I}$. We prove that $g$ is non-expansive w.r.t.\ the pseudometric $d_L$.
Let $d_g \colon L \times L \to [0,\infty]$ be defined as $d_g(x,y) = d_H(g(x), g(y))$. 
It is easy to see that this is an extended pseudometric on $L$. Moreover, for all $i \in \mathcal{I}$ and $x',y' \in X_i$
we have 
\begin{align*}
 d_g(f_i(x'), f_i(y')) 
 &= d_H(g(f_i(x')), g(f_i(y'))) \tag{def.\ $d_g$} \\
 &= d_H(h_i(x'), h_i(y')) \tag{$g \circ f_i = h_i$} \\
 &\leq d_i(x', y') \,. \tag{$h_i$ non-expansive}
\end{align*}
Thus $d_g \in M_L$. Using this we observe that, for all $x, y \in L$
\begin{equation}
d_L(x,y) = \sup_{d \in M_L} d(x,y) \geq d_g(x,y) = d_H(g(x), g(y)) \,.
\label{eq:gnexp}
\end{equation}

Let $g'([x]) = g(x)$, for all $x \in L$. By \eqref{eq:gnexp} and 
the fact that $d_H$ is a metric, $g(x) = g(y)$ whenever $x \cong y$, hence the $g' \colon C \to H$ 
is a well defined function on $C$. Moreover, by definition of $d_C$ and \eqref{eq:gnexp}, $g'$ is also 
non-expansive as a map $g' \colon (C,d_C) \to (H,d_H)$ in $\Met$.

Clearly, $g' \circ [\cdot] \circ f_i = h_i$, for all $i \in \mathcal{I}$. Assume that there exists 
another map $g'' \colon C \to H$ such that $g'' \circ [\cdot] \circ f_i = h_i$, for all $i \in \mathcal{I}$.
By the universal property of $g$, $g' \circ [\cdot] = g = g'' \circ [\cdot]$, thus $g'([x]) = g''([x])$, for all $x \in L$.

Thus $([\cdot] \circ f_i \colon (X_i, d_i) \to (C,d_C))_{i \in \mathcal{I}}$ is the colimit cocone to $D$.
\qedhere
\end{proof}

The situation is very similar when we consider the full subcategory $\CMet$ of complete extended metric spaces.
Indeed, $\CMet$ is closed under limits, which are defined as in $\Met$. Moreover, as $\CMet$ is a reflective subcategory of $\Met$, with reflection the Cauchy completion functor, we have that also $\CMet$ is cocomplete, with colimits constructed as in $\Met$ and completed via Cauchy completion.

\begin{prop} \label{prop:cmetbicomplete}
$\CMet$ is a complete and cocomplete category.
\qed
\end{prop}

%%%%%%%%%%%%%%%%%%%%%%%%%%%%%%%%%%%
\section{Extended Metric Spaces are Locally Countably Presentable}
\label{sec:EMet-lcp}

Let $\lambda$ be a regular infinite cardinal (i.e., one that is not cofinal to any smaller cardinal).
A small category is called \emph{$\lambda$-filtered} if any subcategory of
less than $\lambda$ morphisms has a cocone in it. 
When $\lambda = \aleph_0$, the term \emph{finitely} filtered (or simply, filtered) is most commonly
used, and \emph{countably} filtered in the case $\lambda = \aleph_1$.

\begin{exa}
Let $\aleph_0$ denote the skeleton of the category of finite sets and all functions between them.
Then $\aleph_0$ is finitely filtered, but not countably filtered. While the skeleton category $\aleph_1$
of all countable sets is countably filtered.
\end{exa}

A diagram is \emph{$\lambda$-filtered} if its domain is $\lambda$-filtered, and 
a colimit is \emph{$\lambda$-filtered} when it is the colimit of a $\lambda$-filtered diagram.
A functor $F \colon \mathcal{C} \to \mathcal{D}$ is called \emph{$\lambda$-accessible} if its domain 
$\mathcal{C}$ has $\lambda$-filtered colimits and $F$ preserves them%
\footnote{In some literature $\aleph_0$-accessible functors are said \emph{of finite rank} and 
and $\aleph_1$-accessible functors \emph{of countable rank}. This is the terminology preferred 
by John Power in his seminal work about enriched Lavwere theories.}. 

An object $X$ of a small category $\mathcal{C}$ is \emph{$\lambda$-presentable} if its hom-functor
\begin{equation*} 
\mathcal{C}(X, -) \colon \mathcal{C} \to \Set
\end{equation*}
is $\lambda$-accessible.
Explicitly, $X$ is $\lambda$-presentable iff for each $\lambda$-filtered colimit cocone 
$(c_i \colon D(i) \to C)_{i \in \mathcal{I}}$ of a $\lambda$-filtered diagram $D \colon \mathcal{I} \to \mathcal{C}$, and each morphism $f \colon X \to C$, there exists $i \in \mathcal{I}$ such that
\begin{itemize}
\item $f$ factorizes through $c_i$, \ie, $f = c_i \circ g$ for some $g \colon X \to D(i)$, and
\item the factorization is essentially unique in the sense that if $f = c_i \circ g = c_i \circ g'$, 
then $D(i \to j) \circ g = D(i \to j) \circ g'$, for some $j \in \mathcal{I}$.
\end{itemize}

\begin{defi}[Accessibility and Local Presentability]
A category $\mathcal{C}$ is \emph{$\lambda$-accessible} if
\begin{itemize} 
\item it has all $\lambda$-filtered colimits;
\item there is a set $\mathcal{C}_\lambda$ of $\lambda$-presentable objects such that every object is a 
$\lambda$-filtered colimit of objects of $\mathcal{C}_\lambda$.
\end{itemize}
It is \emph{locally $\lambda$-presentable} if, moreover, it has all small colimits (\ie, it is cocomplete).
\end{defi} 
A category is said \emph{accessible} (resp. \emph{locally presentable}) if it is $\lambda$-accessible 
(resp. $\lambda$-locally presentable) for some regular infinite cardinal $\lambda$.
In the case $\lambda = \aleph_0$, we speak about \emph{locally finitely presentable category},
and for $\lambda = \aleph_1$ about \emph{locally countably presentable category}.

\begin{exa}
The category $\Set$ is locally finitely presentable with finitely presentable objects precisely the finite
sets (for $\Set_{\aleph_0}$ we can choose the set of all natural numbers).
The category $\omega\mathbf{CPO}$ of cpo's (\ie, posets with joints of all increasing 
$\omega$-chains) and $\omega$-continuous functions is not locally finitely presentable, however, 
it is locally countably presentable with countably presentable objects precisely the countable cpo's
(for $\omega\mathbf{CPO}_{\aleph_1}$ we can choose the set of all countable ordinals with 
standard partial order).
\end{exa}

Next, we focus our attention on the category $\Met$. In turn, we prove that 
\begin{enumerate}
\item the only finitely presentable objects in $\Met$ are the finite discrete spaces, 
with distances either $0$ or $\infty$ (Proposition~\ref{prop:discretefinite}); 
\item $\Met$ is locally countably presentable, with countably presentable objects precisely
the countable spaces (Lemma~\ref{lemma:EMetcountablepresentability} and Theorem~\ref{th:locallycountablyrepresentability})
\end{enumerate}

Note that a direct consequence of (1) is that $\Met$ is \emph{not} locally finitely presentable, since
filtered colimits of discrete spaces are discrete.

The proofs of these results are immediate adaptations of~\cite{AdamekMM12} which shows
that the category of 1-bounded pseudometric spaces with non-expansive maps is locally countably 
presentable.

\begin{prop} \label{prop:discretefinite}
Finitely presentable objects in $\Met$ are finite and discrete.
\end{prop}
\begin{proof}
Note that every extended metric space is a colimit of the filtered diagram obtained by taking all its
finite subspaces and their inclusion maps.
Let $(X,d_X)$ be a finitely presentable object in $\Met$. Then the identity must 
factorize through the inclusion of one of the finite subspaces. Thus, $X$ must be finite.

%By contradiction, assume that $d_X$ is not the discrete metric, hence there exists a pair
%$x,y \in X$ of points such that $0 < d_X(x,y) < \infty$.
For each positive integer $n > 0$, define the function $d_n \colon X \times X \to [0,\infty]$ as 
\begin{equation*}
d_n(x,y) = \left( 1+ \frac{1}{n} \right) \cdot d_X(x,y) \,,
\end{equation*}
where $\infty \cdot r = \infty$ for any $r \in [0, \infty)$. Clearly, all $d_n$'s are extended metrics.
Consider the $\omega$-chain of spaces $(X, d_n)$ with identities as connecting maps. 
This is a countably filtered diagram with colimit cocone $(id_X \colon (X, d_n) \to (X,d_X))_{n > 0}$.
Since $(X,d_X)$ is finitely presentable, the identity $id_X \colon (X,d_X) \to (X,d_X)$ 
must factorize through the a colimit map $id_X \colon (X, d_n) \to (X,d_X)$ for some positive 
integer $n > 0$.
But the distances in $(X,d_X)$ which are strictly between $0$ and $\infty$ are increased by $d_n$.
So $(X,d_X)$ must discrete.
\end{proof}

\begin{prop} \label{prop:filteredset}
Let $(f_i \colon D(i) \to (C, d_C))_{i \in \mathcal{I}}$ be the colimit cocone to a small diagram 
$D \colon \mathcal{I} \to \Met$ and $D(i) = (X_i, d_i)$, for each $i \in \mathcal{I}$. 
If $D$ is filtered,  $C$ is the quotient set of 
$\coprod_{i \in \mathcal{I}} X_i$ under the equivalence
\begin{equation*}
  in_i(x) \sim in_{i'}(x')  
  \quad \text{iff there exists $j \in \mathcal{I}$ with $D(i \to j)(x) = D(i' \to j)(x')$} \,,
\end{equation*}
where $in_i \colon X_i \to \coprod_{i \in I} X_i$ are the canonical injections into the coproduct
and $f_i$ assigns to each $x \in X_i$ the equivalence class of $in_i(x)$, for all $i \in \mathcal{I}$.
\end{prop}
\begin{proof}
We prove that $\sim$ is an equivalence relation. Reflexivity and symmetry are trivially satisfied. 
Transitivity follows by the fact that $\mathcal{I}$ is filtered. Indeed, let $in_i(x) \sim in_{i'}(x')$ and
$in_{i'}(x') \sim in_{i''}(x'')$. Then there exist $j, j' \in \mathcal{I}$ such that 
$D(i \to j)(x) = D(i' \to j)(x')$ and $D(i' \to j')(x') = D(i'' \to j')(x'')$. Since $\mathcal{I}$ is filtered, there
exists $j'' \in \mathcal{I}$ above $i, i', i'', j$, and $j'$ such that
\begin{gather*}
D(i \to j'') = D(j \to j'') \circ D(i \to j) \,, \\
D(i' \to j'') = D(j \to j'') \circ D(i' \to j) = D(j' \to j'') \circ D(i' \to j') \,, \text{ and} \\
D(i'' \to j'') = D(j' \to j'') \circ D(i'' \to j') \,.
\end{gather*}
From this we derive $D(i \to j'')(x) = D(i'' \to j'')(x'')$, i.e., $in_{i}(x) \sim in_{i''}(x'')$.

Let $U \colon \Met \to \Set$ denote the forgetful functor into $\Set$.
By the construction of colimits in $\Met$ (see Proposition~\ref{prop:emetbicomplete}) it suffices
to prove that $(f_i \colon X_i \to (\coprod_{i \in I} X_i)/_{\sim})_{i \in \mathcal{I}}$ is a colimit cocone 
to $UD$ in $\Set$.

For convenience, let $X = (\coprod_{i \in I} X_i)/_{\sim}$ and, for each $i \in \mathcal{I}$, let $[in_i(x)]$ denote the equivalence class of $in_i(x)$.
We first prove that $(f_i \colon X_i \to X)_{i \in \mathcal{I}}$ is a cocone to $UD$.
Let $i \to j \in \mathcal{I}$. Then from $D(i \to j)(x) = D(id_j)(D(i \to j)(x))$ for all $x \in X_i$,
we have $in_i(x) \sim in_j(D(i \to j)(x))$. Therefore, for all $x \in X_i$
\begin{equation*}
f_i(x) 
= [in_i(x)] 
= [in_j(x)(D(i \to j)(x))]
= f_j (D(i \to j)(x)) \,.
\end{equation*}
Hence $f_i = f_j \circ D(i \to j)$ for all $i, j \in \mathcal{I}$.
Now we prove that $(f_i \colon X_i \to X)_{i \in \mathcal{I}}$ is a colimit. 
Let $(h_i \colon X_i \to H)_{i \in \mathcal{I}}$ be a cocone to $UD$. Define  
$g \colon X \to H$ for arbitrary $x \in X_i$ as follows:
\begin{equation*}
  g([in_i(x)]) = \gamma (in_i(x)) \,,
\end{equation*}
where $\gamma \colon \coprod_{i \in I} X_i \to H$ is the unique map such 
that $\gamma \circ in_i = h_i$, for all $i \in \mathcal{I}$.
Note that $g$ is well defined. Indeed, if $D(i \to j)(x) = D(i' \to j)(x')$ for some $j \in \mathcal{I}$, we have that
\begin{align*}
g([in_i(x)]) 
&= \gamma (in_i(x)) \tag{def.\ $g$} \\
&= h_i(x) \tag{$\gamma \circ in_i = h_i$} \\
&= h_j (D(i \to j)(x)) \tag{$(h_i \colon X_i \to H)_{i}$ cocone to $UD$} \\
&= h_j (D(i' \to j)(x')) \\
&= h_{i'}(x') \tag{$(h_i \colon X_i \to H)_{i}$ cocone to $UD$} \\
&= \gamma (in_{i'}(x')) \tag{$\gamma \circ in_{i'} = h_{i'}$} \\
&= g([in_{i'}(x')]) \,.  \tag{def.\ $g$}
\end{align*}
Let $i \in \mathcal{I}$. Then, for all $x \in X_i$
\begin{align*}
g \circ f_i (x) 
&= g ([in_i(x)]) \tag{def.\ $f_i$} \\
&= \gamma (in_i(x)) \tag{def.\ $g$} \\
&= h_i(x) \,. \tag{$\gamma \circ in_i = h_i$}
\end{align*}
Thus $g \circ f_i  = h_i$ for all $i \in \mathcal{I}$. Assume now that there exists $g' \colon X \to H$ such that
$g' \circ f_i  = h_i$ for all $i \in \mathcal{I}$. Then, for all $x \in X_i$
\begin{align*}
g'([in_i(x)]) 
&= g'(f_i (x))  \tag{def.\ $f_i$} \\
&= h_i(x) \tag{$g' \circ f_i  = h_i$} \\
&= g(f_i(x)) \tag{$g \circ f_i  = h_i$} \\
&= g([in_i(x)]) \,.   \tag{def.\ $f_i$}
\end{align*}
Thus $g' = g$.
\end{proof}

%%%%%%%%%%%%%%%%%%%%%%%%%

\begin{prop} \label{prop:filteredmetric}
Let $(f_i \colon D(i) \to (C, d_C))_{i \in \mathcal{I}}$ be the colimit cocone to a small diagram 
$D \colon \mathcal{I} \to \Met$ and $D(i) = (X_i, d_i)$, for each $i \in \mathcal{I}$. If $D$ is filtered, then, for all $x,y \in C$
\begin{equation*}
  d_C(x,y) = \inf \{ d_i(x', y') \mid i \in \mathcal{I}\,, f_i(x') = x \,, \text{and } f_i(y') = y \} \,.
\end{equation*}
\end{prop}
\begin{proof}
Assume $D$ is filtered.
As shown in Proposition~\ref{prop:emetbicomplete}, 
\begin{equation}
 d_C(x,y) = \sup_{d \in M_C} d(x,y) \,,
 \label{eq:colimitsupmetric}
\end{equation}
where $M_C$ is the set of all extended pseudometrics $d'$ on $C$ making 
all the functions $f_i \colon (X_i, d_i) \to (C,d')$ non-expansive. 
Let $d \colon C \times C \to [0,\infty]$ be
\begin{equation}
  d(x,y) = \inf \{ d_i(x', y') \mid i \in \mathcal{I}\,, f_i(x') = x \,, \text{and } f_i(y') = y \} \,.
  \label{eq:infdistance}
\end{equation}
Since $D$ is filtered, then also $UD$ is so, where $U \colon \Met \to \Set$ is the obvious forgetful functor 
to $\Set$. Since $\Set$ is locally finitely representable, then for any finite subset $\{x_1, \dots, x_n\} \subseteq C$, there exist $i \in \mathcal{I}$ and $\{x'_1, \dots, x'_n\} \subseteq X_i$ such that $f_i(x'_j) = x_j$, for all $0 \leq j \leq n$.
In particular, this implies that for any $x,y \in C$, the infimum in~\eqref{eq:infdistance} never ranges over an empty set. 

Let $x,y \in C$. We prove $d_C(x,y) \leq d(x,y)$ and $d_C(x,y) \geq d(x,y)$ separately.
\begin{itemize}

\item 
By non-expansivity of the maps $f_i \colon (X_i, d_i) \to (C, d_C)$, for any $i \in \mathcal{I}$
such that $f_i(x') = x$ and $f_i(y') = y$ for some $x', y' \in X_i$, we have
\begin{equation*}
  d_i(x',y') \geq d_C(f_i(x'),f_i(y')) = d_C(x,y) \,.
\end{equation*}
Thus $d_C(x,y) \leq \inf \{ d_i(x', y') \mid i \in \mathcal{I}\,, f_i(x') = x \,, \text{and } f_i(y') = y \}$. 
By~\eqref{eq:infdistance}, $d_C(x,y) \leq d(x,y)$.

\item
We prove $d \in M_L$. We start by showing that $d$ is a pseudometric on $C$.
Since all $d_i$ are extended metrics, we immediately derive that $d(x,x) = 0$ and $d(x,y) = d(y,x)$ 
for all $x,y \in C$. Moreover, for all $x,y,z \in C$:
\begin{align*}
 d(x,y) 
 &= \inf \{ d_i(x', y') \mid i \in \mathcal{I}\,, f_i(x') = x \,, \text{and } f_i(y') = y \}  \\
 &\leq \inf \{ d_i(x', z') + d_i(z',y') \mid i \in \mathcal{I}\,, f_i(x') = x \,, f_i(y') = y \,, f_i(z') = z \}  \\
 &\leq d(x,z) + d(z,y) \,,
\end{align*}
where the last inequality follows by the fact that $D$ is filtered, hence
for all $j,k \in \mathcal{I}$ there exists $i \in \mathcal{I}$ and non-expansive maps $D(j \to i)$, $D(k \to i)$ such that $f_j = f_i \circ D(j \to i)$ and $f_k = f_i \circ D(k \to i)$.
Therefore, $d$ is a pseudometric.
The non-expansiveness of the maps $f_j \colon (X_j, d_j) \to (C,d)$, for all $j \in \mathcal{I}$ 
follows immediately by~\eqref{eq:infdistance}:
\begin{equation*}
  d(f_j(x'), f_j(y')) = \inf \{ d_i(x', y') \mid i \in \mathcal{I} \text{ and } x',y' \in X_j \} \leq d_j(x', y') \,.
\end{equation*}

Thus, $d \in M_L$. Therefore, by \eqref{eq:colimitsupmetric} we have $d_C(x,y) \geq d(x,y)$.
\qedhere
\end{itemize}
\end{proof}

\begin{lem} \label{lemma:EMetcountablepresentability}
$(X,d_X) \in \Met$ is countably presentable iff it is countable.
\end{lem}
\begin{proof}
Every extended metric space $(X,d_X)$ is a countably filtered colimit of its countable subspaces. 
If $(X,d_X)$ is countably presentable, then the identity $id_X$ must factorize through the inclusion of 
one of the countable subspaces of $(X,d_X)$. Thus, $(X,d_X)$ is countable.

Conversely, let $(X,d_X)$ be a countable space,
and let $(f_i \colon D(i) \to (C, d_C))_{i \in \mathcal{I}}$ be the colimit cocone to a countably 
filtered diagram $D \colon \mathcal{I} \to \Met$.
Any morphism $h \colon (X,d_X) \to (C,d_C)$ factorizes through the image $h(X) \subseteq (C,d_C)$.
Note that $h(X)$ is countable space because $X$ is so. 

For each $i \in \mathcal{I}$, let $D(i) = (X_i, d_i)$. Since $D$ is filtered, by Proposition~\ref{prop:filteredmetric},
for any $x,y \in h(X)$,
\begin{equation}
  d_C(x,y) = \inf \{ d_i(x', y') \mid i \in \mathcal{I}\,, f_i(x') = x \,, \text{and } f_i(y') = y \} \,.
  \label{eq:infmetric}
\end{equation}
Thus, for any $n \in \mathbb{N}$ there exist
$j_{n} \in \mathcal{I}$ and $x_{n}, y_{n} \in X_{j_{n}}$ such that
\begin{gather}
  f_{j_{n}}(x_{n}) = x \,, \;  f_{j_{n}}(y_{n}) = y \,,  \text{ and } \, 
  %\label{eq:approximantsmapsto} \\
  d_{j_{n}}(x_{n}, y_{n}) \leq d_C(x,y) + \frac{1}{n+1} \,.
  \label{eq:approximants}
\end{gather}
Since $D$ is countably filtered, by \eqref{eq:approximants} and Proposition~\ref{prop:filteredset}, 
there exist $j_{x,y} \in \mathcal{I}$ and connecting morphisms $D(j_n \to j_{x,y}) \colon X_{j_{n}} \to X_{j_{x,y}}$, 
mapping all $x_{n}$ to a single element in $\bar{x}_y \in X_{j_{x,y}}$, %and all $y_{n}$ to a single element $y'$. 
i.e., for all $n \in \mathbb{N}$
\begin{equation}
  D(j_n \to j_{x,y})(x_n) = \bar{x}_y 
  \quad \text{and} \quad f_{j_{x,y}}(\bar{x}_y) = x
  %\quad \text{and} \quad D(j_n \to j_{x,y})(y_n) = y' 
  \,.
  \label{eq:commonpair}
\end{equation}
From \eqref{eq:commonpair} and the fact that $h(X)$ is countable, by Proposition~\ref{prop:filteredset}, 
there exists $j \in \mathcal{I}$ and connecting morphisms such that, for all $x,y \in h(X)$.
\begin{equation}
  D(j_{x,y} \to j)(\bar{x}_y) = \bar{x} 
   \quad \text{and} \quad f_{j}(\bar{x}) = x \,.
  \label{eq:finalelement}
\end{equation}

From the construction above, we can define the map $g \colon X \to X_j$ as follows:
\begin{equation*}
  g(x) = \overline{h(x)} \,,
\end{equation*}
where $\overline{h(x)}$ is the element in $X_j$ satisfying \eqref{eq:finalelement}.
Next we prove that $g$ is non-expansive. Assume by contradiction that, $d_X(x,y) < d_j(g(x),g(y))$ for some
$x,y \in X$.
Then there exists $n \in \mathbb{N}$ such that $d_X(x,y) + \frac{1}{n+1} \leq d_j(g(x),g(y))$,
and by non-expansivity of $h$
\begin{equation}
  d_C(h(x),h(y)) + \frac{1}{n+1} \leq d_j(g(x),g(y)) = d_j( \overline{h(x)} , \overline{h(y)}) \,.
  \label{eq:contraddiction}
\end{equation}
But, by \eqref{eq:finalelement} $f_j(\overline{h(x)}) = h(x)$ and $f_j(\overline{h(y)}) = h(y)$. Thus  \eqref{eq:contraddiction} contradicts \eqref{eq:infmetric}. Consequently, for all $x,y \in X$, $d_X(x,y) \leq d_j(g(x),g(y))$, i.e., $g$ is non-expansive.

By definition of $g$ and \eqref{eq:finalelement} we immediately obtain that $h = f_j \circ g$. 
Hence $h$ factorizes through $f_j$. By Proposition~\ref{prop:filteredset} the factorization is essentially unique.
Indeed, whenever $y,y' \in X_i$ fulfil $f_i(y) = f_i(y')$, then there exists $j \in \mathcal{I}$ such that 
$D(i \to j)(y) = D(i' \to j)(y')$.
\end{proof}

\begin{thm} \label{th:locallycountablyrepresentability}
$\Met$ is a locally countably presentable category.
\end{thm}
\begin{proof}
By Proposition~\ref{prop:emetbicomplete}, $\Met$ is cocomplete. 
Every extended metric space is the countably filtered colimit of its countable subspaces. 
By Lemma~\ref{lemma:EMetcountablepresentability} the countable spaces are countably presentable. 
Hence for $\Met_{\aleph_1}$ one can take the set of objects of the skeleton category of the full 
subcategory of all countable extended metric spaces, or equivalently, the set of all countable ordinals 
endowed with an extended metric space.
\end{proof}

\end{document}